\documentclass[11pt]{article}
\usepackage[margin=1in]{geometry}
\usepackage[utf8]{inputenc}
\usepackage[english]{babel}
\usepackage{caption}
\usepackage{subcaption}
\usepackage{array} 

\usepackage[dvipsnames]{xcolor}

\usepackage{amsmath, amssymb, amsthm,mathtools,dsfont}
\usepackage{bbm}
\PassOptionsToPackage{numbers,square,sort&compress}{natbib}
\usepackage{natbib}
\usepackage{blkarray}
\usepackage{cancel}
\usepackage{enumitem}
\usepackage{euscript}
\usepackage{float}
\usepackage{bm}
\usepackage{graphicx}
\usepackage{mathrsfs}
\usepackage{microtype}
\usepackage{ragged2e}
\usepackage{sectsty}
\usepackage{thmtools}
\usepackage{tikz}
\usepackage[Symbolsmallscale]{upgreek}

\usepackage{microtype}
\usepackage{graphicx}
\usepackage[pdfborder={0 0 0}]{hyperref}
\hypersetup{
  urlcolor = black,
  pdfauthor = {},
  pdfkeywords = {},
  pdftitle = {},
  pdfsubject = {},
  pdfpagemode = UseNone
}
\definecolor{darkgreen}{rgb}{0,0.5,0} %
\definecolor{darkblue}{rgb}{1,0,0} 
\definecolor{darkblue}{rgb}{1,0,0} 

\input{mathsymbols}

\usepackage{fancyhdr}

\theoremstyle{plain}

\begin{document}
\title{Bid--Ask Martingale Optimal Transport}

\newcommand\blfootnote[1]{%
  \begingroup
  \renewcommand\thefootnote{}\footnote{#1}%
  \addtocounter{footnote}{-1}%
  \endgroup
}

\author{
Bryan Liang%
 \thanks{
  Quantitative Research, Office of the CTO, Bloomberg. Email: bliang17@bloomberg.net. } \hspace{1em}
 Marcel Nutz%
  \thanks{
  Columbia University, Departments of Statistics and Mathematics.  Research supported by NSF Grants DMS-2106056, DMS-2407074. Email: mnutz@columbia.edu.} \hspace{1em}
  Shunan Sheng%
  \thanks{
  Columbia University, Department of Statistics. Email: ss6574@columbia.edu. Research initiated during an internship at Bloomberg.} \hspace{1em}
Valentin Tissot-Daguette%
  \thanks{
  Quantitative Research, Office of the CTO, Bloomberg. Email: vtissotdague@bloomberg.net}  
}
\date{\today}
\maketitle

\begin{abstract}
Martingale Optimal Transport (MOT) provides a framework for robust pricing and hedging of illiquid derivatives. Classical MOT enforces exact calibration of model marginals to the mid-prices of vanilla options. Motivated by the industry practice of fitting bid and ask marginals to vanilla prices, we introduce a relaxation of MOT in which model-implied volatilities are only required to lie within observed bid--ask spreads; equivalently, model marginals lie between the bid and ask marginals in convex order. The resulting Bid--Ask MOT (BAMOT) yields realistic price bounds for illiquid derivatives and, via strong duality, can be interpreted as the superhedging price when short and long positions in vanilla options are priced at the bid and ask, respectively. We further establish convergence of BAMOT to classical MOT as bid--ask spreads vanish, and quantify the convergence rate using a novel distance intrinsically linked to bid--ask spreads. Finally, we support our findings with several synthetic and real-data examples.
\end{abstract}

\vspace{.9em}

{\small
\noindent\textbf{Acknowledgments} We would like to thank  Bruno Dupire, Martin Forde,  Sergey Nadtochiy, Walter Schachermayer, Mete Soner, and %
 Josef Teichmann %
 for helpful discussions on this topic. Shunan Sheng is supported by a Bloomberg Quantitative Finance Ph.D. Fellowship.
}

\section{Introduction}

Martingale Optimal Transport (MOT) \cite{beiglbock2013model,galichonTouziHL,HLBook, BeiglbockJuillet2016MOT, Nutz2017Duality} offers a powerful framework to price and hedge illiquid derivatives. Consider a claim %
$H$ contingent on the evolution of a liquidly traded asset $(X_t)_{0\leq t\leq T}$ up to some fixed maturity $T>0$. 
Suppose also that the static hedging instruments consist of all $T$-vanilla options (i.e., all put and call options with maturity $T$ and any strike $K\geq0$), in addition to dynamic trading in~$(X_t)_{0\leq t\leq T}$. 
Absent market frictions, one classically extracts a unique risk-neutral distribution $\mu$ consistent with the Black--Scholes implied volatility (IV) skew $\sigma(K,T)$. Indeed, assuming zero carry, we have by \cite{breeden1978prices} that 
\begin{equation} \label{eq:BL}
    \frac{\mu(\rd K)}{\dd K} = \partial_{K}^2c(K,T), \quad c(K,T) := c_{\text{BS}}(K,T,\sigma(K,T)),
\end{equation}
where $c_{\text{BS}}(K,T,\sigma)$ is the price of the $(K,T)$ call option in the Black--Scholes model with  volatility~$\sigma$. MOT then considers the set of all arbitrage-free models (martingale measures) for $(X_t)_{0\leq t\leq T}$ such that the marginal law $\mu^{\Q}_T$ of $X_T$ under $\Q$ coincides with $\mu$, leading to the  primal (or measure) problem  
\begin{equation}\label{eq:primalMOT}
    \sup_{\Q\in \cQ(\mu)} \E_{\Q}[H], \qquad  \cQ(\mu) = \big\{\Q \mid \text{martingale measure, } \mu^{\Q}_T = \mu \big\}.
\end{equation}
A similar logic applies when the static hedging instruments are available at multiple maturities $T_1, \dots ,T_N$, leading to marginal constraints $\mu^{\Q}_{T_i} = \mu_i$ for $i=1,\dots,N$. To wit, MOT requires admissible models to exactly match the market’s implied volatility (IV) skews, thus imposing hard constraints on the marginal distributions. The dual of MOT, discussed in more detail below, is a semi-static superhedging problem,  which assumes that vanilla options can be bought and sold at the prices implied by those IV skews, with no bid--ask friction.

\subsection{Bid--Ask Friction}
The present work develops an extension of MOT that accounts for the presence of bid--ask spreads for vanilla prices in real markets. Indeed, in practice, vanilla options are often represented in terms of  bid and ask IV skews $\sigma^{\rm b}(K,T) \le  \sigma^{\rm a}(K,T)$. 
Then a martingale measure $\Q$ is consistent with the options market if and only if the implied volatilities $\sigma^{\Q}(K,T)$ that it generates (see \cref{fig:IVIntro}) satisfy 
\begin{equation}\label{eq:IVSkews-1}
    \sigma^{\rm b}(K,T) \le \sigma^{\Q}(K,T) \le  \sigma^{\rm a}(K,T), \quad K\ge 0.
\end{equation}
Of course, order books only contain quotes for finitely many strikes at any given moment. Data vendors provide  bid/ask IV skews by  interpolating/extrapolating the available quotes. As detailed in \cref{sec:bid-ask-marginals}, a common approach for this task is to fit \emph{bid and ask marginal  distributions} $\mu^{\rm b}, \mu^{\rm a}$ within a parametric family to the available quotes. This motivates the primal formulation of our problem, which replaces the single (mid-price) marginal~$\mu$ of the MOT problem~\eqref{eq:primalMOT} by a set of possible marginals determined via $\mu^{\rm b}$ and $\mu^{\rm a}$ as follows. 
\begin{figure}[!ht]
    \centering
    \caption{Bid, model, and ask implied volatility skews for S\&P~500 Index (SPX) options expiring on 03-21-2025, as of 02-27-2025.}
    \includegraphics[width=0.55\linewidth]{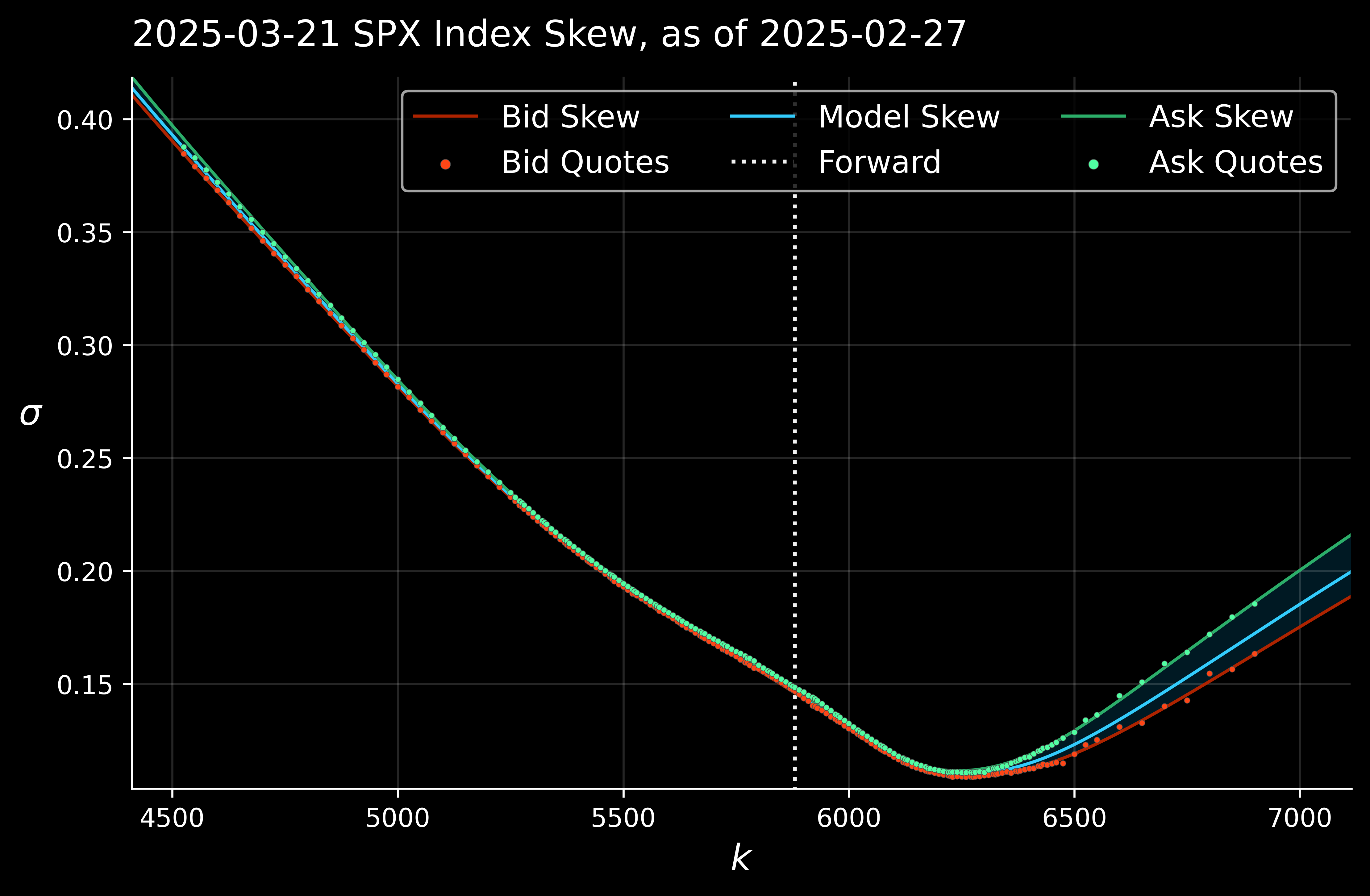}
    \label{fig:IVIntro}
\end{figure}

\subsection{Primal Formulation}\label{sec:primal-intro}

To be consistent with the bid and ask IVs, the possible marginals $\mu_T^{\Q} = \text{Law}^{\Q}(X_T)$ for the asset price $X_T$ must satisfy certain convex order constraints. Indeed, \eqref{eq:IVSkews-1} implies
\begin{equation*}
    \E_{\mu^{\rm b}}[(X-K)^+] = c_{\text{BS}}(K,T,  \sigma^{\rm b}(K,T))  \le  c_{\text{BS}}(K,T, \sigma^{\Q}(K,T)) = \E_{\mu^{\Q}_T}[(X-K)^+] %
\end{equation*}
for all $K\ge 0$ since $c_{\text{BS}}(K,T,\sigma)$ is nondecreasing in $\sigma$. This means that $\mu^{\rm b}$ and $\mu_T^{\Q}$ are in \emph{convex order,} denoted $\mu^{\rm b}\cleq\mu_T^{\Q}$, by the property of convex order recalled in~\eqref{eq:callCVX} below. Similarly, $\mu_T^{\Q} \cleq \mu^{\rm a}$.
This leads us to formulate the following extension of the MOT problem \eqref{eq:primalMOT}, which we call the \emph{Bid--Ask Martingale Optimal Transport} (BAMOT) problem:
\begin{equation}\label{eq:BAMOTIntro}
      \sup_{\Q\in \cQ(\mu^{\rm b}, \mu^{\rm a})} \E_{\Q}[H], \quad   \cQ(\mu^{\rm b}, \mu^{\rm a})  = \big\{\Q \mid \text{martingale measure, } \mu^{\rm b}\cleq \mu_T^{\Q} \cleq \mu^{\rm a}\big\}.
\end{equation}
More generally, we will consider a version of this problem with multiple maturities $T_1, \dots, T_N$. In contrast to its  classical counterpart, BAMOT captures model uncertainty about the underlying asset $(X_t)_{0\leq t\leq T}$ arising not only from its price dynamics, but also from its marginal distributions: while MOT prescribes the marginal~$\mu$ exactly, the bid and ask marginals merely imply inequality constraints for the marginal distribution of the asset. In particular, BAMOT is nontrivial even for claims that depend only on~$X_T$, such as a digital option, whereas MOT simply prices them by taking expectation under~$\mu$. In general, the value of \eqref{eq:BAMOTIntro} is no less than the value of \eqref{eq:primalMOT} (for any $\mu^{\rm b}\cleq \mu \cleq \mu^{\rm a}$), reflecting that the seller of the claim faces additional model uncertainty. %

\subsection{Dual Formulation}\label{sec:DualIntro} 

Next, we consider the effect of bid--ask spreads from the hedging perspective. 
In the classical MOT setting with a single marginal $\mu$, the dual (or portfolio) problem %
aims at minimizing the superhedging cost of the claim $H$, given by
\begin{equation*}
   \inf \{ \mu(\psi) \, \mid \, \psi \in \Psi(H) \}, \quad  %
    \Psi(H) = \Big\{ {\psi} \in L^1(\mu) \,\mid \, \exists \ {\Delta} \text{ such that  }  \textnormal{P\&L}_{\psi,\Delta}^H\ge 0  \Big\}.
\end{equation*}
Here $\Delta$ is a dynamic hedging strategy trading in the asset~$(X_t)_{0\leq t\leq T}$ and 
$$\textnormal{P\&L}_{\psi,\Delta}^H(X) = \psi(X_T) + (\Delta \bullet X)_T - H$$
is the profit and loss resulting from three terms: the static hedge $\psi$, which is constructed from vanilla options at cost $\mu(\psi):=\int \psi \dd\mu$, dynamic (self-financing) trading according to $\Delta$ leading to the stochastic integral $(\Delta \bullet X)_T$, and selling the claim $H$. For the static hedge $\psi$, it is sufficient (cf.~\cite{HLBook}) to look for profiles of the form 
\begin{equation}\label{eq:cashPlusCalls}
\psi(x) = \gamma^{\$} + \int_{\R_+}(x-K)^{+}\lambda(\rd K)
\end{equation}
where $\gamma^{\$}$ is a cash amount and $\lambda$ is a signed measure describing the weights of a call options portfolio (as eventual positions in the forward contract can be absorbed by the delta hedge). Granted that Fubini's theorem applies, 
the cost of $\psi$ reads 
$\mu(\psi) = \gamma^{\$} + \int_{\R_+}c(K,T)\lambda(\rd K)$, where $c(K,T) = \E_{\mu}[(X-K)^+]$.

In our problem formulation, we keep the assumption that trading in~$(X_t)_{0\leq t\leq T}$ is frictionless but highlight 
the static profile $\psi$, which is now exposed to the bid--ask frictions in the options market. We naturally assume that static hedging is limited to put and call options (and the forward). In fact, puts would be redundant given calls and the forward, hence we may suppose again that $\psi$ consists of a cash amount and a portfolio of call options as in~\eqref{eq:cashPlusCalls}.
In the presence of bid--ask spreads, 
we need to separate the contracts that were bought from those that were sold. To this end, consider the Jordan decomposition $\lambda = \lambda^{+} - \lambda^{-}$ of the portfolio weights into nonnegative measures $\lambda^{\pm}$. We can then decompose the static profile as
$$\psi = \psi^{\rm a} - \psi^{\rm b}, \quad \psi^{\rm a}(x) = \gamma^{\$} +  \int_{\R_+}(x-K)^{+}\lambda^+(\rd K), \quad \psi^{\rm b}(x) = \int_{\R_+}(x-K)^{+}\lambda^{-}(\rd K).$$
Since $\lambda^{\pm}$ are nonnegative, both $\psi^{\rm a}, \psi^{\rm b}$ are \emph{convex} profiles,  the first being bought at the ask and the second sold at the bid. 
Writing $c^{\rm b}(K,T)$ and $c^{\rm a}(K,T)$ for the bid and ask call option prices, the cost of $\psi$ becomes
$$ \gamma^{\$} + \int_{\R_+}c^{\rm a}(K,T)\lambda^{+}(\rd K) - \int_{\R_+}c^{\rm b}(K,T)\lambda^{-}(\rd K). $$
\Cref{fig:bidaskHeatmap} displays the bid--ask spread for put and call options on the S\&P~500 Index (SPX) across strikes and maturities as of the close of 02-07-2025. As can be seen, the spreads are non-negligible, even for liquid underlyings such as SPX. With bid and ask marginals $\mu^{\rm b} \cleq\mu^{\rm a}$ as in \cref{sec:primal-intro}, we can write the cost of the static profile succinctly as $\mu^{\rm a} (\psi^{\rm a})- \mu^{\rm b} (\psi^{\rm b})$. 
 \begin{figure}[!tbh]
    \centering
    \caption{Bid--ask spread ($\$$) of S\&P~500 Index (SPX) options as of 02-07-2025. }
    \includegraphics[width=0.65\linewidth]{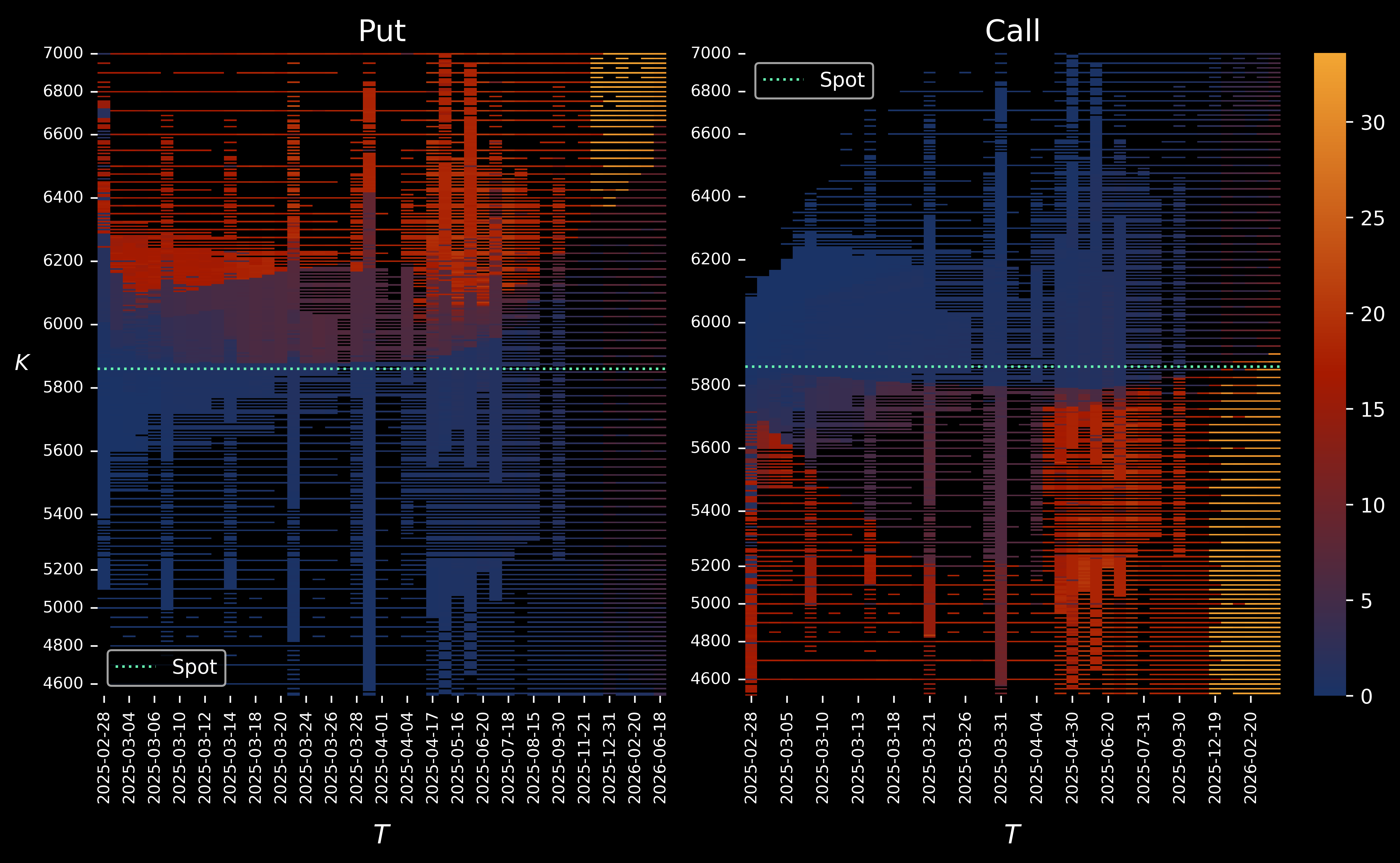}
    
    \label{fig:bidaskHeatmap}
\end{figure}

More generally, we formulate our dual problem over \emph{pairs} $\psi^{\rm b,a}=(\psi^{\rm a},\psi^{\rm b})$ of convex functions, with $\psi^{\rm a}$ corresponding to call options bought (plus a cash position) and $\psi^{\rm b}$ corresponding to call options sold.\footnote{We do not insist that $\psi^{\rm a}-\psi^{\rm b}$ be a Jordan decomposition. However, a Jordan decomposition has the minimal cost $\mu^{\rm a} (\psi^{\rm a})- \mu^{\rm b} (\psi^{\rm b})$ among all decompositions into convex functions.} %
Writing 
\begin{equation*}%
    \text{cost}(\psi^{\rm b,a}) = \mu^{\rm a} (\psi^{\rm a})- \mu^{\rm b} (\psi^{\rm b}) \quad\text{for}\quad \psi^{\rm b,a}=(\psi^{\rm a},\psi^{\rm b}),
\end{equation*}
we arrive at the dual formulation 
\begin{equation}\label{eq:dual-intro}
    \inf_{\psi^{\rm b,a} \in \Psi^{\rm b,a}(H)}  \text{cost}(\psi^{\rm b,a}), 
\end{equation}
where
$$\Psi^{\rm b,a}(H) = \big\{ \psi^{\rm b,a}=(\psi^{\rm a},\psi^{\rm b}) \text{ convex} \,\mid \, \exists \ {\Delta} \text{ such that } \textnormal{P\&L}_{\psi^{\rm a}-\psi^{\rm b},\Delta}^H\ge 0\big\}.$$
Our strong duality result (\cref{thm:strong-duality}) will prove that the optimal superhedging cost~\eqref{eq:dual-intro} coincides with the value of the primal problem~\eqref{eq:BAMOTIntro}. The presence of bid--ask spreads increases the value of the dual problem as $\text{cost}(\psi^{\rm b,a}) \ge \mu(\psi^{\rm a} - \psi^{\rm b})$ for any $\mu^{\rm b}\cleq \mu \cleq \mu^{\rm a}$. We will see in \cref{sec:numerics} that the value may differ substantially from the price under the mid-marginal $\mu=(\mu^{\rm a}+\mu^{\rm b})/2$ given real market data.

In contrast to the dual in classical MOT, \eqref{eq:dual-intro} is nontrivial even for vanilla claims $H=h(X_T)$. %
Even if $h=\psi^0-\psi^1$ is itself a difference of convex functions, e.g., via %
Jordan decomposition, there may be convex pairs $(\psi^{\rm a},\psi^{\rm b})=:\psi^{\rm b,a}$ such that $\psi^{\rm a}-\psi^{\rm b} \geq h$ and $\text{cost}(\psi^{\rm b,a}) < \mu^{\rm a}(\psi^0) - \mu^{\rm b}(\psi^1)$. That is, the \emph{Jordan decomposition can be suboptimal} for the dual problem. A concrete example is shown in \cref{fig:BS_numerics}, where the risk-reversal strategy $h(x)=(x-1.05)^+-(0.95-x)^+$ is optimally superhedged by a linear combination of \emph{three} call options.

\subsection{Contributions}

On the conceptual side, we introduce 
a novel generalization of martingale optimal transport that is consistent with bid--ask spreads for vanilla options. While directly using order book data of bid and ask quotes would lead to a general linear programming problem, a key observation is that the existing market practice of fitting bid and ask marginal distributions opens the door to a tractable formulation in the realm of optimal transport.

On the theoretical side, in \cref{sec:BAMOT}, we first formalize the corresponding primal (transport) problem and the dual (superhedging) problem. We also describe more precisely the relation between discrete bid and ask quotes in order books and the corresponding idealized marginal distributions. In \cref{sec:duality}, we then establish strong duality: the BAMOT primal value~\eqref{eq:BAMOTIntro} coincides with the dual superhedging cost~\eqref{eq:dual-intro} for any upper semicontinuous payoff $h(X_{T_1},\dots,X_{T_N})$ with at most linear growth. Methodologically, we first show strong duality for the single-maturity case by a Hahn--Banach argument (\cref{thm:band-duality}). This case corresponds to superhedging a payoff $h(X_T)$ by a difference $\psi^{\rm a} - \psi^{\rm b}$ of convex functions of $X_T$ and is genuinely novel compared to MOT. We then show strong duality for the general, multi-maturity case (\cref{thm:strong-duality}) by combining the single-period result with MOT duality via a minimax argument.

\Cref{sec:convergence} continues the theoretical analysis by studying the relationship between BAMOT and the classical MOT problem without bid--ask spreads. \Cref{prop:consistency} establishes consistency, in the sense that the BAMOT value converges to the MOT limit as the bid--ask spread vanishes. By duality, the optimal superhedging costs then also converge. In the important special case of a single maturity, we further provide a quantitative analysis. To that end, we introduce a novel metric, the \emph{bid--ask distance}, which aligns with the spreads of quoted vanilla options and captures the convex ordering between marginals. We then show a linear rate of convergence of the BAMOT price for payoffs, which are differences of convex functions (\cref{prop:MidEstimate}), and a square-root rate for upper semicontinuous payoffs (\cref{prop:USC_sqrt_rate}). Both rates are shown to be numerically sharp in examples, namely, for a risk-reversal strategy and an at-the-money digital option (\cref{sec:convDigital}).

On the practical side, we provide analytical and computational examples illustrating the formulation and its implications. In \cref{sec:digital_options}, we derive a closed-form solution for digital call options in a one-sided market and validate the impact of bid--ask frictions using real market data on the S\&P~500 Index (SPX). Our findings reveal that the common practice of pricing digital options using the mid-price marginal can substantially underestimate (superhedging) prices relative to BAMOT in the presence of bid--ask frictions. In \cref{sec:forward-start}, we compare BAMOT super- and subhedging bounds for a forward-start payoff with those obtained under classical MOT. Finally, in \cref{sec:convDigital}, we document the convergence behavior as bid--ask spreads shrink for a risk-reversal strategy and an at-the-money digital call option.

\subsection{Related Literature}
Addressing market friction arising from bid--ask spreads has a long history in robust finance (see, e.g., \citep{BayraktarZhang2016, BouchardNutz2015, CheriditoKupperTangi2017, DupireBundle, JOUINIKALLAL} and the references therein). Transaction cost analysis was pioneered by \cite{JOUINIKALLAL}, which introduced a sublinear pricing operator %
to capture the bid--ask spread of trading a contingent claim. 
It is shown that the superhedging cost of the claim
agrees with the maximum contract value among all martingale measures consistent with the bid and ask prices of the hedging instruments. 
Closer to our robust finance framework, \cite{DolinskySoner2014} utilized sublinear operators to describe the cost of static hedging and established general duality results that also include proportional transaction costs for dynamic hedging. 
In the present work, we assume that costs for trading in the underlying are negligible relative to the bid--ask spreads of the vanilla options, as is typically satisfied in practice for moderate trading frequency. Moreover, we focus on a particular model close to industry practices, rather than aiming for general results.
A related duality was recently established by \cite{DupireBundle} in the context of generalized Limit Order Books (LOB) involving bundles of securities. 

Over the past several years, researchers have increasingly leveraged MOT to address market frictions. However, many works focus on capturing bid--ask spreads in less liquid exotic options~\cite{FahimHuangJui, HenryLabordere2013}. In particular, these studies assume that vanilla options can still be bought and sold at a unique price, leading to exact calibration of the marginal distributions. Other approaches within the MOT framework specifically accommodate the uncertainty in inferring marginal distributions from market spreads. For instance, \cite{HouObloj2018} introduced \textit{$\eta$-market models}, which utilize a uniform tolerance parameter $\eta$ for all vanilla options, effectively resulting in a constant bid--ask spread across all strikes. Alternatively, \cite{zhou2021} directly optimized over a set of martingale measures whose marginals reside within a pre-defined $\varepsilon$-tolerance set. Nonetheless, selecting appropriate hyperparameters for $\eta$ and $\varepsilon$ is often challenging in practice. By contrast, the bid and ask marginals imposed in the present work are motivated by industry practice and can be directly calibrated to the bid and ask prices observed in market data. We mention that \cite{Neufeld2022} took steps toward addressing these concerns, albeit from a more empirical perspective: the authors proposed a neural network-based approach to solve the dual problem, incorporating an objective function that explicitly accounts for the discrepancy between bid and ask prices.

In a different direction, not directly related to optimal transport, the modeling of bid--ask spreads using separate marginals has been explored through ``lower" and ``upper" measures by \cite{MadanSchoutens, madan2023modeling}. In these studies, positions are grouped by directional sentiment: bullish views (long calls, short puts) are valued under the upper measure, while bearish views (long puts, short calls) utilize the lower measure. Consequently, the ask prices for calls and puts originate from different distributions. While empirically robust, this framework does not provide the unified marginal constraints required for MOT, where bid and ask prices are recovered from a single consistent pair of distributions.

Finally, our results in \cref{sec:convergence} about convergence of BAMOT for shrinking bid--ask spreads are related to the literature on the stability of MOT, that is, its behavior as a function of the given marginals. %
Here~\cite{Juillet2016} derived the stability of the left-curtain coupling, thereby establishing stability results for MOT. Independently, \cite{Wiesel2023,BeiglbockJourdainMargheritiPammer2023} obtained continuity of the value functions and optimal couplings by exploiting the adapted Wasserstein topology. Subsequent works~\cite{NeufeldSester2021,BeiglbockJourdainMargheritiPammer2022} further confirmed this topology as the natural framework for studying stability in MOT. On the computational side, \cite{GuoObloj} investigated stability in order to provide guarantees for numerical methods when the marginals are approximated discretely. %

\subsection{Notation}

Let $\cC_L(\R^N_+)$ be the Banach space  of 
 continuous  functions $f: \R^N_+\to\R$ with linear growth, endowed with the norm 
 $$\lVert f \rVert_L := \sup_{x\in \R_+^N}\frac{|f(x)|}{1 + \|x\|}.$$ 
 Introduce also the larger set $\textnormal{USC}_L(\R_+^N)$ of upper semicontinuous (u.s.c.) functions  with linear growth, and the smaller set $\text{CVX}_L(\R_+^N)$ of convex functions with linear growth. When $N= 1$, we drop the dependence on the ambient space and simply write $\cC_L, \text{CVX}_L, \textnormal{USC}_L$.

Let $\cP_1(\R^N_+)$ be the set of probability measures on $\R^N_+$ with finite first moment, and abbreviate again $\cP_1 = \cP_1(\R_+)$. The 1-Wasserstein distance $\cW_1(\mu,\nu)$ is defined via 
$
 \cW_1(\mu,\nu) = \inf \int \|x-y\| \PP(\rd x,\rd y),
$
where the infimum is taken over all couplings $\PP$ of $(\mu, \nu)$, i.e., all $\PP \in \cP(\R_+^N\times \R_+^N)$ with marginals $(\mu,\nu)$. Integrals are denoted, interchangeably, by  
 $$\mu(f) = \E_{\mu}[f]= \int f \dd\mu. $$%
We say that $\mu,\nu \in \cP_1$ are in \textit{convex order} ($\mu\cleq \nu$) if
$\mu(\psi) \le \nu(\psi)$ for all $\psi\in \text{CVX}_L$. In fact, it is sufficient to verify the latter inequality for call payoffs, provided that the measures have identical barycenters    \begin{equation}\label{eq:callCVX}
    \mu \cleq \nu \qquad  \Longleftrightarrow \qquad \E_\mu[X] = \E_\nu[X] \quad\text{and}\quad \E_\mu[(X-K)^+]   \leq  \E_\nu[(X-K)^+] \quad \forall \ K\ge 0.
    \end{equation}
See, e.g., \cite{Shaked2007} for further background on  convex order, and in particular \cite[Theorem~3.A.1]{Shaked2007} for the above equivalence.

\paragraph{Organization}
The remainder of the paper is organized as follows. \Cref{sec:BAMOT} presents a practitioners' construction of bid--ask marginals and formulates the primal and dual BAMOT problems. \Cref{sec:duality} establishes strong duality. \Cref{sec:convergence} proves consistency and derives the convergence rate in the single-maturity case $N=1$. \Cref{sec:numerics} presents illustrative examples and numerical results. We conclude by discussing possible directions for future work in \cref{sec:future-work}.  
\cref{sec:existenceBidAsk} discusses the existence of bid and ask measures. %
Technical proofs and other complementary materials are presented in  \cref{app:proof,app:LP}.

\section{Bid--Ask MOT}\label{sec:BAMOT}

This section introduces the primal (measure) problem and the dual (portfolio) problem of BAMOT in detail. Both problems depend crucially on the assumption of bid and ask marginals that induce the prices of liquidly traded call and put options. Working with such marginals is market practice, offering  a flexible method to calibrate bid--ask implied volatility skews, while  ruling out  static arbitrage. \cref{fig:IVIntro} replicates a typical market interface for implied volatility skews, such as the OVDV function on the Bloomberg Terminal. %
On the other hand, a snapshot of real-world order book data may contain quotes that cannot be exactly reproduced by two distributions. The next subsections explain how practitioners typically calibrate bid and ask marginals to quotes; see \cite{brigoMercurio} and the OVDV documentation available on the Bloomberg Terminal. 

\subsection{Bid and Ask Marginals}\label{sec:bid-ask-marginals}

Next, we describe a common procedure used by practitioners to construct bid and ask marginals by calibrating to bid and ask quotes of out-of-the-money (OTM) put and call options.\footnote{OTM options are chosen as they are usually more liquid and therefore lead to more reliable and consistent quotes.} Fix a maturity $T$, let $F_T, x_0$ denote, respectively, the forward and spot price of the underlying, let $\{K_m\}_{m\in[M]}$ be the set of available strikes, and let $\{\text{OTM}^{\rm a}(K_m,T)\}_{m\in[M]}$ and $\{\vartheta^{\rm a}(K_m,T)\}_{m\in[M]}$ denote, respectively, the ask prices and Vegas of the corresponding OTM options (put if $K_m< F_T$ and call otherwise). %

Practitioners then select a parametric family of densities $\{\phi^\theta\mid \theta\in\Theta\}$, where $\Theta$ is chosen so that the forward price is matched (i.e., $\int x \phi^\theta(x)\dd x = F_T$ for all $\theta\in\Theta$), and solve the  optimization
\begin{align*}
   \theta^{\rm a} \in \arg\min_{\theta \in \Theta} \sum_{m\in[M]}
   \bigg|\frac{\text{OTM}^{\rm a}(K_m,T) - \text{OTM}^{\theta}(K_m,T)}{\vartheta^{\rm a}(K_m,T)}\bigg|^2,
\end{align*}
where $\text{OTM}^{\theta}(K,T)$ denotes the OTM option price under $\phi^\theta$. Here pricing errors are weighted by Vegas to approximately convert price discrepancies into errors in implied volatilities~(IVs). The resulting distribution $\mu^{\rm a}(\rd x)=\phi^{\theta^{\rm a}}(x)\dd x$ is referred to as the calibrated \emph{ask marginal}. One then perturbs the fitted parameters $\theta^{\rm a}$ to obtain bid parameters $\theta^{\rm b}$ that fit bid quotes, and sets $\mu^{\rm b}(\rd x)=\phi^{\theta^{\rm b}}(x)\dd x$ as the corresponding \emph{bid marginal}. In particular, one uses a perturbation such that $\mu^{\rm b}\cleq \mu^{\rm a}$ to rule out \emph{static arbitrage}.

A common choice of parametric family is mixtures of log-normal densities \cite{brigoMercurio}. %
Suppose the ask marginal is calibrated as a log-normal mixture with $J$ components, i.e.,
\[
\phi^{\rm a}(x)=\sum_{j=1}^J w_j\,\phi_{\rm BS}(x;z_j,\sigma_j^{\rm a}),
\qquad \sum_{j=1}^J w_j=1,
\]
where $\phi_{\rm BS}(\cdot;z,\sigma)$ denotes the density of a log-normal distribution with mean and volatility parameters $(z,\sigma)\in\R_+^2$. The bid marginal is then constructed by scaling down the volatility parameters, while keeping the remaining parameters unchanged:
\[
\phi^{\rm b}(x)=\sum_{j=1}^J w_j\,\phi_{\rm BS}(x;z_j,\sigma_j^{\rm b}),
\]
for some calibrated $\sigma_j^{\rm b}\leq \sigma_j^{\rm a}$, $j\in[J]$.

Given multiple maturities $T_1 <  T_2 <  \cdots < T_N$, we apply the above procedure at each maturity and obtain bid and ask marginals $\mu_i^{\rm b},\mu_i^{\rm a}$ such that $\mu_i^{\rm b}\cleq \mu_i^{\rm a}$ for each $i\in[N]:=\{1,\dots, N\}$.

For simplicity, assume zero interest rate and no carry cost.\footnote{With positive interest rates and/or dividends, the marginals $\mu_i^{\rm b},\mu_i^{\rm a}$ should be calibrated from \emph{discounted, dividend-adjusted quotes}.} In this case, $F_{T_i}=x_0$ for all $i\in[N]$. The absence of \emph{calendar arbitrage} across maturities suggests the cross-maturity condition $\mu_i^{\rm b}\cleq \mu_j^{\rm a}$ for any $1\leq i\leq j\leq N$. Indeed, if there exists $K\geq 0$ such that $\E_{\mu_i^{\rm b}}[(X-K)^+]>\E_{\mu_j^{\rm a}}[(X-K)^+]$, then one could sell the $(K,T_i)$ call at a price  $\E_{\mu_i^{\rm b}}[(X-K)^+]$ and buy the $(K,T_j)$ call at a price   $\E_{\mu_j^{\rm a}}[(X-K)^+]$; together with dynamic trading in the underlying, this yields an arbitrage.

The above discussion motivates the following assumption, which we adopt throughout the paper.

\begin{assumption}\label{asm:bidAskConvexOrder}
There exist bid and ask marginals $\mu^{\rm b}=(\mu^{\rm b}_1,\ldots,\mu^{\rm b}_N)$ and $\mu^{\rm a}=(\mu^{\rm a}_1,\ldots,\mu^{\rm a}_N)$ with finite first moment, which determine the bid and ask prices of all call and put options with the respective maturities and satisfy
\begin{equation}\label{eq:bidAskConvexOrder}
\mu_i^{\rm b}\cleq \mu_j^{\rm a}, \qquad \forall\, 1\leq i \leq j \leq N .
\end{equation}
\end{assumption}

\begin{remark}
\begin{enumerate}[label=(\alph*)]
\item 
Condition~\eqref{eq:bidAskConvexOrder} rules out butterfly arbitrage;  for any $i\in[N]$, $K_1<K_2<K_3$ and $K_2=\lambda K_1+(1-\lambda)K_3$ with $\lambda\in(0,1)$, 
\begin{equation}\label{eq:butterfly}
\lambda\, \E_{\mu_i^{\rm a}}[(X-K_1)^+] - \E_{\mu_i^{\rm b}}[(X-K_2)^+] + (1-\lambda)\, \E_{\mu_i^{\rm a}}[(X-K_3)^+] \geq 0.
\end{equation}
Indeed, $\E_{\mu_i^{\rm b}}[(X-K_2)^+] \leq \E_{\mu_i^{\rm a}}[(X-K_2)^+]$ by~\eqref{eq:bidAskConvexOrder}, so the left-hand side of~\eqref{eq:butterfly} is bounded from below by
$\E_{\mu_i^{\rm a}}[B(x;K_1,K_2,K_3)]$, with the butterfly spread  
$$B(x;K_1,K_2,K_3)=\lambda (x-K_1)^+-(x-K_2)^++(1-\lambda)(x-K_3)^+.$$ Since $B(\cdot;K_1,K_2,K_3)\geq 0$, its price under $\mu_i^{\rm a}$ (and hence under any admissible bid--ask calibration) must be nonnegative.

\item By a suitable choice of the bid and ask marginals, our framework covers situations in which the payoff depends on maturities for which vanilla options are not liquidly traded. For example, if $N=2$ and the payoff is $H=h(X_{T_1},X_{T_2})$, while vanillas are available at $T_2$ but not at $T_1$, one may take $\mu^{\rm b}_1:=\delta_{x_0}$ and $\mu^{\rm a}_1:=\mu^{\rm a}_2$, thereby leaving the $T_1$-marginal effectively unconstrained, while still enforcing the cross-maturity condition~\eqref{eq:bidAskConvexOrder}. Such a situation naturally arises when the product references an intermediate fixing date, e.g., forward-start, whereas market quotes are only liquid at coarser maturities.
\end{enumerate}
\end{remark}

\subsection{Primal (Measure) Problem}
Fix maturities $T_1<  T_2 <  \cdots  <  T_N$ with corresponding bid marginals $\mu^{\rm b}=(\mu^{\rm b}_1,\dots,\mu^{\rm b}_N)$ and ask marginals $\mu^{\rm a}=(\mu^{\rm a}_1,\dots,\mu^{\rm a}_N)$. We write $X_1, \dots, X_N$ for the canonical process representing the spot price, where we identify $X_i=X_{T_i}$ for notational simplicity. By convention, the initial price $X_0:=x_0\in\R_+$ is deterministic. Thus, a \emph{martingale measure} is a probability measure $\Q$ on $\R_+^N$ such that $(X_i)_{0\leq i\leq N}$ is a $\Q$-martingale. The set of \emph{calibrated} martingale measures is
\begin{equation}\label{eq:martingale-couplings}
    \cQ(\mu^{\rm b}, \mu^{\rm a}):= \left\{\Q \mid \text{martingale measure}, \,\, \mu_i^{\rm b}\cleq \mu_i^{\Q}\cleq \mu_i^{\rm a},\,\, i \in [N]\right\},
\end{equation}
where $\mu_i^{\Q}$ denotes the $i$-th marginal distribution of  $\Q$. 
Fix a contingent claim $H = h(X_1, \ldots, X_N)$ for some   payoff function $h\in \text{USC}_L(\R^N_+)$. Then the \emph{primal} (or \emph{measure}) problem is defined as
\begin{equation}\label{eq:BAMOT-primal}
            P(h):=\sup_{\Q \in \cQ(\mu^{\rm b}, \mu^{\rm a})} \E_\Q[h(X_1, \ldots, X_N)].
\end{equation}
We observe that 
\[
    \text{Assumption } \eqref{eq:bidAskConvexOrder} \quad \Longleftrightarrow \quad \cQ(\mu^{\rm b}, \mu^{\rm a})\neq \emptyset.
\]
Indeed, define $\underline{\mu}_i^{\rm a}=\min_{j\geq i}\mu_j^{\rm a}$, where the infimum is with respect to the convex order. Then $\underline{\mu}_i^{\rm a}\cleq \underline{\mu}_j^{\rm a}$ for all $i \leq j$ and $\mu_i^{\rm b}\cleq \underline{\mu}_i^{\rm a}\cleq \mu_i^{\rm a}$ for all $i$. By Strassen's theorem, the former implies the existence of a martingale measure $\Q$ with marginals $(\underline{\mu}_1^{\rm a},\dots,\underline{\mu}_N^{\rm a})$, and $\Q\in \cQ(\mu^{\rm b}, \mu^{\rm a})$ by the latter. Conversely, if $\Q\in \cQ(\mu^{\rm b}, \mu^{\rm a})$, let $(\mu_1,\dots,\mu_N)$ be its vector of marginals. The martingale property implies $\mu_i\cleq \mu_j$ for all $i \leq j$ and then $\Q\in \cQ(\mu^{\rm b}, \mu^{\rm a})$ implies~\eqref{eq:bidAskConvexOrder}.     

\begin{remark}\label{rem:peacockEnvelopes}
    Consider again the lower envelope $\underline{\mu}_i^{\rm a}:=\min_{j\geq i}\mu_j^{\rm a}$ used in the above argument, and similarly the upper envelope $\overline{\mu}_i^{\rm b}:=\max_{j\leq i}\mu_j^{\rm b}$. We see that~\eqref{eq:bidAskConvexOrder}, and in turn $\cQ(\mu^{\rm b}, \mu^{\rm a})\neq\emptyset$, are equivalent to $\overline{\mu}_i^{\rm b}\cleq \underline{\mu}_i^{\rm a}$ for all $i$. That is, the prices of convex payoffs with respect to the lower and upper envelopes do not cross, as illustrated in \cref{fig:termStructure} for an illiquid, in-the-money call on the S\&P~500 Index (SPX).   %
It turns out that the 
    bid and ask flows of marginals can be made increasing with respect to convex order by enhancing the quotes \emph{across maturities}, see Appendix~\ref{app:quoteAcrossMaturities}. In this case, the flows coincide with the peacock envelopes $\overline{\mu}^{\rm b}, \underline{\mu}^{\rm a}$. We  favor the minimal assumption \eqref{eq:bidAskConvexOrder} here  for greater generality.
\end{remark}

 \begin{figure}[H]
    \centering
    \caption{ Term structure of  S\&P~500 Index  (SPX) call options struck at $K=5550$ as of 2025-02-27,  and corresponding price envelopes. }%
    \includegraphics[width=3.2in, height = 2.2in]{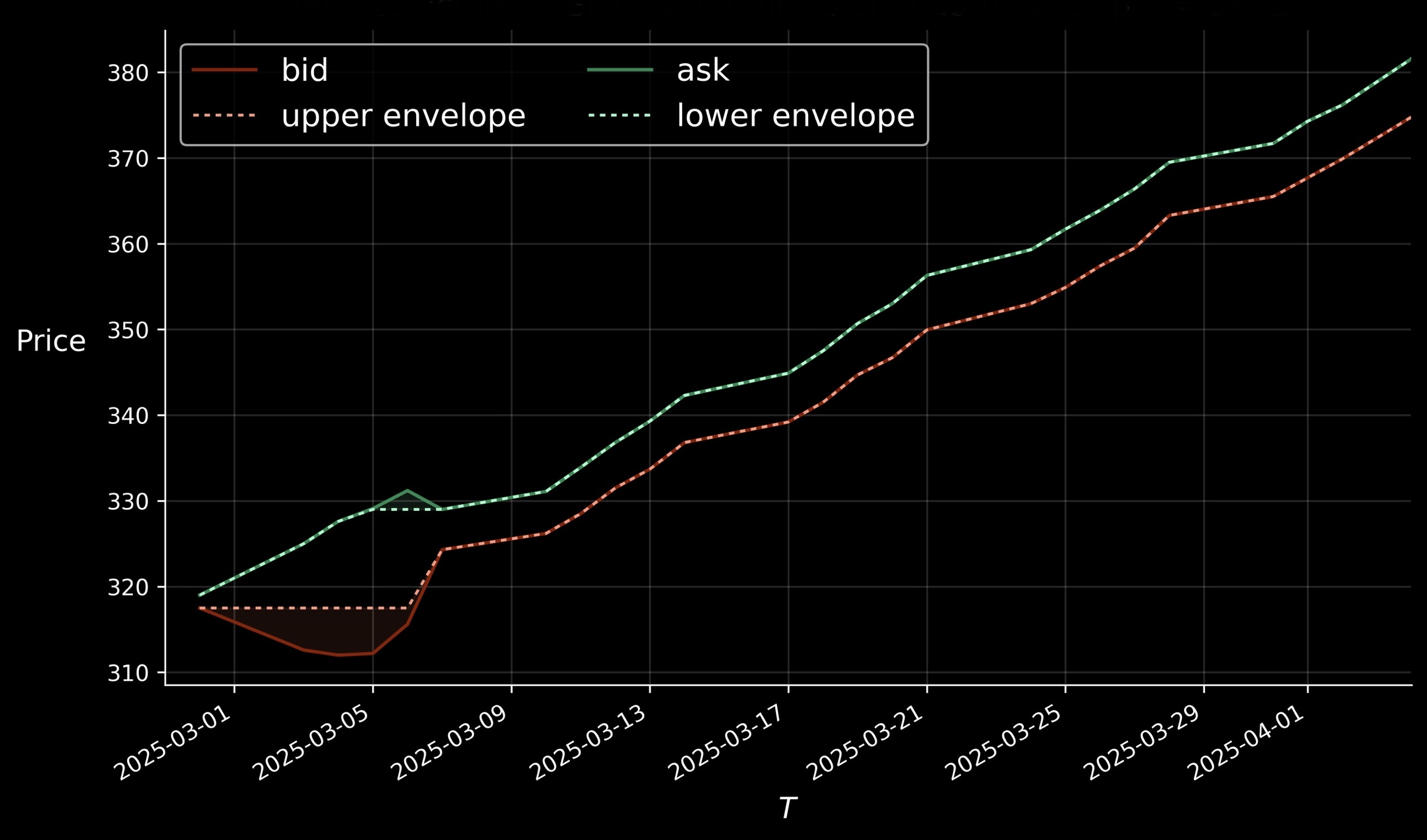}
    
    \label{fig:termStructure}
\end{figure}

\subsection{Dual (Portfolio) Problem}
 
 Consider a payoff function $h\in \text{USC}_L(\R_+^N)$, meaning that the claim is $h(X_1,\ldots, X_N)$, and static hedging instruments $\psi = (\psi_1,\ldots, \psi_N)$ with payoffs $\psi_i(X_i)$. Moreover, let $\Delta = (\Delta_1,\ldots, \Delta_{N-1})$ be a progressively measurable trading strategy (delta hedge), that is, $\Delta_i$ is a measurable function of $X_1,\ldots, X_i$. The profit and loss associated with $\psi,\Delta, h$ is defined as 
\begin{equation*}\label{eq:PnL}
    \textnormal{P\&L}_{\psi,\Delta}^h(x_1,\ldots,x_N) = \sum_{i=1}^N\psi_i(x_i) + \sum_{i=1}^{N-1} \Delta_i(x_1,\ldots, x_i)(x_{i+1}-x_i) - h(x_1,\ldots,x_N).
\end{equation*}
As customary in MOT \cite{beiglbock2013model,HLBook}, %
it is sufficient to look for static profiles of the form 
$$\psi_i(x) = \gamma^{\$}_i + \int_{\R_+}(x-K)^{+}\lambda_i(\rd K),$$
noting that eventual positions in the forward contract can be absorbed by the delta hedge. 
In our dual problem, the delta hedges are arbitrary, whereas the static positions are expressed as the difference $\psi_i = \psi^{\rm{a}}_i - \psi^{\rm{b}}_i$ of convex functions. Here $\psi^{\rm{a}}_i$ corresponds to a combination of call options with maturity $T_i$ that are bought, while $\psi^{\rm{b}}_i$ corresponds to call options which are sold. (Put options are redundant given calls and the forward contract, hence can be ignored.)
For brevity, we combine these functions into a vector 
$$\psi^{\rm b,a}=(\psi^{\rm a},\psi^{\rm b})=(\psi^{\rm a}_1, \dots, \psi^{\rm a}_N, \psi^{\rm b}_1, \dots, \psi^{\rm b}_N)\in \text{CVX}_L^{2N}$$
whose total cost is 
\begin{equation*}
\text{cost}(\psi^{\rm b,a}) =  \sum_{i=1}^N  \mu^{\rm a}_i(\psi^{\rm a}_i) - \mu^{\rm{b}}_i(\psi^{\rm b}_i).
\end{equation*}
We denote the set of all admissible superhedging strategies by
 \begin{equation}\label{eq:admissibleDual}
     \Psi^{\rm b,a}(h) = \Big\{  (\psi^{\rm a},\psi^{\rm b})\in \text{CVX}_L^{2N} \, \mid \, \exists \ {\Delta} \text{ s.t.\  }   \textnormal{P\&L}_{\psi^{\rm a} - \psi^{\rm b},\Delta}^h\ge 0  \Big\};
 \end{equation} 
here and below, it is tacitly understood that $\Delta$ is progressively measurable. The \emph{dual} (or \emph{portfolio}) problem is then defined as
\begin{equation}\label{eq:BAMOT-dual}
D(h):=\inf_{\psi^{\rm b,a} \in \Psi^{\rm b,a}(h)} \text{cost}(\psi^{\rm b,a}). %
\end{equation} 
The interpretation is straightforward: sell one unit of $h$; at each maturity $T_i$, buy a convex profile $\psi_i^{\rm a}$ priced at the ask, sell another profile $\psi_i^{\rm b}$ at the bid,  %
and hold  $\Delta_i$ shares  of the underlying asset to hedge against potential losses and minimize cost.

\begin{remark}\label{rem:calibration-uncertainty}
The primal problem in~\eqref{eq:BAMOT-primal} is conditional on a fixed
pair of calibrated bid and ask marginals. Finite market quotes need not determine this pair uniquely: different
arbitrage-free interpolation or fitting procedures, and even different
solutions within a fixed parametric family, may yield different
admissible pairs. Given a
fixed class $\mathfrak C$ of admissible calibrated pairs satisfying~\cref{asm:bidAskConvexOrder}, one could instead consider
\[
\sup_{(\mu^{\rm b},\mu^{\rm a})\in\mathfrak C}
\ \sup_{\Q\in\cQ(\mu^{\rm b},\mu^{\rm a})}
\E_\Q[h(X_1,\ldots,X_N)].
\]
The inner optimization is the BAMOT problem conditional on a calibrated
envelope, whereas the outer optimization additionally accounts for
calibration uncertainty. The corresponding dual problem then becomes 
\[
\inf_{\psi^{\rm b,a} \in \Psi^{\rm b,a}(h)} \text{cost}_{\mathfrak C}(\psi^{\rm b,a})\quad{\rm where}\quad \text{cost}_{\mathfrak C}(\psi^{\rm b,a}) = \sup_{(\mu^{\rm b},\mu^{\rm a})\in\mathfrak C} \sum_{i=1}^N  \mu^{\rm a}_i(\psi^{\rm a}_i) - \mu^{\rm{b}}_i(\psi^{\rm b}_i).
\]
Establishing strong duality in this case therefore requires additional structural
assumptions on $\mathfrak C$, which we leave for future work.
\end{remark}

\section{Duality}\label{sec:duality}

The goal of this section is to show the strong duality $P(h)=D(h)$, where $P(h)$ and $D(h)$ denote the values of the primal \eqref{eq:BAMOT-primal} and dual \eqref{eq:BAMOT-dual} problems, respectively. We first state the weak duality $P(h)\le D(h)$, which is straightforward. 

\begin{lemma}[Weak Duality]\label{lemma:weak-duality}
We have $P(h)\le D(h)$ for all $h\in \textnormal{USC}_L(\R^N_+)$.
\end{lemma}

\begin{proof}
Let $\Q \in  \cQ(\mu^{\rm b}, \mu^{\rm a})$ and $\psi^{\rm b,a} = (\psi^{\rm a},\psi^{\rm b}) \in \Psi^{\rm b,a}(h)$, and let $\Delta$ be a corresponding dynamic strategy such that $\textnormal{P\&L}_{\psi^{\rm a} - \psi^{\rm b},\Delta}^h\ge 0 $. %
    Since $\Q$ is a martingale measure with $\mu_i^{\rm b}\cleq \mu_i^{\Q}\cleq \mu_i^{\rm a}$, and  $\psi_i^{\rm a}, \psi_i^{\rm b}$ are convex with linear growth, we deduce
    \begin{align*}
          \E_\Q[h(X_1,\ldots,X_N)] &\leq \E_\Q\bigg[\sum_{i=1}^N \psi_i^{\rm a}(X_i) - \psi_i^{\rm b}(X_i)\bigg] + \E_\Q\bigg[\sum_{i=1}^{N-1}\Delta_i(X_1,\ldots, X_i)(X_{i+1}-X_i)\bigg] \\ 
          &= \sum_{i=1}^N \mu_i^{\Q}(\psi_i^{\rm a}) - \mu_i^{\Q}(\psi_i^{\rm b}) 
          \leq \sum_{i=1}^N \mu_i^{\rm a}(\psi_i^{\rm a}) - \mu_i^{\rm b}(\psi_i^{\rm b})=\text{cost}(\psi^{\rm b,a}).
    \end{align*}
    Here we have used that the discrete stochastic integral has vanishing expectation. Indeed, the negative part of the terminal value of this local martingale is $\Q$-integrable due to the superhedging property and the linear growth of $h$, and that implies that it is a true martingale \cite[Theorem~2]{JacodShiryaev.98}. The result follows by 
    taking supremum over $\Q$ and infimum over $\psi^{\rm b,a} \in \Psi^{\rm b,a}(h)$. %
\end{proof}

The nontrivial direction of the strong duality will be established in two steps. We first focus on the single-maturity case $N=1$, where we prove strong duality from first principles via a Hahn--Banach argument. We then show the result for the general case by combining the strong duality of classical MOT with the single-maturity case.

\subsection{Single Maturity}

Let $N=1$ and consider the single maturity $T_1>0$ as well as corresponding bid and ask marginals $\mu^{\rm{b}}, \mu^{\rm{a}} \in \cP_1$ satisfying $\mu^{\rm{b}} \cleq \mu^{\rm{a}}$.
We slightly abuse notation by writing $\mu^{\rm s}$ for $\mu^{\rm s}_1$, etc., and identifying martingale measures with their marginal $\mu$ at $T_1$ (rather than using the full martingale measure $\delta_{x_0}\otimes \mu$).
Then the primal problem \eqref{eq:BAMOT-primal} boils down to %
\begin{equation}\label{eq:primal1M}
P(h) = \sup_{\mu \in  \cQ(\mu^{\rm{b}}, \mu^{\rm{a}})}\mu(h), \qquad 
    \cQ(\mu^{\rm{b}}, \mu^{\rm{a}}) = \{\mu \in \cP_1 \, \mid \, \mu^{\rm{b}} \cleq \mu \cleq  \mu^{\rm{a}} \}, 
\end{equation}
where $h\in \textnormal{USC}_L $ is the payoff function of a $T_1$-claim $h(X_1)$. 
On the other hand, the dual problem~\eqref{eq:BAMOT-dual} becomes 
\begin{align*}
  D(h) &= \inf_{\psi^{\rm b,a} \in \Psi^{\rm b,a}(h)} \text{cost}(\psi^{\rm b,a}), \qquad 
 \text{cost}(\psi^{\rm b,a}) = \mu^{\rm{a}}(\psi^{\rm a}) - \mu^{\rm{b}}(\psi^{\rm b}), \label{eq:DualOneMat}\\[1em]
  \Psi^{\rm b,a}(h)  &=   \{ \psi^{\rm b,a}  = (\psi^{\rm a},\psi^{\rm b})\in \text{CVX}_L^2 \,\mid \, \psi^{\rm a} - \psi^{\rm b} \geq h\}. \nonumber
\end{align*}
That is, $D(h)$ is the cost of the cheapest static superhedging portfolio of $h$, given by the difference of convex functions. Observe that $P(h) = \mu^{\rm a}(h) = D(h)$ if $h$ is convex, and $P(h) = \mu^{\rm b}(h) = D(h)$ when $h$ is concave. The next result extends strong duality to general payoffs. 

\begin{theorem}[Strong Duality,  $N=1$]\label{thm:band-duality}
For all $h \in \textnormal{USC}_{L}$, we have
$$P(h) = \sup_{\mu \in \cQ(\mu^{\rm b}, \mu^{\rm a})} \mu(h) = 
\inf_{\psi \in \Psi^{\rm{b,a}}(h)} \textnormal{cost}(\psi) =  D(h).
$$
Moreover, there exists a primal optimizer $\mu^{\star}\in \cQ(\mu^{\rm b}, \mu^{\rm a})$. 
\end{theorem}

\begin{proof}%
    The weak duality $P(h)\le D(h)$ was shown in \cref{lemma:weak-duality}. To show the reverse inequality, we first assume that $h\in \cC_L(\R_+)$. Our aim is to construct a measure $\mu^{\star} \in   \cQ(\mu^{\rm{b}}, \mu^{\rm{a}}) $ such that $\mu^{\star}(h) = D(h)$. Since we already know that $\mu^{\star}(h) \le P(h)\le D(h)$, this will imply that $P(h)= D(h)$ and $\mu^{\star}$  is a maximizer. \\[-0.75em] %

\underline{\textit{Step 1.}}  
We show that $D: \cC_L(\R_+)\to \R$ is well-defined and a sublinear functional. %
Let $f\in\cC_L(\R_+)$. The convex function $\psi^{\rm a}(x):=\|f\|_L(1+|x|)$
dominates $f$, and hence $(\psi^{\rm a},0)\in\Psi^{\rm b,a}(f)$. Therefore, $D(f)\leq \mu^{\rm a}(\psi^{\rm a})<\infty$.
On the other hand, fixing any $\mu_0\in\cQ(\mu^{\rm b},\mu^{\rm a})$, weak duality yields
\[
-\infty<\mu_0(f)\leq P(f)\leq D(f).
\]
Thus $D(f)\in\R$.

We now show that $D$ is sublinear. Clearly  $D(\lambda f) = \lambda D(f)$ for all $\lambda > 0$ and $f\in \cC_L(\R_+)$. 
Note also that  $D(\lambda) = \lambda$ for all $\lambda \in \R$, using, e.g.,  $\psi \equiv (\lambda,0) \in \Psi^{\rm{b,a}}(\lambda)$. Moreover, $D$ is subadditive, 
$$D(f + f') \le D(f) + D(f') \quad \forall \ f,f' \in \cC_L(\R_+).$$ %
Indeed,  $\Psi^{\rm{b,a}}(f) + \Psi^{\rm{b,a}}(f') \subseteq \Psi^{\rm{b,a}}(f+f')$, which gives
$$D(f+f') = \inf_{\psi \in \Psi^{\rm{b,a}}(f+f')} \text{cost}(\psi) \le \inf_{\psi \in \Psi^{\rm{b,a}}(f), \,  \psi'\in\Psi^{\rm{b,a}}(f')} \text{cost}(\psi + \psi') = D(f)  +  D(f')$$
since  $\text{cost}(\cdot)$ is linear.  This completes the proof that $D$ is sublinear.  \\[-0.75em] %

\underline{\textit{Step 2.}} Consider the one-dimensional linear subspace $\mathcal{M}:=\{\lambda h : \lambda \in \R\} \subseteq \cC_L(\R_+)$ and introduce the linear functional  $\ell(\lambda h) := \lambda D(h)$ on~$\mathcal{M}$. 
Then $\ell(\lambda h) = D(\lambda h)$ $\forall \lambda > 0$ due to the positive homogeneity of $D$. Moreover,  $\ell(-\lambda h) \le D(-\lambda h)$ for all $\lambda \ge 0$, as follows from 
$$\ell(\lambda h) + \ell(-\lambda h) = \ell(0) = 0 = D(0) = D(\lambda h - \lambda h) \le D(\lambda h) + D(- \lambda h). $$
Hence $\ell(f) \le D(f)$ for all  $f \in \mathcal{M}$, that is, $\ell$ is dominated by $D$ on $\mathcal{M}$. By the Hahn--Banach Dominated Extension Theorem, $\ell$ can be extended to a linear functional (again denoted $\ell$) on $\cC_L(\R_+)$ such that   $\ell(f) \le D(f)$  for all $f\in \cC_L(\R_+)$. %
Moreover, $\ell$ is a positive functional in the sense that $\ell(f) \ge 0$ whenever $f \ge 0$. Indeed, if $f\ge 0$ but $\ell(f)<0$, then as $(0,0)\in \Psi^{\rm{b,a}}(-f)$, we obtain the contradiction
$$0 < -\ell(f) = \ell(-f) \le D(-f) \le \text{cost}(0) = 0.$$
Finally, same as the argument in Step~1, the convex function 
 $\psi^{\rm a}(x) :=  \lVert f \rVert_{L}(1 + |x|)$ dominates $f$. In other words,   $(\psi^{\rm a}, 0) \in \Psi^{\rm{b,a}}(f)$, leading to 
$$\ell(f) \le D(f) \le \mu^{\rm{a}}(\psi^{\rm a}) =  (1+c)\ \lVert f \rVert_{L} ,$$%
where the constant $c =   \int |x|\mu^{\rm{a}}(\rd x)$ is finite due to $\mu^{\rm{a}} \in \cP_1$. In summary, $\ell$ is a bounded, linear, positive functional on $\cC_L(\R_+)$.  \\[-0.75em] 

\underline{\textit{Step 3.}} Let us verify the following tightness condition: for any $\vae >0$, there exists $R>0$ such that 
\[
|\ell(f)|\leq \vae \|f\|_L,\quad \forall f\in \cC_L(\R_+) \quad {\rm s.t.}\, f \equiv 0 \,\,{\rm on} \,\,[0,R].
\]
Let $f \in \cC_L(\R_+)$ satisfy $f \equiv 0$ on $[0,R]$, where $R\geq2$. Then $|f(x)|
\leq 3\|f\|_L(x-R/2)^+ =: \|f\|_L g_R(x)$. Moreover, as $\mu^{\rm a}$ has finite first moment, dominated convergence shows that $\mu^{\rm a}(g_R)\to 0$ for $R\to \infty$. By the positivity of $\ell$, we deduce that 
\[
\ell(f) \leq  \ell(\|f\|_L\,g_R) \leq \|f\|_L\, D(g_R) =  \|f\|_L\,\mu^{\rm a}(g_R)\to 0,\quad R\to \infty.
\]
The same argument applies to $-\ell(f)=\ell(-f)$, completing the proof of tightness.\\[-0.75em] 

\underline{\textit{Step 4.}} 
As $(\cC_L,\lVert \cdot \rVert_{L})$ is isomorphic to $(\cC_b,\lVert \cdot \rVert_{\infty})$ via $f \mapsto \frac{f}{1+|\cdot|}$ and $\ell$ is bounded linear (Step~2) satisfying the tightness condition (Step~3), the Riesz Representation Theorem \cite[Theorem~7.10.6]{Bogachev2007} implies the existence of a signed Radon measure $\mu^{\star}$ on $\R_+$ (with finite first moment), which represents $\ell$ via $\ell(f)=\mu^{\star}(f)$. Positivity of $\ell$ implies that $\mu^{\star}$ is a nonnegative measure. Moreover, $\ell(1) = 1$ due to  $\ell(\pm 1) \le D(\pm 1)  = \pm 1$, so that $\mu^{\star}$ is a probability measure. Let $f \in \text{CVX}_L$, then    $(f,0)\in \Psi^{\rm{b,a}}(f)$ and 
$$\mu^{\star}(f) = \ell(f) \le D(f) = \mu^{\rm{a}}(f).$$
Similarly, $(0,f)\in \Psi^{\rm{b,a}}(-f)$, leading to $\mu^{\star}(f) \ge  \mu^{\rm{b}}(f)$. Thus $\mu^{\rm{b}}(f)\leq \mu^{\star}(f) \leq \mu^{\rm{a}}(f)$ for all $f \in \text{CVX}_L$, showing that $\mu^{\star} \in \cQ(\mu^{\rm{b}},\mu^{\rm{a}})$. This completes the proof that $D(h) = \ell(h)=\mu^{\star}(h) \le P(h)$ for the given claim $h\in \cC_L(\R_+)$.  \\[-0.75em] 

\underline{\textit{Step 5.}} It remains to extend the result to $h \in \textnormal{USC}_{L}$. Consider the sup convolution
\begin{equation*}%
    h^{\varepsilon}(x) =  \sup_{y\in \R_+} \left( h(y)  - \frac{1}{\varepsilon} |y-x|\right), \quad \varepsilon > 0.
\end{equation*}
Then $ h^{\varepsilon}$ is finite for all $\varepsilon < \overline{\varepsilon} :=\lVert h \rVert_L^{-1}$. 
For any such $\varepsilon$, recall that $h^{\varepsilon}$ is  Lipschitz continuous with constant $1/\varepsilon$, as follows from
$$ h^{\varepsilon}(x) - h^{\varepsilon}(x') \le \frac{1}{\varepsilon} \sup_{y\in \R_+} \big [ |y-x| - |y-x'| \big ] \le \frac{1}{\varepsilon}  |x-x'|, \quad x,x' \in \R_+.$$
Hence $h^{\varepsilon} \in \cC_L(\R_+)$ and  from the previous steps,  $P(h^{\varepsilon}) = D(h^{\varepsilon})$. Moreover, $h^{\varepsilon} \downarrow h$ as $\varepsilon\downarrow 0$. By the $\cW_1$-compactness of the set $\cQ(\mu^{\rm b}, \mu^{\rm a})$ in the definition~\eqref{eq:primal1M}, this implies that  $P(h^{\varepsilon})\downarrow P(h)$ (as in, e.g., \cite[Theorem~31]{DenisHuPeng2011}). On the other hand, $D(h^{\varepsilon}) \ge D(h)$ for all $\varepsilon >0$ by monotonicity, so that $P(h) = \lim_{\varepsilon \to 0 }P(h^{\varepsilon}) \ge D(h)$. This completes the proof that $P(h) = D(h)$. The existence of a maximizer $\mu^{\star}\in \cQ(\mu^{\rm b}, \mu^{\rm a})$ follows from the $\cW_1$-compactness of $\cQ(\mu^{\rm b}, \mu^{\rm a})$ and $h\in \textnormal{USC}_{L}$.
\end{proof}

\begin{remark}
    Dolinsky and Soner \cite{DolinskySoner2014} treat market friction in the static hedging instruments through an abstract pricing operator $\mathscr{P}(\cdot)$ describing their prices. This operator is used in both the dual and the primal formulation. In the primal, the marginal constraint becomes $\mu(f) \le \mathscr{P}(f)$ for all measurable functions $f:\R_+ \to \R$ satisfying a certain growth condition. By contrast, our primal formulation~\eqref{eq:primal1M} with the supremum over $\mu \in \cQ(\mu^{\rm b},\mu^{\rm a})$ makes the constraint explicit in terms of bid and ask marginals obtained from market data. While \cref{thm:band-duality} could also be derived from the abstract duality of \cite{DolinskySoner2014} with some additional work, we preferred to provide an elementary proof.
\end{remark}

\subsection{Multiple Maturities} 

We now turn to strong duality for path-dependent claims involving $N\geq1$ maturities. Fix marginals $\mu^{\rm b}, \mu^{\rm a}\in\cP_1^N$ satisfying \cref{asm:bidAskConvexOrder}. We recall the set
$$    \cQ(\mu^{\rm b}, \mu^{\rm a}) = \left\{\Q \mid \text{martingale measure}, \,\, \mu_i^{\rm b}\cleq \mu_i^{\Q}\cleq \mu_i^{\rm a},\,\, i \in [N]\right\}$$
 of calibrated martingale measures and consider a payoff of the form $h(X_1,\ldots, X_N)$.

\begin{theorem}[Strong Duality]\label{thm:strong-duality} 
For all $h\in \textnormal{USC}_L(\R^N_+)$, we have   
$$P(h) = \sup_{\Q \in \cQ(\mu^{\rm b}, \mu^{\rm a})} \E_\Q[h(X_1, \ldots, X_N)] = 
\inf_{\psi \in \Psi^{\rm{b,a}}(h)} \textnormal{cost}(\psi) =  D(h)
$$
and there exists an optimizer $\Q^{\star} \in  \cQ(\mu^{\rm b}, \mu^{\rm a})$ to the primal problem. %
\end{theorem}
\begin{proof}%
Set $\cQ_i:= \cQ(\mu_i^{\rm b}, \mu_i^{\rm a})$ and $\cQ_{\otimes}:= \prod_{i=1}^N\cQ_i$. %
Then $\cQ_{\otimes}$ is convex and $\cW_1$-compact by \cref{lemma:convex-order-topology}. 
Next, observe that 
    \begin{align}\label{eq:multiDualityStep1}
    P(h) = \sup_{\mu \in \cQ_{\otimes}} \sup_{\Q \in\cQ(\mu)} \E_\Q[h], \quad\text{where}\quad     \cQ(\mu) = \big\{\Q \mid \text{martingale measure}, \,\, \mu_i^{\Q}= \mu_i,\,\, i \in [N]\big\}.
    \end{align}
    Consider $\mu \in \cQ_{\otimes}$ satisfying $\mu_i\cleq \mu_j$ for $1\leq i<j\leq N$. 
    We then apply the strong duality of classical MOT \cite[Corollary~1.2]{beiglbock2013model} to obtain 
    \begin{align}\label{eq:classicalDualityAppl}
       \sup_{\Q \in\cQ(\mu)} \E_\Q[h] =\inf_{\psi \in \Psi(h)} \sum_{i=1}^N \mu_i(\psi_i),
    \end{align}
where  $\Psi(h) = \{ \psi \in \cC_L^N  \, \mid \, \exists \ {\Delta} \text{ s.t.\ }   \textnormal{P\&L}_{\psi,\Delta}^h\ge 0  \}$. 
For $\mu \in \cQ_{\otimes}$ with $\mu_i\ncleq \mu_j$ for some $i<j$, we have $\cQ(\mu) = \varnothing$ and hence the supremum on the left-hand side equals $-\infty$. On the other hand, such $\mu$ allows for calendar arbitrage, and using also the linear growth of $h$, we see that the right-hand side also equals $-\infty$. In summary, \eqref{eq:classicalDualityAppl} holds for any  $\mu \in \cQ_{\otimes}$, and combining with~\eqref{eq:multiDualityStep1} yields
\begin{align*}
    P(h) = \sup_{\mu \in \cQ_{\otimes}} \sup_{\Q \in\cQ(\mu)} \E_\Q[h] = \sup_{\mu \in \cQ_{\otimes}}\inf_{\psi \in \Psi(h)} \sum_{i=1}^N \mu_i(\psi_i)\,.
\end{align*}
Next, we apply the minimax theorem. Consider $\cX= \cQ_{\otimes}$ equipped with the $\cW_1$-topology and $\cY = \Psi(h)$ equipped with the product topology induced by $\|\cdot\|_{L}$. Then $\cX, \cY$ are convex and  $\cX$ is compact. Moreover, $\cX\times \cY \ni (\mu, \psi)\mapsto \sum_{i=1}^N \mu_i(\psi_i)$ is linear and continuous with respect to the product topology. Therefore, the minimax theorem \cite[Corollary~3.3]{Sion.58} yields   
     \begin{equation}\label{eq:primal-inf-sup}
             \begin{aligned}
        P(h) =\inf_{\psi \in \Psi(h)}\sup_{\mu \in \cQ_{\otimes}} \sum_{i=1}^N \mu_i(\psi_i)= \inf_{\psi \in \Psi(h)}\sum_{i=1}^N \sup_{\mu_i \in \cQ_i}  \mu_i(\psi_i).
    \end{aligned}
     \end{equation}
    For each $i$, applying \cref{thm:band-duality} to the claim $f=\psi_i\in \cC_{L}$ yields 
    \begin{equation*}\label{eq:convex-band}
           \sup_{\mu_i\in \cQ_i} \mu_i(\psi_i) = \inf_{\psi_{i}^{\rm b,a}\in \Psi^{\rm b,a}(\psi_i)}\mu_i^{\rm a}(\psi_i^{\rm a}) -\mu_i^{\rm b}(\psi_i^{\rm b}).
    \end{equation*}
    Inserting this in~\eqref{eq:primal-inf-sup} yields
    \begin{align*}
        P(h)& = \inf_{\psi \in \Psi(h)}\sum_{i=1}^N \inf_{\psi_{i}^{\rm b,a}\in \Psi^{\rm b,a}(\psi_i)}\mu_i^{\rm a}(\psi_i^{\rm a}) -\mu_i^{\rm b}(\psi_i^{\rm b})\\ 
        &\geq \inf_{\psi^{\rm b,a}\in \Psi^{\rm b,a}(h)} \sum_{i=1}^N \mu_i^{\rm a}(\psi_i^{\rm a}) -\mu_i^{\rm b}(\psi_i^{\rm b}) = D(h),
    \end{align*}
    with $\Psi^{\rm b,a}(h)$ given in \eqref{eq:admissibleDual}, and now
    combining with weak duality of~\cref{lemma:weak-duality} concludes the proof of $P(h)=D(h)$. The existence of a primal optimizer again follows by compactness and semicontinuity.
\end{proof}

The following generalization is relevant for the quote-enhancement procedure in
Appendix~\ref{sec:existenceBidAsk}: as observed in \cref{rem:bid}, the
enhanced bid call-price curve need not be convex in the strike, and its
associated Breeden--Litzenberger measure is then generally signed, although
the induced pricing functional remains nonnegative on all nonnegative
convex payoffs.

\begin{remark}[Signed bid measures]\label{rem:signed-bid-measures}
The duality and attainment results of this section remain valid if some or
all of the bid measures $\mu_i^{\rm b}$ are \emph{signed} measures in the following sense. Suppose
that each $\mu_i^{\rm b}$ is a finite signed Radon measure satisfying
\[
    \int_{\R_+}(1+x)\,\lvert\mu_i^{\rm b}\rvert(\rd x)<\infty,
    \qquad
    \mu_i^{\rm b}(\R_+)=1,
    \qquad
    \int_{\R_+}x\,\mu_i^{\rm b}(\rd x)=x_0,
\]
and extend $\cleq$ by declaring, for $\nu\in\cP_1$,
\[
    \mu_i^{\rm b}\cleq\nu
    \quad\Longleftrightarrow\quad
    \mu_i^{\rm b}(\psi)\leq\nu(\psi)
    \quad\text{for all }\psi\in\textnormal{CVX}_L.
\]
For probability measures $\nu$ with barycenter $x_0$, this condition is,  similar to \eqref{eq:callCVX}, 
equivalently characterized by the corresponding inequalities for all call
payoffs. For $N>1$, \cref{asm:bidAskConvexOrder} should be replaced by the feasibility condition
\[
    \cQ(\mu^{\rm b},\mu^{\rm a})\neq\varnothing
    \quad\Longleftrightarrow\quad
    \mu_i^{\rm b}\cleq\underline{\mu}_i^{\rm a}
    \ \text{for all }i\in[N],
    \qquad
    \underline{\mu}_i^{\rm a}:=\min_{j\geq i}\mu_j^{\rm a}.
\]
The pairwise condition~\eqref{eq:bidAskConvexOrder} remains necessary, but
need not be sufficient.

We emphasize that in this extended formulation, the marginals of the calibrated martingale measures remain probability measures, like the martingale measures themselves. The signed measures act only as lower bounds for the pricing functionals.

Keeping the dual cost unchanged, the conclusions of
\cref{lemma:weak-duality,thm:band-duality,thm:strong-duality} continue to
hold in this extension. Indeed, weak duality uses only the inequalities against convex
functions. In the proof of \cref{thm:band-duality}, positivity is established
for the Hahn--Banach extension (and hence for the optimizer) without using
$\mu^{\rm b}\geq0$; all boundedness and tightness estimates use only
$\mu^{\rm a}$. Moreover, the sets
\(
    \{\nu\in\cP_1:
        \mu_i^{\rm b}\cleq\nu\cleq\mu_i^{\rm a}\}
\)
remain convex and $\cW_1$-compact, so the minimax proof of
\cref{thm:strong-duality} is unchanged.
\end{remark}

\section{Convergence to Classical MOT}\label{sec:convergence}

In this section, we study the convergence of BAMOT to classical MOT when the bid--ask spread converges to zero. That is, for each maturity, the bid and ask marginals converge to a single marginal, generating the unambiguous prices of vanillas for that maturity. The first subsection establishes consistency in the sense that the BAMOT value (and hence the optimal superhedging cost) converges to the natural limit. \Cref{sec:bid_ask_dist} introduces a novel metric, termed the \emph{bid--ask distance}, which is directly tied to the bid--ask spreads of call and put options. \Cref{sec:convergence-rate} uses this metric for a quantitative analysis of the convergence to MOT. We restrict our attention to the single-maturity case $N=1$ and establish two results: a linear rate of convergence when $h$ can be written as a difference of Lipschitz convex functions, and a square-root rate of convergence when $h \in {\rm USC}_L$ with bounded variation. The general multi-maturity setting is left for future investigation.

\subsection{Consistency}

The following result shows the continuity of the BAMOT value for a sequence of bid and ask marginals with monotonically decreasing spread. In particular, if the spread converges to zero, meaning that bid and ask marginals converge to the same limit, it establishes the consistency of BAMOT with classical MOT. Thus, if the spread is small, the superhedging price with bid--ask spread is approximately equal to the superhedging price obtained with static hedging instruments priced with the mid marginal.

\begin{proposition}\label{prop:consistency}
Let $\{(\mu_n^{\rm b}, \mu_n^{\rm a})\}_{n\in \N} \subset \cP_1^N\times \cP_1^N$ be a sequence of bid and ask marginals such that
    $$\mu_{n,i}^{\rm b} \cleq \mu_{m,i}^{\rm b} \cleq \mu_{m,i}^{\rm a} \cleq \mu_{n,i}^{\rm a}\quad \forall n\leq m, \ i\in [N].$$
    Then $\mu_{\infty}^{\rm b}=\lim_m\mu_{m}^{\rm b}$ and $\mu_{\infty}^{\rm a}=\lim_m\mu_{m}^{\rm a}$ exist and are again valid bid and ask marginals. Let $h\in \text{USC}_L(\R^N_+)$ be a payoff function that is bounded from above. Consider the corresponding BAMOT problems
\begin{equation*}%
    \begin{aligned}
            P_n(h):=\sup_{\Q \in \cQ(\mu^{\rm b}_n, \mu^{\rm a}_n)} &\E_\Q[h(X_1, \ldots, X_N)], \quad n\in\N\cup\{\infty\}.
\end{aligned}
\end{equation*}
Then, it holds that
$$
  P_n(h) \to P_\infty(h).
$$
In particular, if $\mu_\infty^{\rm b}=\mu_\infty^{\rm a}=:\mu_\infty$, then $P_n(h)$ converges to the value $\sup_{\Q \in \cQ(\mu_\infty)} \E_\Q[h(X_1, \ldots, X_N)]$ of the classical MOT problem associated with~$\mu_\infty$. 
\end{proposition}

\begin{proof}
The existence of $\lim \mu_{n}^{\rm b}$ can be seen from the associated call prices, which are monotone in~$n$ and therefore convergent. Similarly for $\mu_{n}^{\rm a}$.

Set $\cQ_n=\cQ(\mu^{\rm b}_n, \mu^{\rm a}_n)$, so that $\cQ_\infty=\cap_n\cQ(\mu^{\rm b}_n, \mu^{\rm a}_n)$. Since $\cQ_n$ decreases to $\cQ_\infty$,  clearly $\{P_n(h)\}_{n\in \N}$ is a decreasing sequence with $\lim_{n\to \infty} P_n(h) \geq  P_\infty(h)$. We show that $\lim_{n\to \infty} P_n(h)\leq P_\infty(h)$.

Suppose first that $h$ is bounded and continuous. Let $\Q_n^* \in \cQ_n$ denote a maximizer to $P_n(h)$.  As $\{\Q_n^*\}_{n \in \N}\subset \cQ_1$ and $\cQ_1$ is weakly compact, after taking a subsequence, there exists $\Q_\infty^*$ such that $\Q_n^*$ converges weakly to $\Q_\infty^*$. Then $\E_{\Q_n^*}[h] \to \E_{\Q_\infty^*}[h]$ since $h$ is bounded and continuous. Moreover, as $\{\cQ_m\}_{m\in \N}$ is decreasing, $\Q_m^* \in \cQ_m \subset \cQ_n$ for all $m\geq n$. Letting $m\to\infty$ yields $\Q_\infty^* \in \cQ_n$ as  $\cQ_n$ is weakly closed, and as $n$ was arbitrary, we conclude that $\Q_\infty^* \in \cap_{n\in \N}\cQ_n = \cQ_\infty$. In summary, $\lim_{n\to \infty}P_n(h) = \lim_{n\to \infty}\E_{\Q_n^*}[h] = \E_{\Q_\infty^*}[h] \leq  P_\infty(h)$, completing the proof when $h$ is bounded and continuous.

In the general case, there exist $h_k \in \cC_b(\R^N_+)$ decreasing to $h$, and the above shows that $P_n(h_k)$ decreases to $P_\infty(h_k)$ for every $k$. For each~$n\in\N\cup\{\infty\}$, the sublinear expectation $P_n$ satisfies $P_n(h_k) \downarrow P_n(h)$ by \citep[Theorem~31]{DenisHuPeng2011}. %
Given $\varepsilon>0$, choose $k$ such that $|P_\infty(h)-P_\infty(h_k)|<\varepsilon/2$ and then $n\in\N$ such that $|P_\infty(h_k)-P_n(h_k)|<\varepsilon/2$. By monotonicity of $P_n$, we have $P_n(h) \leq P_n(h_k) \leq P_\infty(h) + \varepsilon$, showing that $\lim_{n\to \infty} P_n(h)\leq P_\infty(h)$ as desired.
\end{proof}

\subsection{The Bid--Ask Distance}\label{sec:bid_ask_dist}

Next, we introduce a metric that is directly tied to the bid--ask spreads of call and put options. It also provides a characterization of the convex order between probability measures.

\begin{definition}\label{def:bidaskdist}
Given $\mu,\nu \in \cP_1$, introduce the directed distance   
\begin{equation*}\label{eq:directedDist}
    \vec{d}(\mu,\nu) = \sup \left\{(\mu-\nu)(\psi) \,\mid \, \psi \text{ convex, } \textnormal{Lip}(\psi) \le 1 \right\},
\end{equation*}
and the \emph{bid--ask distance} given by the symmetrization   
\begin{equation*}
    d(\mu,\nu) = \frac{\vec{d}(\mu,\nu) + \vec{d}(\nu,\mu)}{2}.
\end{equation*}
\end{definition}
Recently, \cite{WieselZhang2023} (see also \cite{acciaio}) proposed an alternative characterization of  convex order in the $2$-Wasserstein space. Their criterion relies on maximal covariances, a formulation that does not directly translate to the bid--ask interpretation needed for our analysis.   The bid--ask distance also relates to the relaxed $L^\infty$ metric introduced in \cite{SonerTissot}, which is constructed from an asymmetric distance to compare semicontinuous functions in a nearly uniform manner. The following proposition collects key properties of $\vec{d}$ and $d$. 
\begin{proposition}\label{prop:bidAskDistance}
Let $\cW_1$ denote the 1-Wasserstein distance. Then the following hold. 
    \begin{enumerate}[label = (\roman*)]

        \item  $ \vec{d}(\mu,\nu) = 0$ if and only if $\mu \cleq \nu$.

              \item $\vec{d}$ is an asymmetric distance \cite{mennucci}; it is nonnegative, satisfies the triangle inequality, and separates points in the sense that $\vec{d}(\mu,\nu) = \vec{d}(\nu,\mu) = 0 $ implies $\mu = \nu$.  
      Consequently, $d$ is a metric.

\item  If $\mu,\nu$ have identical barycenters, then 
$$ \vec{d}(\mu,\nu) = 2\sup_{K\geq 0} \left(\E_{\mu}[(X-K)^+]-\E_{\nu}[(X-K)^+]\right).$$

        \item We have $d(\mu,\nu) \le \cW_1(\mu,\nu)$ for all  $\mu,\nu \in \cP_1$. Moreover, there exist sequences  $(\mu_n),(\nu_n)$ in $\cP_1$ such that 
        $$\lim_{n\to \infty} d(\mu_n,\nu_n) = 0, \;  \text{ while } \cW_1(\mu_n,\nu_n) = 1 \quad \forall n \in \N.$$ 

        \item The metrics $d$, $\cW_1$ induce the same topology on $\cP_1$.

    \end{enumerate}
\end{proposition}

\begin{proof}
    See \cref{proof:bidAskProp}.
\end{proof}

\begin{remark}\label{rmk:bid-ask-distance}
\begin{enumerate}[label=(\alph*)]
\item Property~(iii) highlights the connection between the proposed distance and the bid--ask spreads of quoted vanilla options; if $c^{\rm s}(K,T) = \E_{\mu^{\rm s}}[(X-K)^+]$, $\rm s \in \{\rm b, \rm a\}$, denotes the  price of the $(K,T)$ call option on each side of market, then 
\begin{equation}\label{eq:bidaskDistanceExample}
    d(\mu^{\rm b},\mu^{\rm a}) = \frac{1}{2} \vec{d}(\mu^{\rm a},\mu^{\rm b}) = \sup_{K \ge 0} \left(c^{\rm a}(K,T) - c^{\rm b}(K,T)\right),
\end{equation}
where the first equality follows from property~(i). 
As each spread in \eqref{eq:bidaskDistanceExample} represents the (round-trip) transaction cost of trading these options, the distance indeed  measures the level of bid--ask friction in the market.
\item The bid--ask distance $d$ is equivalent to the following \emph{integral probability metric}, %
\begin{equation}\label{eq:dualnorm}
\tilde{d}(\mu,\nu)
= \sup \big\{(\mu-\nu)(\psi) \, \mid \, \psi=\psi^{\rm a}-\psi^{\rm b}\big\},
\end{equation}
where $\psi^{\rm a}$ and $\psi^{\rm b}$ are convex and $1$-Lipschitz. Indeed, observe that $\vec{d}(\mu,\nu)\le \tilde{d}(\mu,\nu)$, so $d(\mu,\nu)\le \tilde{d}(\mu,\nu)$ as well. On the other hand, for any test function $\psi=\psi^{\rm a}-\psi^{\rm b}$ in \eqref{eq:dualnorm},
\[
(\mu-\nu)(\psi)
= (\mu-\nu)(\psi^{\rm a})+(\nu-\mu)(\psi^{\rm b})
\le \vec{d}(\mu,\nu)+\vec{d}(\nu,\mu)
= 2\,d(\mu,\nu).
\]
\item The directed distance $\vec{d}$ is related to weak optimal transport through the convex Kantorovich--Rubinstein duality formula (see, e.g., \cite{beiglboeck2025}),
\begin{equation}\label{eq:CvxKantRubin}
\sup_{\psi \text{ convex},\ \textnormal{Lip}(\psi)\le 1} (\mu-\nu)(\psi)
= \inf_{\PP \in \Pi(\mu,\nu)} \E_{\PP}\big[|\E_{\PP}[X_1\mid X_0]-X_0|\big],
\end{equation}
where $(X_0,X_1)\sim \PP$ and $\Pi(\mu,\nu)$ denotes the set of all couplings between $\mu$ and $\nu$. The right-hand side of \eqref{eq:CvxKantRubin} thus provides an alternative characterization of $\vec{d}(\mu,\nu)$. As observed in~\cite{beiglboeck2025}, \eqref{eq:CvxKantRubin} is  a quantitative analogue of Strassen's theorem; the left-hand side quantifies the extent to which convex order is violated, while the dual measures the deviation from the martingale property.
\end{enumerate}
\end{remark}

\subsection{Convergence Rate}\label{sec:convergence-rate}

Next, we focus on the single-maturity case $N=1$ and show how the bid--ask distance controls the discrepancy between the BAMOT price and its frictionless counterpart. 
When $h$ is given by the difference of convex functions, then the rate is linear as shown next. In particular, the result applies to all multi-leg option strategies  $h(x) = \int_{\R_+}(x-K)^{+}\lambda(\rd K)$ such as bull, bear, and butterfly spreads. 
The following results are presented from the primal perspective \eqref{eq:primal1M}, noting that similar rates apply to the dual value by strong duality.

\begin{proposition}\label{prop:MidEstimate}
Let $\mu^{\rm b} \cleq \mu^{\rm a}$ and denote the mid marginal by $\mu^{\rm m} = \frac{\mu^{\rm b} + \mu^{\rm a}}{2} $. 
Suppose that $h = \psi^{\rm a} - \psi^{\rm b}$ for  some Lipschitz pair $(\psi^{\rm a}, \psi^{\rm b})\in \textnormal{CVX}_L^2$. Then, 
\begin{equation*}\label{eq:midRate1}
    0\le P(h) - \mu^{\rm m}(h) \le  C \hspace{0.3mm}  d(\mu^{\rm b},\mu^{\rm a}),
\end{equation*}
where $C = \textnormal{Lip}(\psi^{\rm a}) + \textnormal{Lip}(\psi^{\rm b})$.
\end{proposition}

\begin{proof} 
The lower bound follows as $\mu^{\rm m} \in \cQ(\mu^{\rm b}, \mu^{\rm a})$. On the other hand, from the definitions of $\mu^{\rm m}$ and $\vec{d}$, we obtain after rearranging that 
\begin{equation}\label{eq:mid-estimate}
    \begin{aligned}
     P(h) - \mu^{\rm m}(h) \le \frac{1}{2}(\mu^{\rm a}-\mu^{\rm b})(\psi^{\rm a}) + \frac{1}{2}(\mu^{\rm a}-\mu^{\rm b})(\psi^{\rm b}) 
    \le \frac{1}{2}(\textnormal{Lip}(\psi^{\rm a}) + \textnormal{Lip}(\psi^{\rm b}))\vec{d}(\mu^{\rm a},\mu^{\rm b}).
\end{aligned}
\end{equation}
Since $\vec{d}(\mu^{\rm b},\mu^{\rm a}) =0$ by $\mu^{\rm b} \cleq \mu^{\rm a}$ and \cref{prop:bidAskDistance}~(i), we know that $\frac{1}{2}\vec{d}(\mu^{\rm a},\mu^{\rm b}) = d(\mu^{\rm a},\mu^{\rm b})$, which, together with \eqref{eq:mid-estimate}, concludes the upper bound.
\end{proof}

If $h$ is upper semicontinuous with bounded variation, as in the case of digital options, the convergence rate is typically square-root, as shown next. We will see in \cref{sec:convDigital} that this rate is in fact numerically sharp.
\begin{proposition}\label{prop:USC_sqrt_rate}
Let $\mu^{\rm b} \cleq \mu^{\rm a}$ and assume that the mid-marginal $\mu^{\rm m}=\frac{\mu^{\rm b}+\mu^{\rm a}}{2}$ admits a density $\phi^{\rm m}$ that is bounded by $M>0$. 
Let $h\in{\rm USC}_L$ have bounded variation $V(h)<\infty$. 
Then, 
\[
0 \le P(h)-\mu^{\rm m}(h) \le C\sqrt{d(\mu^{\rm b},\mu^{\rm a})},
\]
where $C = \sqrt{2M} V(h)$.
\end{proposition}

\begin{proof} 
Fix $\mu\in\cQ(\mu^{\rm b},\mu^{\rm a})$ and denote its call price function by $c(K)=\E_{\mu}[(X-K)^+]$.  Define also $c^{\rm s}(K)=\E_{\mu^{\rm s}}[(X-K)^+]$, ${\rm s}\in\{{\rm b},{\rm m},{\rm a}\}$. For technical reasons, we extend $c$ to $\R$ by setting $c(K)= \E_{\mu}[X] - K$ when $K<0$.
Since $\mu\in\cQ(\mu^{\rm b},\mu^{\rm a})$, \cref{prop:bidAskDistance}~(iii) implies that
\[
|c(K)-c^{\rm m}(K)|
\le \tfrac12\big(c^{\rm a}(K)-c^{\rm b}(K)\big)
\le \frac{1}{2} d(\mu^{\rm b},\mu^{\rm a}), \quad\, \forall K\geq 0.
\]
In addition, $c(K)-c^{\rm m}(K)=0$ for any $K<0$ because $\E_{\mu}[X]=\E_{\mu^{\rm m}}[X]$ and both $\mu$ and $\mu^{\rm m}$ are supported on $\R_+$. Thus,
\[
\sup_{K\in\R}|c(K)-c^{\rm m}(K)| \le \frac{1}{2} d(\mu^{\rm b},\mu^{\rm a}).
\]

Moreover, since $\phi^{\rm m}$ is bounded by $M$, we have $c^{\rm m}\in\cC^{1,1}(\R)$ and $\partial_K c^{\rm m}$ is $M$-Lipschitz. Applying \cref{lem:CVXEstimate} to the convex functions $\psi=c^{\rm m}$, $\varphi=c$, and $\varepsilon=\frac{1}{2} d(\mu^{\rm b},\mu^{\rm a})$ on $\R$, we obtain
\begin{equation}\label{eq:deriv_bound}
\sup_{K\ge 0}\big|\partial_K^{+}c(K)-\partial_K c^{\rm m}(K)\big|
\le \sqrt{2M\,d(\mu^{\rm b},\mu^{\rm a})},
\end{equation}
where $\partial_K^{+}c$ denotes the right derivative of $c$. Recall the identities $\partial_K^{+}c=\Phi-1$ and $\partial_K c^{\rm m}=\Phi^{\rm m}-1$, where $\Phi$ and $\Phi^{\rm m}$ denote the distribution functions of $\mu$ and $\mu^{\rm m}$, respectively~(see, e.g., ~\cite{breeden1978prices}, \cite[Lemma~1]{hobson2010skorokhod}). Together with~\eqref{eq:deriv_bound}, this yields
\begin{equation}\label{eq:CDF_bound}
\sup_{K\ge 0}|\Phi(K)-\Phi^{\rm m}(K)|
\le\sqrt{2M\,d(\mu^{\rm b},\mu^{\rm a})}.
\end{equation}
Applying~\cref{lem:BV-ibp} with $\nu:=\mu-\mu^{\rm m}$ and defining $G_\nu(K):=\Phi(K)-\Phi^{\rm m}(K)$ for $K\geq 0$, then
\[
|(\mu-\mu^{\rm m})(h)|
=|\nu(h)|
\le V(h)\max\left\{
\sup_{K\ge0}|G_\nu(K)|,\,
\sup_{K\ge0}|G_\nu(K^-)|
\right\}.
\]
Combining this estimate with~\eqref{eq:CDF_bound} and taking the supremum over $\mu\in\cQ(\mu^{\rm b},\mu^{\rm a})$ proves the result.

\end{proof}

\section{Examples}\label{sec:numerics}

This section presents a collection of numerical examples involving multiple illiquid claims in the market. In \cref{sec:digital_options}, we focus on digital call options and present a closed-form solution in a one-sided market. We further validate our theoretical findings via a numerical experiment using options data on the S\&P~500 Index (SPX).\footnote{Data retrieved from Bloomberg.}
In \cref{sec:forward-start}, we compare the super- and subhedging prices obtained under BAMOT with those derived using classical MOT. Finally, in \cref{sec:convDigital}, we demonstrate the convergence behavior as the bid--ask spread shrinks for a risk-reversal payoff and an at-the-money digital call option.  Unless stated otherwise, all numerical solutions are obtained via a discretization scheme (see~\cref{app:LP} for details) and solved using linear programming. We refer to the resulting numerical solutions as the primal and dual optimizers. However, it should be emphasized that a dual optimizer for the original, undiscretized BAMOT problem need not exist.

\subsection{Digital Options}\label{sec:digital_options}

This section is devoted to digital options. Beyond their long-standing use in fixed income, foreign exchange, and commodities markets, their popularity has surged recently on prediction market platforms such as Polymarket and Kalshi. We first study a special case in which the market is one-sided, that is, bid quotes are absent, leading to explicit primal and dual optimizers. We then present numerical results for a digital call option written on the S\&P~500 Index (SPX) using real market data, from which we observe that pricing digital options using mid marginals can substantially underestimate the superhedging prices in the presence of bid--ask frictions.

Our analysis relies on the concept of \emph{local concentration}. Given $\mu \in \cP_1$, we can construct another probability measure $\mu'$ such that $\mu' \cleq \mu$ as follows. Let $I_1,\ldots,I_n$ be a partition of $\R_+$ into disjoint intervals $\{I_i\}$ such that $w_i := \mu(I_i) > 0$ for each $i$. Let $x_i = \frac{1}{w_i}\int_{I_i} x \, \mu(\rd x)$ denote the barycenter of $\mu$ on $I_i$. Then $\mu' = \sum_{i=1}^n w_i \delta_{x_i} \in \cP_1$, and $\mu' \cleq \mu$.

\subsubsection{Digital Option in a One-Sided Market}\label{sec:illiquidVanillas}
We consider a one-sided option market with a single maturity $T$ and no available bid quotes. Hence  $\mu^{\rm b} = \delta_{x_0}$, where we choose $x_0=1$, while the ask marginal distribution is an arbitrary atomless measure with mean $x_0$.
We aim to find the %
superhedging price and construct an optimal measure for
an out-of-the-money (OTM) digital call option with payoff $h(x) = \mathds{1}_{\{x \ge K\}}$, where $K > x_0$. Then, \cref{thm:band-duality} becomes
\begin{equation}\label{eq:digitalDuality}
    P(h) = \sup_{\mu \in \underline{\cQ}(\mu^{\rm a})} \mu([K,\infty)) = \inf_{\psi^{\rm b,a}\in \Psi^{\rm b,a}(h)} \mu^{\rm a}(\psi^{\rm a}) - \psi^{\rm b}(x_0) = D(h).
\end{equation}
It turns out that both primal and dual optimizers of \eqref{eq:digitalDuality} can be found explicitly. 
Indeed, introduce $b^{\rm a}(L) := \int_{L}^{\infty} (x-K) \mu^{\rm a}(\rd x)$. Since $b^{\rm a}(0) = x_0 - K<0$, $b^{\rm a}(K) \geq 0$ and by the  continuity of $L\mapsto b^{\rm a}(L)$, there exists a critical strike $L^{\star} \le K$ such that $b^{\rm a}(L^\star)=0$. For our purposes,  we assume that $L^{\star} < K$ as otherwise $\mu^{\rm a}$ is supported on $[0,K]$, in which case the problem is trivial.

\begin{proposition} 
Introduce   the convex profile  $\psi_L(x):=\frac{1}{K-L}(x-L)^+$. Then, 
\begin{enumerate}[label=(\roman*)]
    \item $\mu^\star :=\mu^{\rm a}\big|_{[0,L^\star]} + \mu^{\rm a}([L^\star,\infty))\delta_K  \in \underline{\cQ}(\mu^{\rm a}) $ is a primal  optimizer of \eqref{eq:digitalDuality}. %
    \item $(\psi_{L^\star},0) \in \Psi^{\rm b,a}(h)$ is a dual  optimizer of \eqref{eq:digitalDuality}. %
    \item $P(h)=D(h)=\frac{c^{\rm a}(L^\star,T)}{K-L^\star} = \mu^{\rm a}([L^\star,\infty))$.
\end{enumerate}
\end{proposition}
\begin{proof}
By local concentration,  $\mu^\star\in \cQ(\delta_{x_0},\mu^{\rm a})=\underline{\cQ}(\mu^{\rm a})$. 
It is also clear that $(\psi_{L^\star},0)\in \Psi^{\rm b,a}(h)$, as illustrated in the right panel of \cref{fig:digitalOneSidedHedge}. Then by definition of $\mu^{\star}$, $L^{\star}$,  and $\psi_{L}$, 
\[
D(h)\leq \mu^{\rm a}(\psi_{L^\star})
= \frac{1}{K-L^\star}\int_{L^{\star}}^\infty(x-L^\star)\,\mu^{\rm a}(\rd x)
= \mu^{\rm a}([L^\star,\infty)) = \mu^{\star}(h) 
\leq P(h).
\]
 The primal and dual optimality follow by weak duality.
\end{proof}

\paragraph{Primal-Dual Relations} 
Let $\sigma^\star(\cdot , T)$ denote the implied volatility skew  of the optimal measure~$\mu^\star$. Then, the implied volatility at the optimal strike $L^{\star}$ satisfies 
\begin{equation}\sigma^\star(L^\star, T) = \sigma^{\rm a}(L^\star, T).
\end{equation}
Indeed, the price of the $(L^\star, T)$ call under $\mu^\star$ is $c^\star(L^\star, T) = \mu^\star(h)(K - L^\star) = c^{\rm a}(L^\star, T)$. Similarly, observe that $c^\star(K, T) = c^{\rm b}(K, T) =(x_0-K)^+ = 0$, since the optimal measure is supported on $[0,K]$. Equivalently,  $\sigma^\star(K, T) = \sigma^{\rm b}(K, T) = 0$, as shown in the following example. 

\begin{example}\label{eq:digitalBS}
Let $T=1/12$ (one month), $K=1.05$, and let $\mu^{\rm a}$ be the one-month marginal distribution in the Black--Scholes model with volatility $\sigma^{\rm a} = 20\%$ and zero interest rates. %
Then the superhedging price is approximately equal to $D(h) = 0.46$, with optimal lower strike $L^\star \approx 1.004$. \cref{fig:digitalOneSidedHedge} displays the cost function $L\mapsto \mu^{\rm a}(\psi_L)$ and the optimal superhedging portfolio. The superhedging price is significantly higher than the  price of the digital under the ask marginal, roughly equal to $0.19$. 
As shown in \cref{fig:digitalOneSidedOptimalMu}, %
different optimal measures can be constructed by varying the distribution on $[0, L^\star]$ using local concentration. %
\end{example}

\begin{figure}[!ht]
\caption{Superhedging of a digital call in a one-sided market. }
    \begin{subfigure}[b]{0.49\linewidth}
        \centering
          \caption{Super-replication cost as function of $L$}
        \includegraphics[width=2.5in, height = 1.75in]{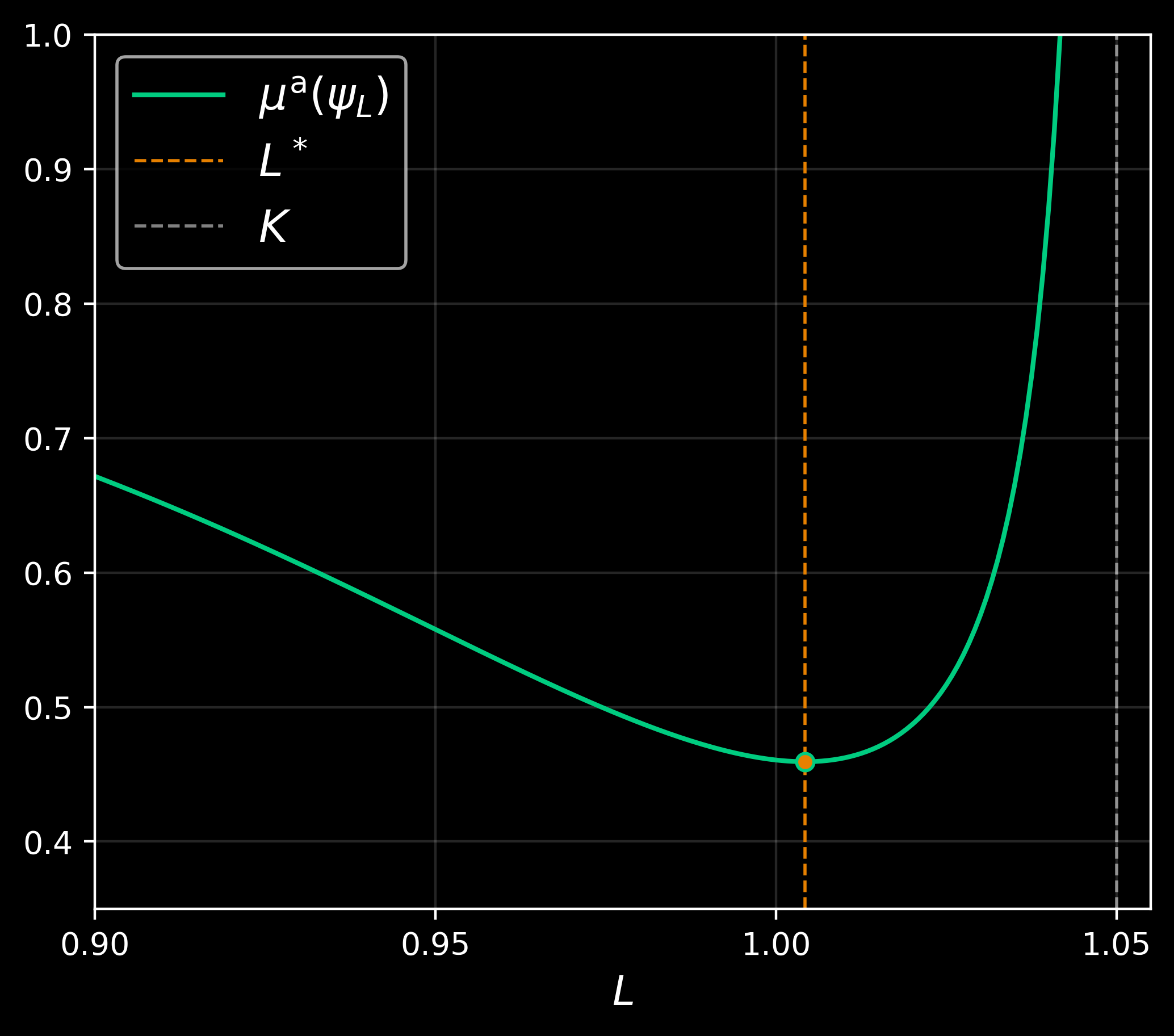}
    \end{subfigure}
        \begin{subfigure}[b]{0.49\linewidth}
        \centering
           \caption{Optimal superhedging portfolio}
        \includegraphics[width=2.5in, height = 1.75in]{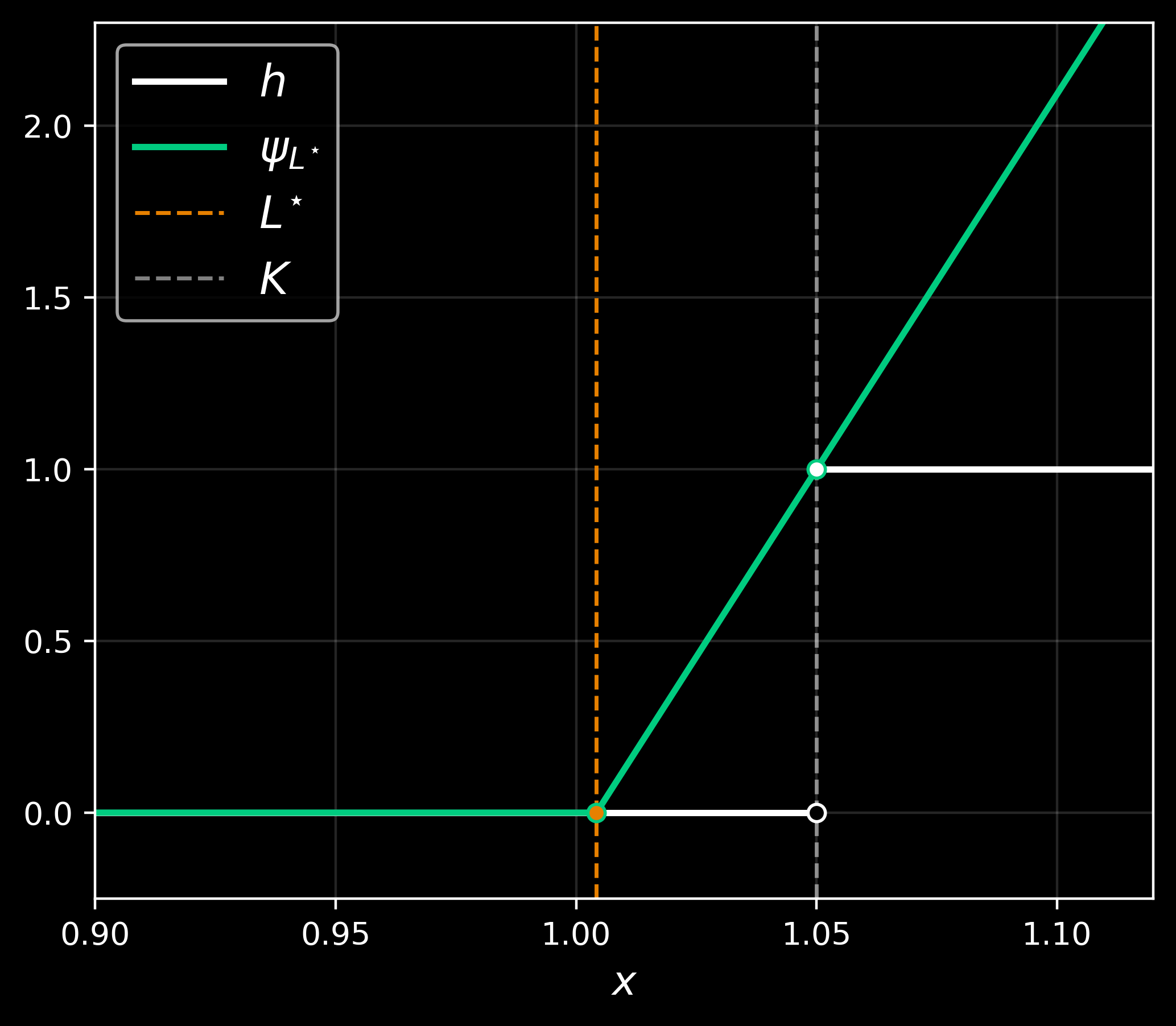}
    \end{subfigure}
        \label{fig:digitalOneSidedHedge}
\end{figure}

\begin{figure}[!ht]
\caption{Digital call in a one-sided market. Illustrations of different optimal measures (left) and corresponding implied volatility skews (right).}
    \begin{subfigure}[b]{0.49\linewidth}
        \centering
        \includegraphics[width=2.5in, height = 1.75in]{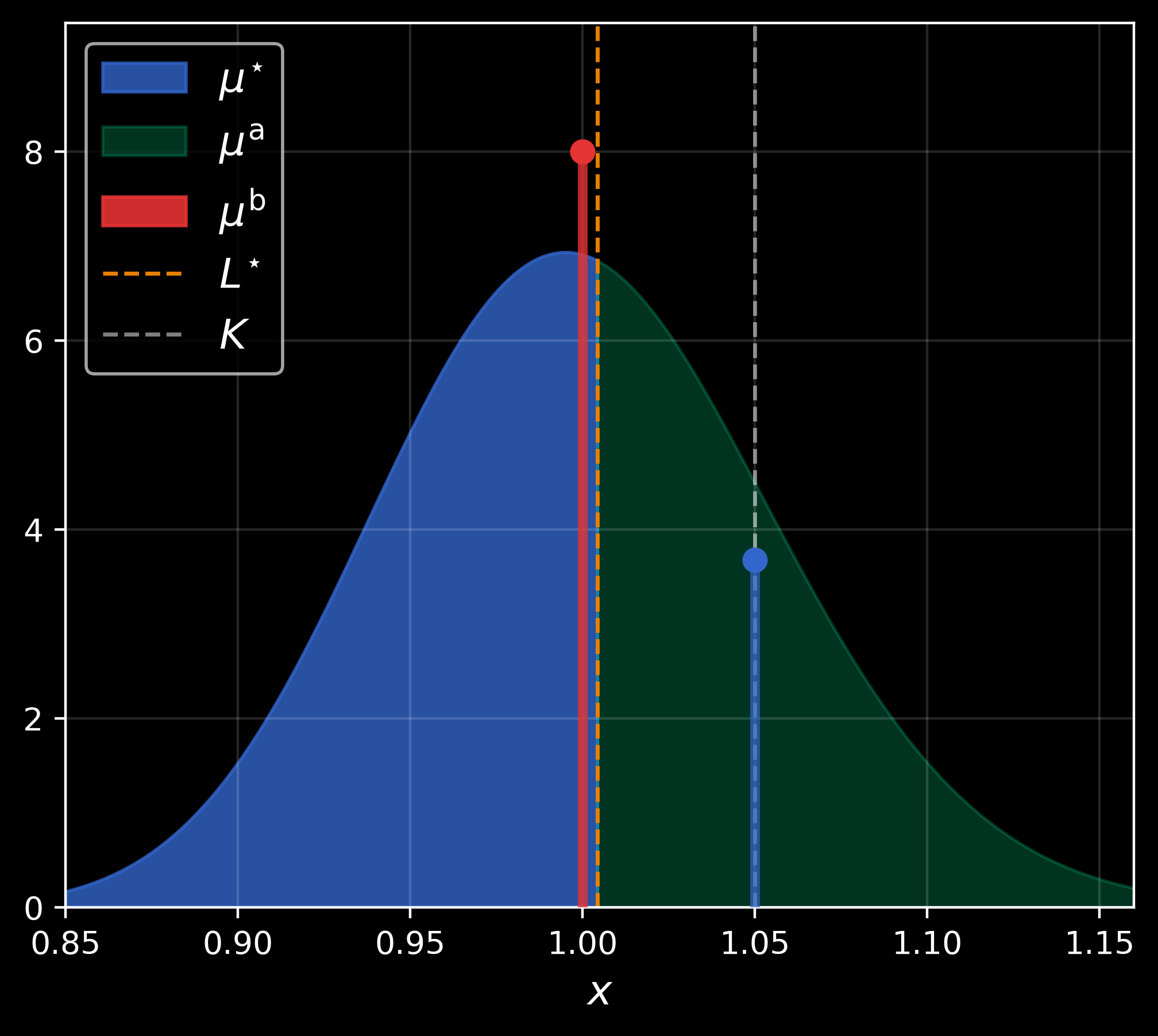}
    \end{subfigure}
    \begin{subfigure}[b]{0.49\linewidth}
        \centering
        \includegraphics[width=2.5in, height = 1.75in]{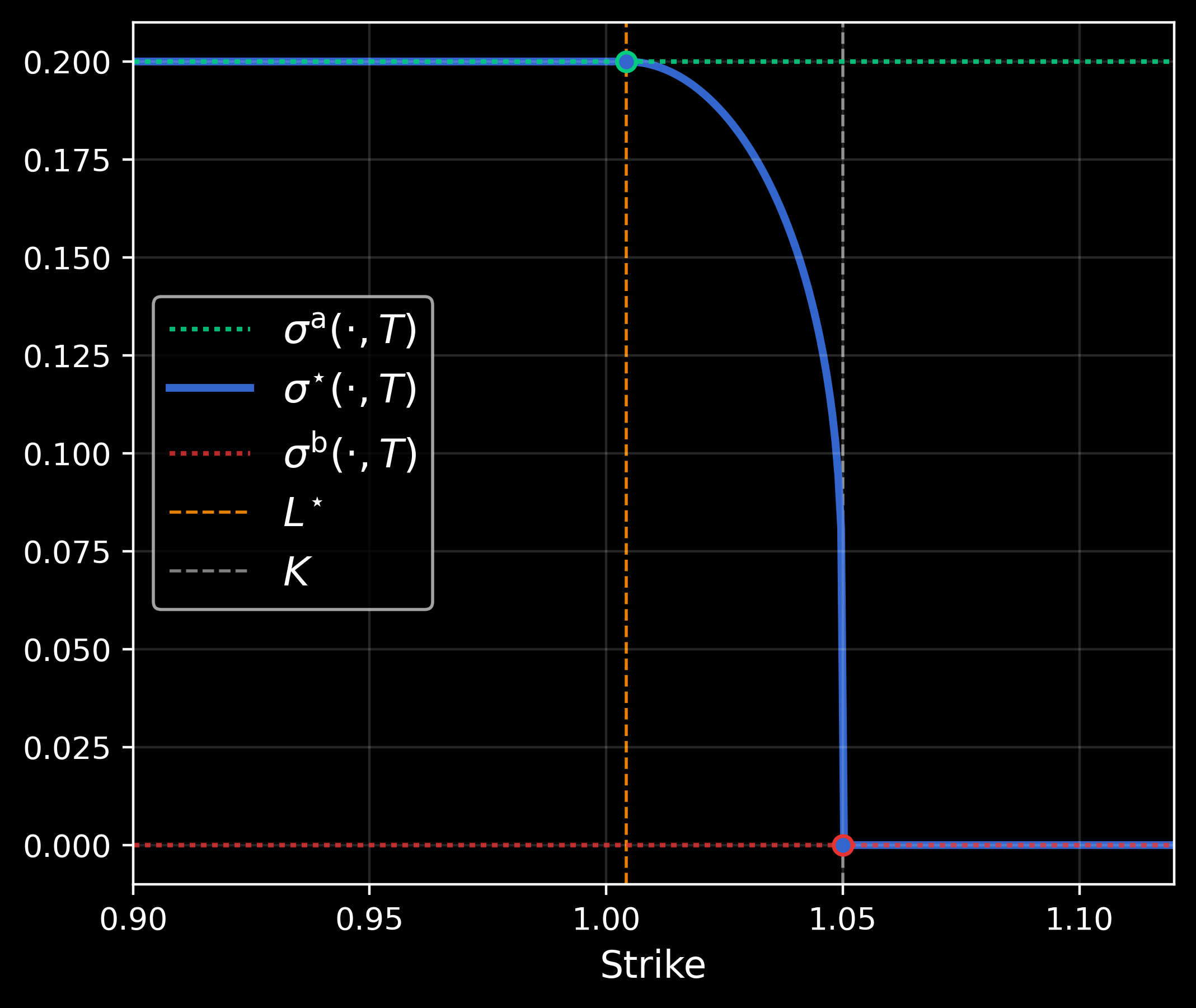}
    \end{subfigure}
    \vspace{3mm}
    
        \begin{subfigure}[b]{0.49\linewidth}
        \centering
        \includegraphics[width=2.5in, height = 1.75in]{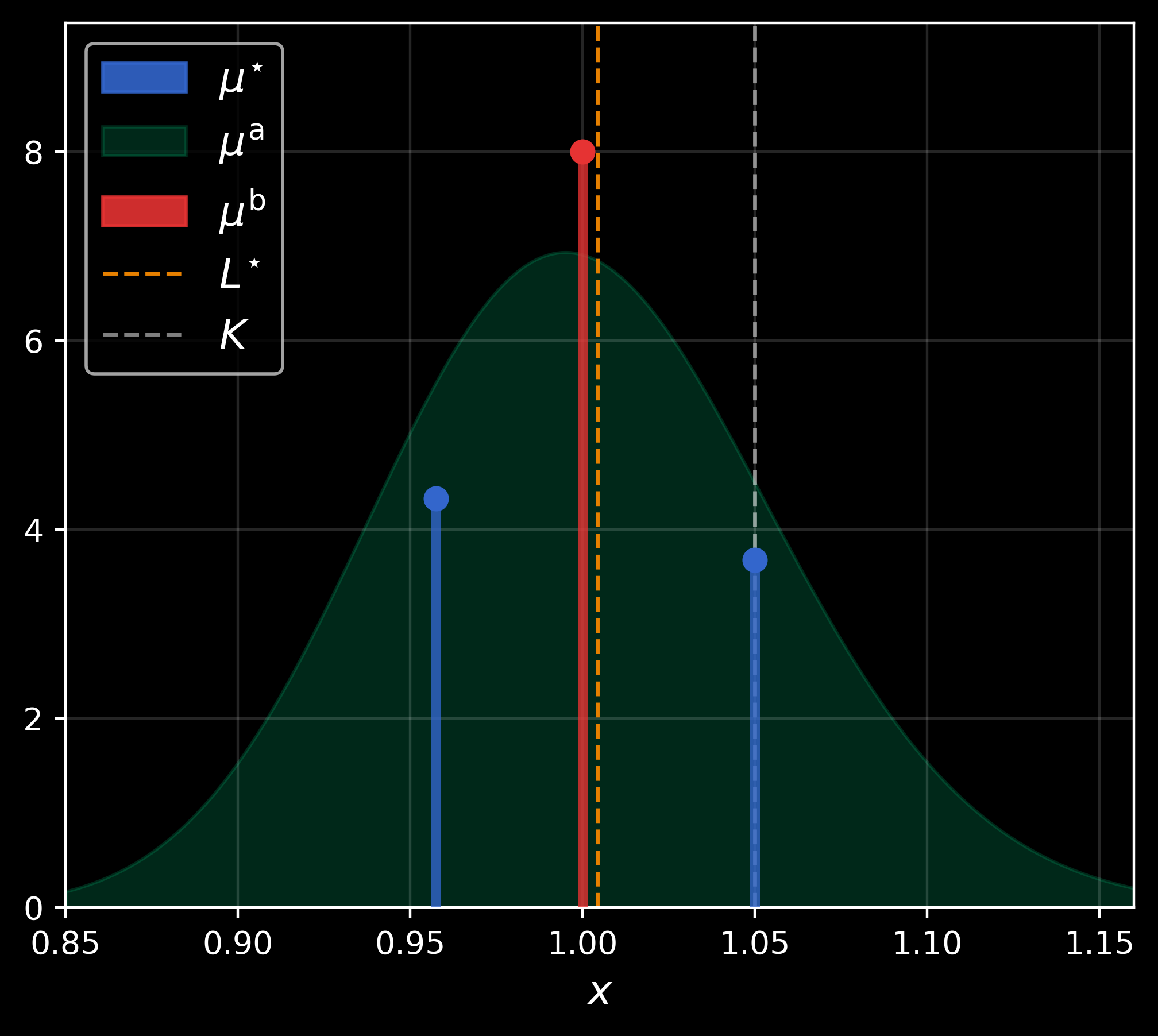}
    \end{subfigure}
    \begin{subfigure}[b]{0.49\linewidth}
        \centering
        \includegraphics[width=2.5in, height = 1.75in]{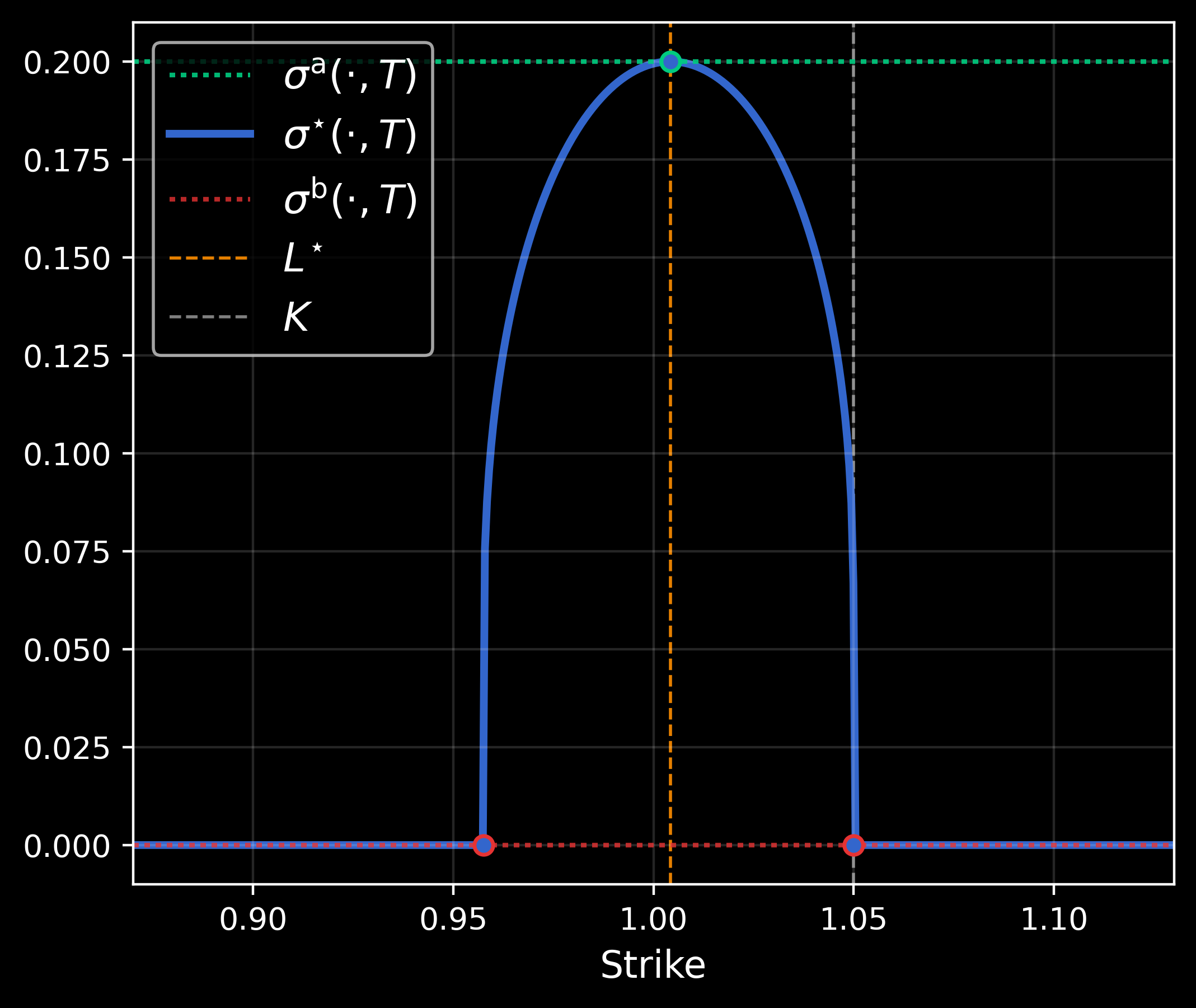}
    \end{subfigure}
    \label{fig:digitalOneSidedOptimalMu}
\end{figure}

\paragraph{Connection with One-touch Options.}
The present example can be seen as a static analogue of Hobson’s robust hedging of one-touch (American digital) options \cite{hobson2010skorokhod}. Indeed, consider the discrete optimal measure $\mu^{\star}$ displayed in the bottom left chart of \cref{fig:digitalOneSidedOptimalMu}. 
In view of $\delta_{x_0} = \mu^{\rm b}\cleq \mu^{\star} \cleq \mu^{\rm a}$, there exists a continuous-time càdlàg martingale $(X_t)_{t\in [0,1]}$ such that  $X_0 = x_0$, which jumps to $X_{1/2}\sim \mu^{\star}$ at $t=1/2$ and to $X_1\sim \mu^{\rm a}$ at the terminal time.  
This is precisely the model constructed in \cite{hobson2010skorokhod} that achieves the superhedging price for the one-touch call option. 
Specifically, we have
\begin{equation}\label{eq:HobsonParallel}
\sup_{\Q, \, X_T\sim  \mu^{\rm a}} \Q(X_\tau \ge K)
= \sup_{\mu \cleq \mu^{\rm a}} \mu([K, \infty)),
\end{equation}
where $\tau$ is the first time $X$ exceeds the strike $K$ (equal to $\tau = 1$ if $X_t< K$ for all $t$), and 
 $\Q$ is a continuous-time martingale measure. In fact, the identity  still holds if the marginal constraint on the left-hand side is relaxed to $\text{Law}^{\Q}(X_T)\cleq \mu^{\rm a}$. 
 This is because Hobson’s optimal superhedging strategy utilizes only long positions in call options alongside a dynamic hedge; since no options are sold, the strategy is independent of the bid marginal. 
As the right-hand side corresponds to a simpler problem, our framework also provides a useful compression tool for some path-dependent robust pricing problems. In addition, by the optional sampling theorem and Strassen's theorem, \eqref{eq:HobsonParallel} can be extended to connect our framework with the robust pricing of American options,
\begin{equation*}
\sup_{\tau \le T}\sup_{\Q,\, X_T\sim \mu^{\rm a}} \E_{\Q}[h(X_\tau)]
= \sup_{\mu \cleq \mu^{\rm a}} \mu(h),
\end{equation*}
where $h\in \text{USC}_L$ and $\tau$ ranges over all stopping times taking values in $[0,T]$.

\subsubsection{Digital Option on the S\&P~500 Index (SPX)}\label{subsec:digtial_spx}
We now compute the superhedging price of an out-of-the-money digital call using real market data. The bid and ask marginals are calibrated to raw quotes from the option chains of SPX put and call options quoted on 2025-02-27 with maturity on 2025-07-18, using four-component log-normal mixture models. The parameters are fitted using the discounted forward price $x_0=5861$ and are reported in \cref{tab:bid_ask_lognormal}. Observe that the mixture weights and component means are shared across the bid and ask marginals, while the component volatilities are calibrated so as to preserve convex order as mentioned in \cref{sec:bid-ask-marginals}.

\begin{figure}[!htbp]
    \centering
        \caption{A dual optimizer, the potential function, and implied volatility skew of a primal optimizer for the discretized problem} with digital payoff $h(x)=100\mathds{1}_{\{x\geq 1.05x_0\}}$ on the S\&P~500 Index (SPX).
    \begin{subfigure}[t]{0.32\linewidth}
        \centering
           \caption{Dual optimizer}
        \includegraphics[width=\linewidth]{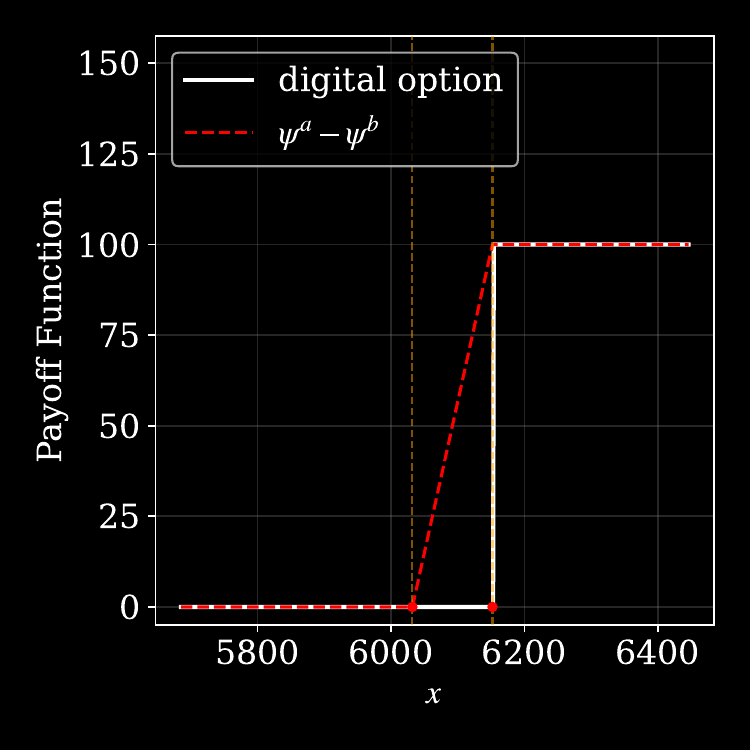}
        \label{fig:spx_dual}
    \end{subfigure}
    \hfill
    \begin{subfigure}[t]{0.32\linewidth}
        \centering
        \caption{Potential function}
        \includegraphics[width=\linewidth]{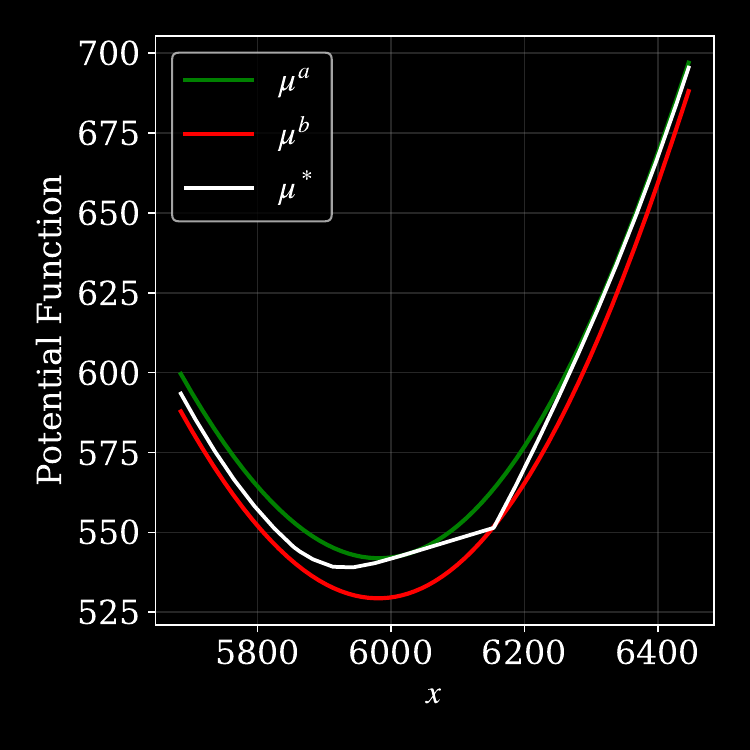}
        \label{fig:spx_primal}
    \end{subfigure}
    \hfill
    \begin{subfigure}[t]{0.32\linewidth}
        \centering
        \caption{Implied volatility}
        \includegraphics[width=\linewidth]{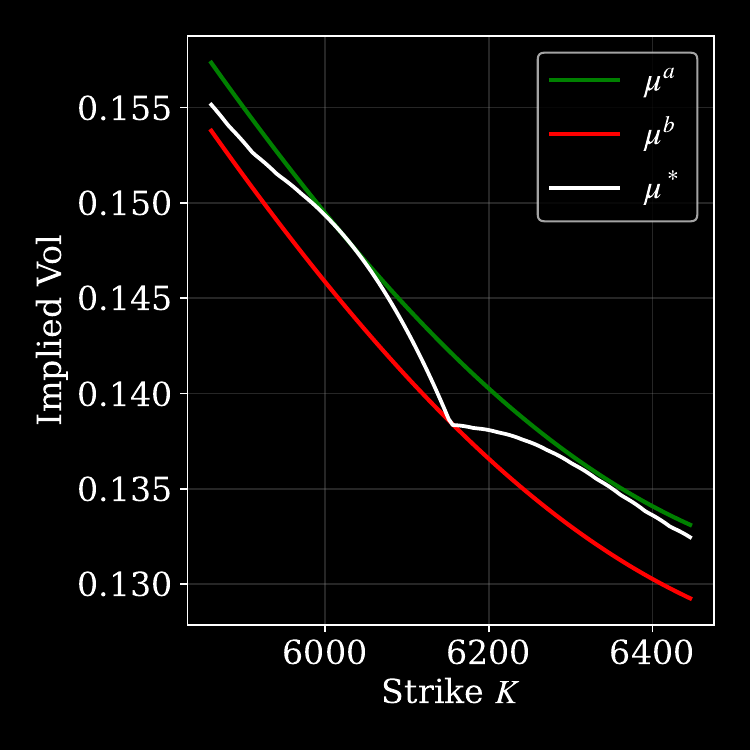}
        \label{fig:spx_iv}
    \end{subfigure}
    \label{fig:digital_spx_combined}
\end{figure}

\begin{table}[ht]
\centering
\caption{Log-normal mixture parameters for $\mu^{\rm a}$ and $\mu^{\rm b}$ (values are kept up to 4 significant figures).}
\begin{tabular}{c|ccc|ccc}
\hline
 & \multicolumn{3}{c|}{\textbf{Bid}} & \multicolumn{3}{c}{\textbf{Ask}} \\
\cline{2-7}
\textbf{Index} 
& \textbf{means} & \textbf{vols} & \textbf{weights}
& \textbf{means} & \textbf{vols} & \textbf{weights} \\
\hline
1 & 6250 & 0.04008 & 0.09591 & 6250 & 0.04009 & 0.09591 \\
2 & 6098 & 0.07222 & 0.5814  & 6098 & 0.07408 & 0.5814  \\
3 & 5531 & 0.1293  & 0.2741  & 5531 & 0.1360  & 0.2741  \\
4 & 4116 & 0.3150  & 0.04852 & 4116 & 0.3198  & 0.04852 \\
\hline
\end{tabular}

\label{tab:bid_ask_lognormal}
\end{table}

\Cref{fig:digital_spx_combined} presents a dual optimizer, the potential function and the implied volatility skew of a primal optimizer for the discretized BAMOT problem with $N=1$ and payoff $h(x) = 100\,\mathds{1}_{\{x \geq K\}}$ with $K = 1.05 x_0$. Recall that the potential function of a probability measure $\mu \in \cP_1$ is defined as $U_\mu(x) := \int |x-y| \, \mu(\rd y)$, and that $U_\mu \leq U_\nu$ if $\mu \cleq \nu$. 

The computed superhedging price is $46.84$, while the price obtained under the mid marginal is $37.14$, resulting in a premium of $P(h)-\mu^{\rm m}(h)=9.70$. This suggests that the common practice of pricing digital options using the mid marginal can substantially underestimate prices relative to BAMOT in the presence of bid--ask frictions.

The dual optimizer shown in \cref{fig:spx_dual} is a call spread with strikes $L^\star= 6032$ and $K = 1.05 x_0 = 6154.05$, yielding an optimal superhedging portfolio in the market. 
On the other hand, \cref{fig:spx_primal} shows that the potential function of the numerical optimizer $\mu^\star$ lies between the potential functions of the bid and ask marginals up to a small discretization residual. Over the interval $[0,2x_0]$, the minimum signed residuals are $\min_{x\in[0,2x_0]}
\{U_{\mu^{\rm a}}(x)-U_{\mu^\star}(x)\}
= -0.1136$ and $\min_{x\in[0,2x_0]}
\{U_{\mu^\star}(x)-U_{\mu^{\rm b}}(x)\}
= -0.3529$.
Thus the observed violations are small relative to the scale of the plotted potentials, and we interpret the plots as numerical evidence for feasibility of the solution.
In addition, the potential function associated with $\mu^{\star}$ touches those of the ask and bid marginals at $L^\star, K$, respectively. This behavior is consistent with the structure of the digital call payoff, which is a step function that vanishes below the strike. Finally, \cref{fig:spx_iv} illustrates that the implied volatility skew of the primal optimizer lies within the region delineated by the bid and ask marginals and touches the bid and ask implied volatility skews at $L^\star, K$.
To understand this, note that the potential function of $\mu^{\star}$ in \cref{fig:spx_primal} is linear between the lower strike $L^{\star}$ and $K$, implying that $\mu^{\star}((L^{\star},K))=0$. Hence, the digital call payoff and the dual optimizer $\psi^{\star}(x)=100\frac{(x-L^{\star})^+-(x-K)^+}{K-L^{\star}}$ have the same price under $\mu^{\star}$. Therefore,
\[
P(h)=\mu^{\star}(h)=\mu^{\star}(\psi^{\star})
=100\times\frac{c^{\star}(L^{\star},T)-c^{\star}(K,T)}{K-L^{\star}}
\le 100\times\frac{c^{\rm a}(L^{\star},T)-c^{\rm b}(K,T)}{K-L^{\star}}
= D(h).
\]
By strong duality, we conclude that $c^{\star}(L^{\star},T)=c^{\rm a}(L^{\star},T)$ and $c^{\star}(K,T)=c^{\rm b}(K,T)$, and the same holds for the corresponding implied volatilities.

The discussion above can be rigorously justified by the following proposition, showing that a dual optimizer exists under our setting; that is, for marginals presented in~\cref{tab:bid_ask_lognormal}. In particular, the optimizer is a call spread with upper strike $K$ and lower strike $L^\star$ being a minimizer of
\[
[0,K) \ni L\mapsto f_K(L):=\frac{c^{\rm a}(L,T)-c^{\rm b}(K,T)}{K-L}.
\]

\begin{proposition}\label{prop:digital-call-spread}
Let $\mu^{\rm b}\cleq \mu^{\rm a}$ satisfy $c^{\rm a}(K,T)>c^{\rm b}(K,T)$. Let $h(x) = \mathds{1}_{\{x\geq K\}}$ be the payoff function.
It follows that the function $[0,K)\ni L\mapsto f_K(L)$ attains a minimizer $L^\star$. Then
\begin{enumerate}[label = (\roman*)]
    \item $\frac{1}{K-L^\star}((x-L^\star)^+, (x-K)^+) \in \Psi^{\rm b,a}(h)$ is a dual optimizer.
\item $\mu^\star = \mu^{\rm a}|_{[0,L^\star)} + (\mu^{\rm a}([L^\star, \infty)) - q^\star) \delta_{L^\star} + (q^\star - \mu^{\rm b}((K, \infty)) ) \delta_{K} +  \mu^{\rm b}|_{(K,\infty)}$ is a primal optimizer, where $q^\star =  f_K(L^\star)$.
\item $P(h)=D(h)= q^\star$.
\item Write $c^\star(k,T):=\E_{\mu^\star}[(X-k)^+]$, it holds that $c^\star(L^\star,T)=c^{\rm a}(L^\star,T)$, $c^\star(K,T)=c^{\rm b}(K,T)$.
\end{enumerate}
\end{proposition}
\begin{proof}
Since $c^{\rm a}(K,T)>c^{\rm b}(K,T)$, we know that $f_K(L)\to\infty$ as
$L\uparrow K$. Together with the continuity of call prices in strikes, this implies the existence of a minimizer $L^\star$. 

Set
$\ell(k):=c^{\rm b}(K,T)+q^\star(K-k)$. Minimality of $L^\star$ gives
$\ell\le c^{\rm a}(\cdot,T)$ on $[0,K]$, while convexity of
$c^{\rm b}(\cdot,T)$ gives $c^{\rm b}(\cdot,T)\le\ell$ on
$[L^\star,K]$. Hence
\[
c^\star(k,T):=
\begin{cases}
c^{\rm a}(k,T),&k\le L^\star,\\
\ell(k),&L^\star<k<K,\\
c^{\rm b}(k,T),&k\ge K,
\end{cases}
\]
is piecewise convex and satisfies
$c^{\rm b}\le c^\star\le c^{\rm a}$. To verify its global convexity, let $s=-q^\star$ denote the slope of
$\ell$. Since $\ell\le c^{\rm a}(\cdot,T)$ with equality at $L^\star$,
we have $\partial^-c^{\rm a}(L^\star,T)\le s$. Similarly, since
$c^{\rm b}(\cdot,T)\le\ell$ on $[L^\star,K]$ with equality at $K$, we
have $s\le\partial^-c^{\rm b}(K,T)\le\partial^+c^{\rm b}(K,T)$. Thus
the one-sided slopes of $c^\star$ are nondecreasing at $L^\star$ and
$K$, proving that $c^\star(\cdot,T)$ is convex.

Now we construct a primal optimizer from $c^\star$. In view of the identities
$\partial^-c^{\rm a}(L^\star,T)
=-\mu^{\rm a}([L^\star,\infty))$ and
$\partial^+c^{\rm b}(K,T)
=-\mu^{\rm b}((K,\infty))$, we have $\mu^{\rm a}([L^\star,\infty))-q^\star\ge0$ and
$q^\star-\mu^{\rm b}((K,\infty))\ge0$. If $L^\star=0$, the first
inequality follows instead from $q^\star = f_K(0)\le1$. We may
therefore define
\[
\mu^\star
=
\mu^{\rm a}\big|_{[0,L^\star)}
+\big(\mu^{\rm a}([L^\star,\infty))-q^\star\big)\delta_{L^\star}
+\big(q^\star-\mu^{\rm b}((K,\infty))\big)\delta_K
+\mu^{\rm b}\big|_{(K,\infty)}.
\]
It is clear that $\mu^\star$ is a probability
measure. For $L^\star\le k\le K$, direct computation gives
\[
\begin{aligned}
\E_{\mu^\star}[(X-k)^+]
&=\big(q^\star-\mu^{\rm b}((K,\infty))\big)(K-k)
  +\int_{(K,\infty)}(x-k)\,\mu^{\rm b}(\rd x)\\
&=c^{\rm b}(K,T)+q^\star(K-k)=\ell(k).
\end{aligned}
\]
For $k\ge K$, the same expectation is $c^{\rm b}(k,T)$. At
$k=L^\star$, the preceding computation gives
$\E_{\mu^\star}[(X-L^\star)^+]=c^{\rm a}(L^\star,T)$, while
$\mu^\star([L^\star,\infty))
=\mu^{\rm a}([L^\star,\infty))$. Splitting the call payoff at
$L^\star$ therefore gives
$\E_{\mu^\star}[(X-k)^+]=c^{\rm a}(k,T)$ for $k\le L^\star$.
Consequently,
$\E_{\mu^\star}[(X-k)^+]=c^\star(k,T)$ for all $k\in\R_+$.
Since $c^{\rm b}\le c^\star\le c^{\rm a}$, we conclude that
$\mu^\star\in\cQ(\mu^{\rm b},\mu^{\rm a})$. Moreover,
$\mu^\star([K,\infty))=q^\star$.
\\
Finally, let
$\psi^\star(x)=((x-L^\star)^+-(x-K)^+)/(K-L^\star)$. Then
$\psi^\star\ge h$ and its cost is $q^\star$. Together with weak
duality, this yields
\[
q^\star\le P(h)\le D(h)\le q^\star.
\]
Thus $\mu^\star$ and the stated call spread are primal and dual
optimizers, respectively. The identities
$c^\star(L^\star,T)=c^{\rm a}(L^\star,T)$ and
$c^\star(K,T)=c^{\rm b}(K,T)$ follow directly from the construction.
\end{proof}

\Cref{fig:spx_digital_bounds} illustrates the performance of BAMOT by comparing the super- and subhedging prices of SPX digital options obtained from the fitted bid--ask marginals with the super- and subhedging prices obtained by hedging directly with the finitely many quotes available in the raw market data. We use quotes observed on 2026-08-06 for options expiring on 2026-09-18 to provide an additional validation of our approach. For both approaches, the hedging inequalities are enforced on the same state-space grid, and the strikes of the call options are taken directly from the market quotes. As shown in~\cref{fig:spx_digital_bounds}, the resulting price bounds are nearly identical. In general, one can expect the match to be close for liquid markets with many quotes, with a potentially larger gap when hedging instruments are sparse.
\begin{figure}
    \centering
    \caption{Robust bounds for SPX digital calls with payoff $\mathbf{1}_{\{x\geq K\}}$, computed from raw market quotes and MLN-fitted bid--ask marginals. 
    Quotes observed on 2026-08-06 for options expiring on 2026-09-18.
    }
    \label{fig:spx_digital_bounds}
    \begin{subfigure}[t]{0.45\linewidth}
        \centering
        \includegraphics[width=2.75in,height = 2in]{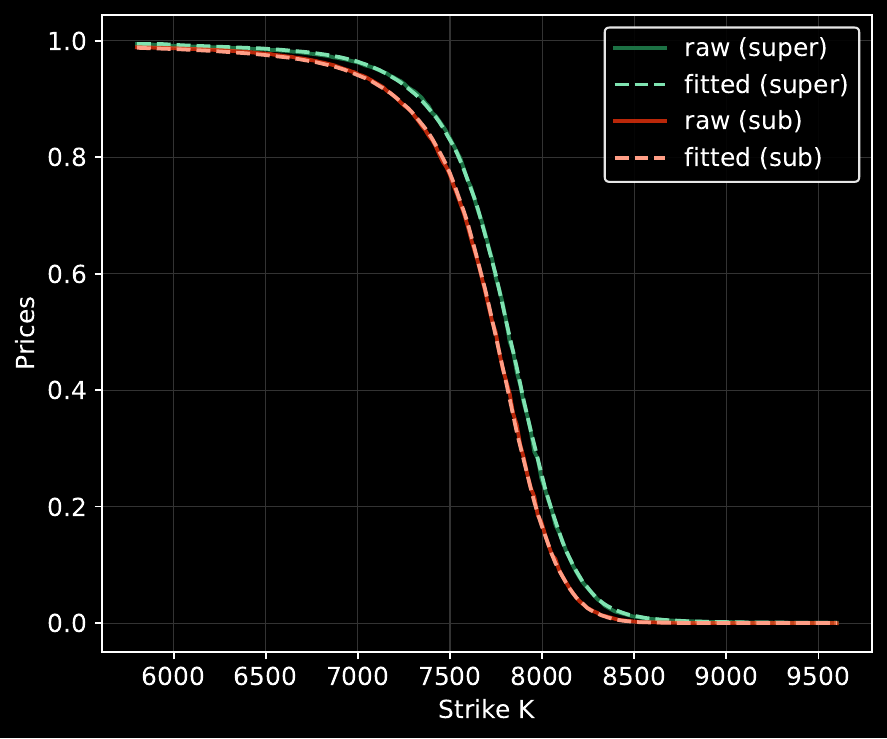}
        \caption{Super- and subhedging bounds}
        \label{fig:spx_digital_bound_levels}
    \end{subfigure}
    \hfill
    \begin{subfigure}[t]{0.45\linewidth}
        \centering
        \includegraphics[width=2.75in,height = 2in]{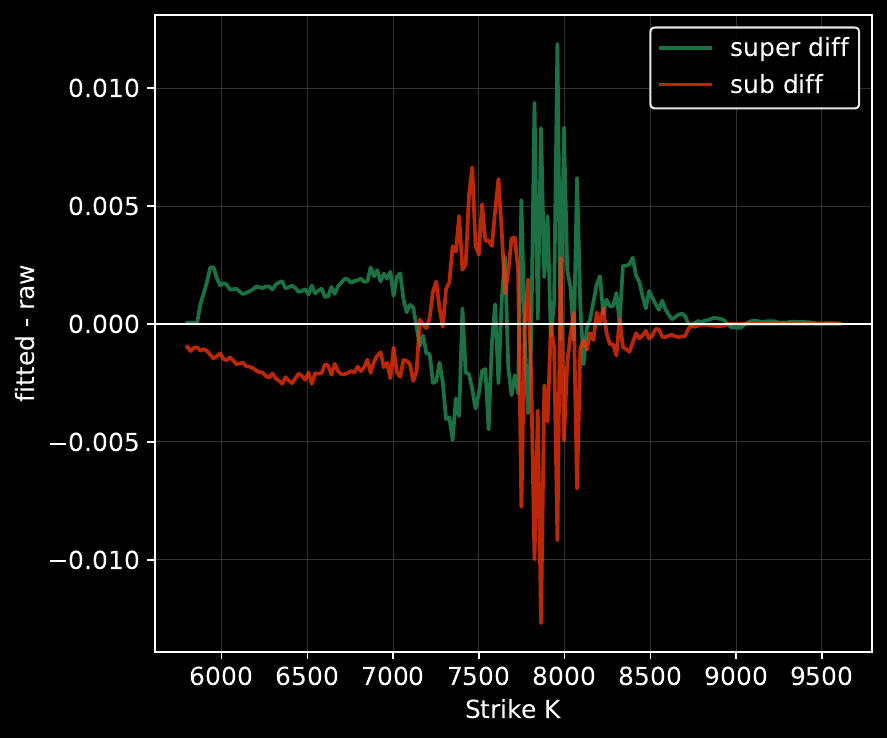}
        \caption{Difference of bounds}
        \label{fig:spx_digital_bound_diff}
    \end{subfigure}
\end{figure}

\subsection{Forward Start Option}\label{sec:forward-start}
In this section, we compare the price bounds computed under BAMOT with those obtained from MOT.
Consider the case $N=2$ and take the forward-start payoff
\[
h(x_1,x_2) = (x_2 - K x_1)^+.
\]
Observe that $
|x_2 - x_1|
= 2(x_2 - x_1)^+ - (x_2 - x_1)$, 
and hence
$\E_{\Q}[|X_2 - X_1|] = 2\E_{\Q}[(X_2 - X_1)^+]$
for any martingale coupling $\Q \in \cQ(\mu^{\rm b},\mu^{\rm a})$, where
$\mu^{\rm a} = (\mu_1^{\rm a},\mu_2^{\rm a})$ and
$\mu^{\rm b} = (\mu_1^{\rm b},\mu_2^{\rm b})$. Assume zero interest rates, we initialize the bid--ask marginals $\mu_i^{\rm a},\mu_i^{\rm b}$,
$i=1,2$, from the Black--Scholes model with parameters
\[
(T_1,\sigma_1^{\rm b})=(0.5,0.19),\quad
(T_1,\sigma_1^{\rm a})=(0.5,0.20),\quad
(T_2,\sigma_2^{\rm b})=(1,0.17),\quad
(T_2,\sigma_2^{\rm a})=(1,0.18),
\]
whose density and potential functions are presented in \cref{fig:bs-density} and \cref{fig:bs-potential}, respectively.

\begin{figure}[!t]%
    \centering
        \caption{(a)--(b): The left shows the density functions of $\mu_1^{\rm b}, \mu_1^{\rm a}, \mu_2^{\rm b}, \mu_2^{\rm a}$, and the right displays the corresponding potential functions, confirming \cref{asm:bidAskConvexOrder}. 
    (c)--(d): Super- and subhedging bounds and their differences for $h(x_1,x_2)=(x_2-Kx_1)^+$ with $K\in \{0.8, 0.85, \ldots, 1.15, 1.2\}$, given strikes $\{60,65, 70,\ldots,140\}$ at maturities $T_1,T_2$, spot $x_0=100$ and zero interest rates.}

    \begin{subfigure}[b]{0.48\linewidth}
        \centering
                \caption{Densities of $\mu_1^{\rm b}, \mu_1^{\rm a}, \mu_2^{\rm b}, \mu_2^{\rm a}$}
        \includegraphics[width=0.9\linewidth]{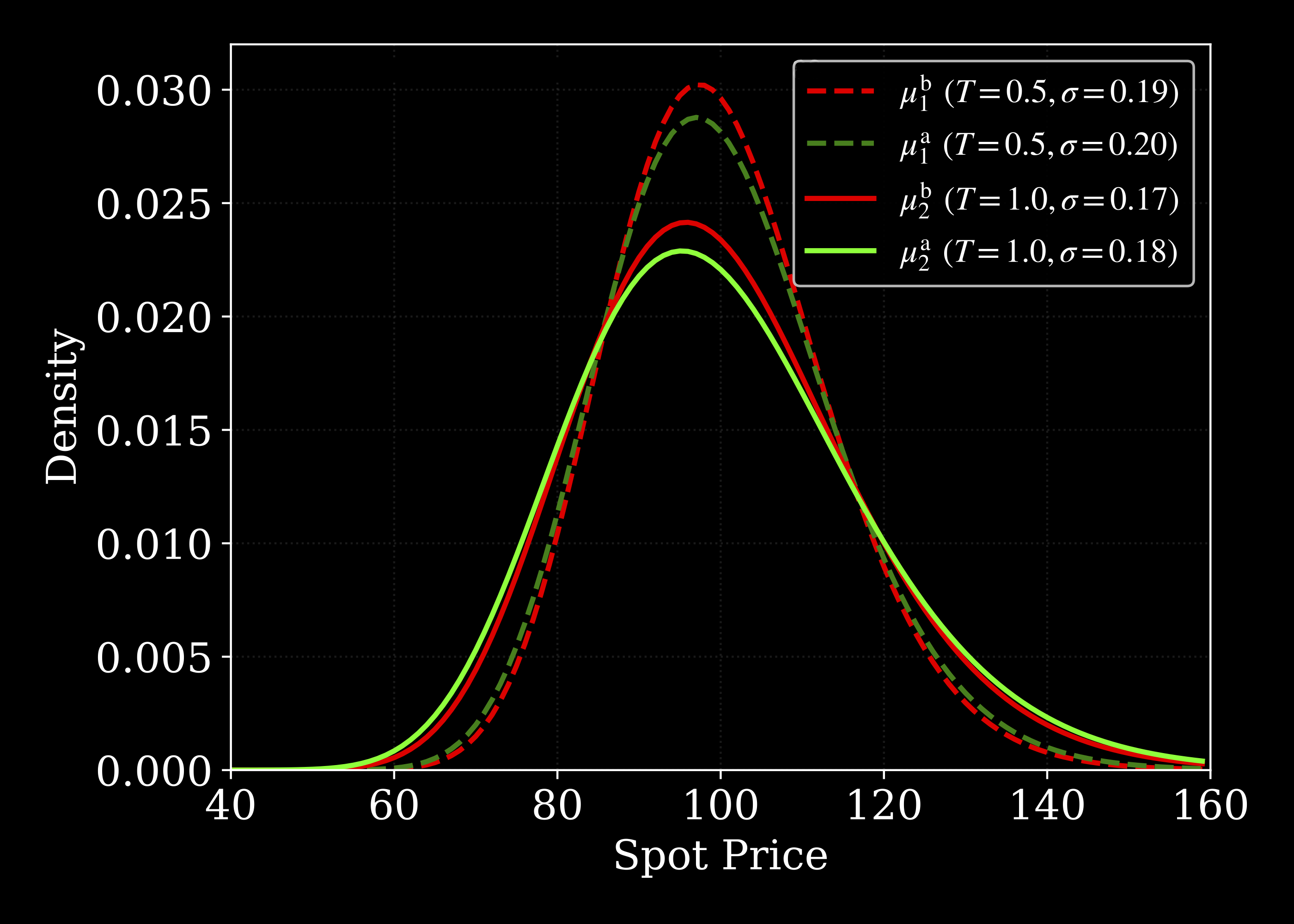}
        \label{fig:bs-density}
    \end{subfigure}
    \hfill
    \begin{subfigure}[b]{0.48\linewidth}
        \centering
           \caption{Potential functions}
        \includegraphics[width=0.9\linewidth]{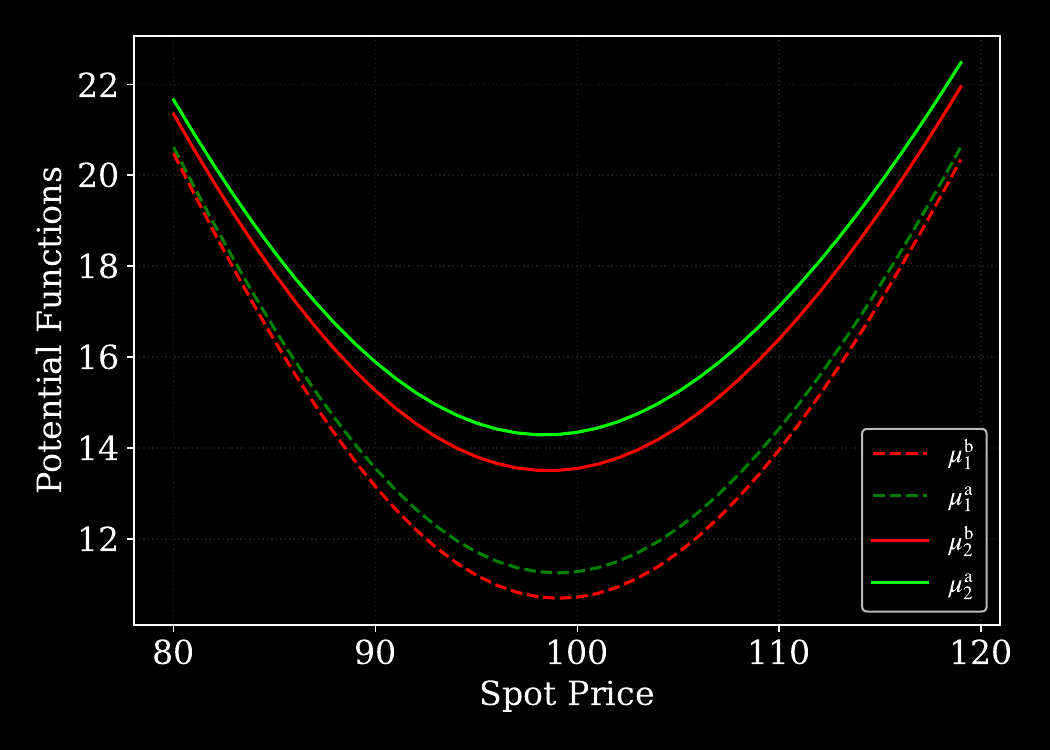}
        \label{fig:bs-potential}
    \end{subfigure}

    \begin{subfigure}[b]{0.48\linewidth}
        \centering
        \caption{Super- and subhedging bounds}%
        \includegraphics[width=0.9\linewidth]{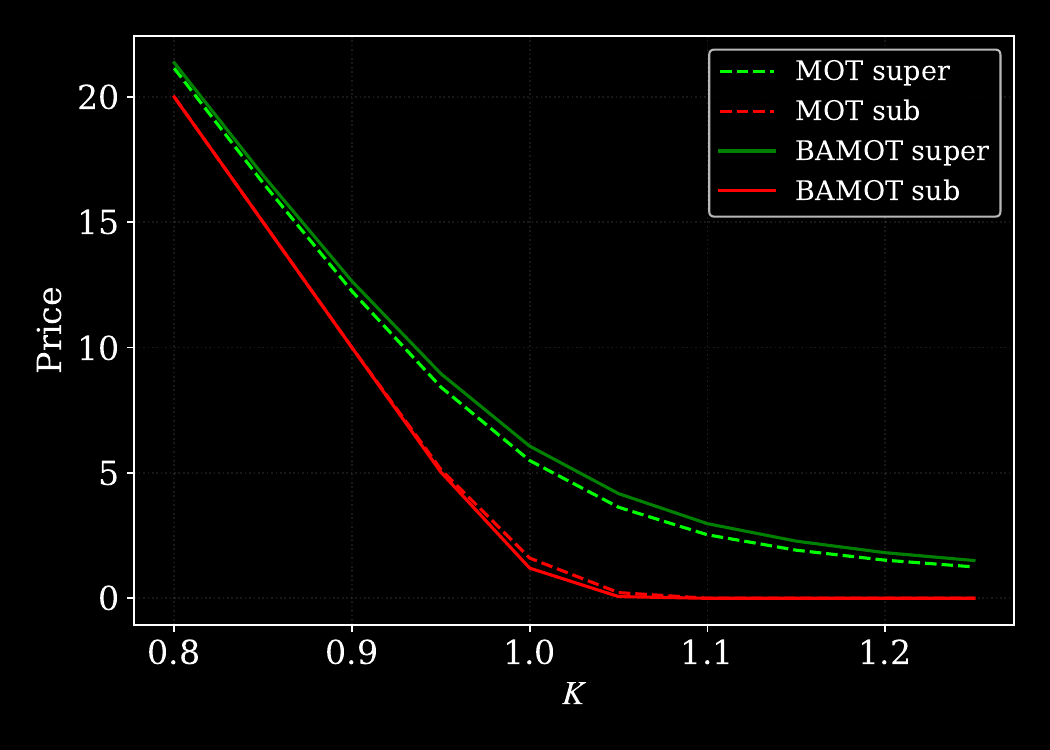}
        \label{fig:mixture-bounds}
    \end{subfigure}
    \hfill
    \begin{subfigure}[b]{0.48\linewidth}
        \centering
         \caption{Differences of robust bounds}
        \includegraphics[width=0.9\linewidth]{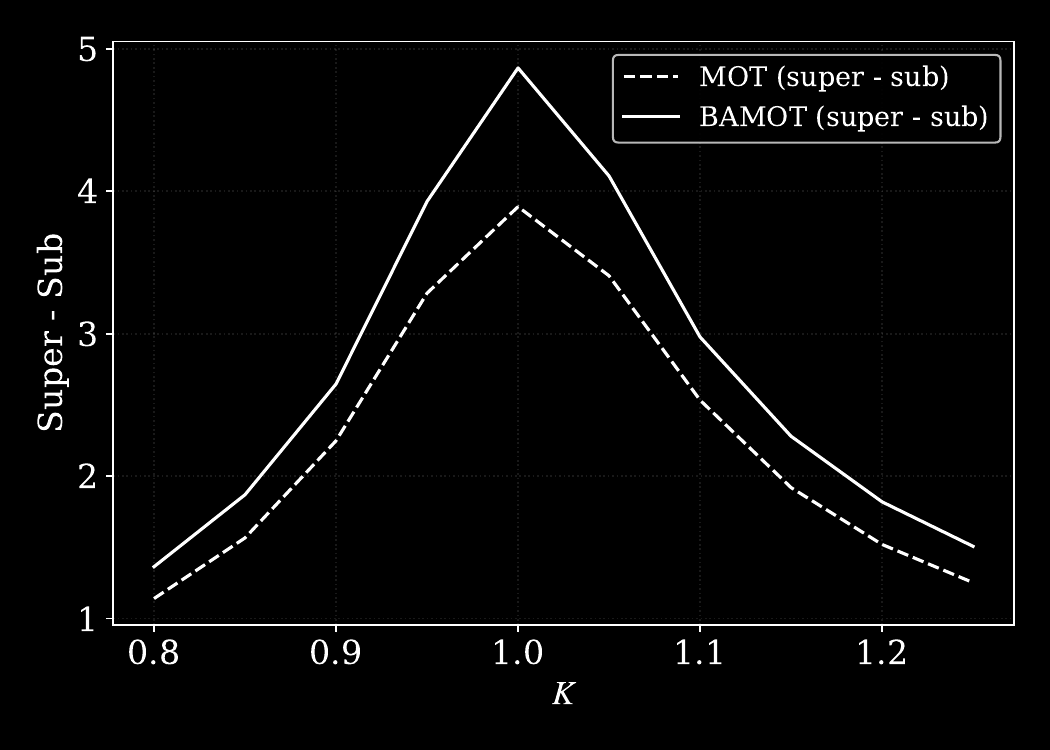}
        \label{fig:mixture-bounds-diff}
    \end{subfigure}
    \label{fig:combined}
\end{figure}
Fix a collection of call strikes $\{K_{m}\}_{m\in [M]}$ available at both $T_i, i=1,2$. The dual problem $D(h)$ then leads to the following discretized optimization (see~\cref{app:LP} for details):
\begin{align*}
   \inf \quad &
   a + \sum_{i=1}^2 b_i x_0
   + \sum_{i=1}^2 \sum_{m=1}^M
   \Bigl(
   c_{i,m}^{\rm a} c^{\rm a}(K_{m},T_i)
   - c_{i,m}^{\rm b} c^{\rm b}(K_{m},T_i)
   \Bigr) \\
   \text{such that }\quad &
   a + \sum_{i=1}^2 b_i x_i
   + \sum_{i=1}^2 \sum_{m=1}^M
   (c_{i,m}^{\rm a}-c_{i,m}^{\rm b})(x_i-K_{m})^+ \\
   &\qquad
   + \Delta(x_1)(x_2-x_1)
   \;\ge\;
   (x_2-Kx_1)^+,
   \qquad (x_1,x_2)\in(\R_+)^2,\\
   & a,b_i\in\R,
   \qquad c_{i,m}^{\rm a},c_{i,m}^{\rm b}\ge 0, \qquad i =1,2,\, m\in [M].
\end{align*}
We compute and compare the super- and subhedging prices under BAMOT and MOT using mid marginals, which are constructed by averaging the bid and ask marginals, for
$K \in \{0.8, 0.85, \ldots, 1.2\}$, and present the results in \cref{fig:combined}.
As illustrated in \cref{fig:mixture-bounds} and \cref{fig:mixture-bounds-diff}, BAMOT produces a wider price range
than MOT. The discrepancy is most pronounced near $K = 1$ and
diminishes as $K$ approaches the boundary values.

\subsection{Convergence Rates} \label{sec:convDigital}%

In this section, we present two results on the convergence rate of BAMOT to classical MOT when the payoff is either a risk-reversal strategy or an at-the-money (ATM) digital option. We show that the former exhibits a linear rate of convergence, while the latter admits a square-root rate.

For both payoffs, we consider the setting when $\mu^{\rm b}$ and $\mu^{\rm a}$ are the one-year marginals, i.e., $T=1$, in the Black--Scholes model with zero interest rates and volatilities $\sigma^{\rm b}=15\%$ and $\sigma^{\rm a}=20\%$, respectively. We set the spot price $x_0=1$ for simplicity. We then continuously deform the bid and ask marginals toward the mid marginal $\mu^{\rm m}=\frac{\mu^{\rm b}+\mu^{\rm a}}{2}$ by setting $\mu^{\rm s,\gamma}=(1-\gamma)\mu^{\rm s}+\gamma\mu^{\rm m}$, ${\rm s}\in\{{\rm b},{\rm a}\}$, with $\gamma\in[0,1]$ increasing to one.

\begin{figure}[!ht]
    \caption{Left: Bid--ask spread of one-year call options  in the Black--Scholes model with $\sigma^{\rm b}=15\%$ and $\sigma^{\rm a}=20\%$. Comparison with Black--Scholes Vega based on bid prices. Right: Risk-reversal strategy $h(x)=(x-1.05)^+-(0.95-x)^+$ 
            and  optimal dominating profile.} %
    \centering
    \begin{subfigure}[t]{0.45\linewidth}
        \centering
        \includegraphics[width=2.75in,height = 2in]{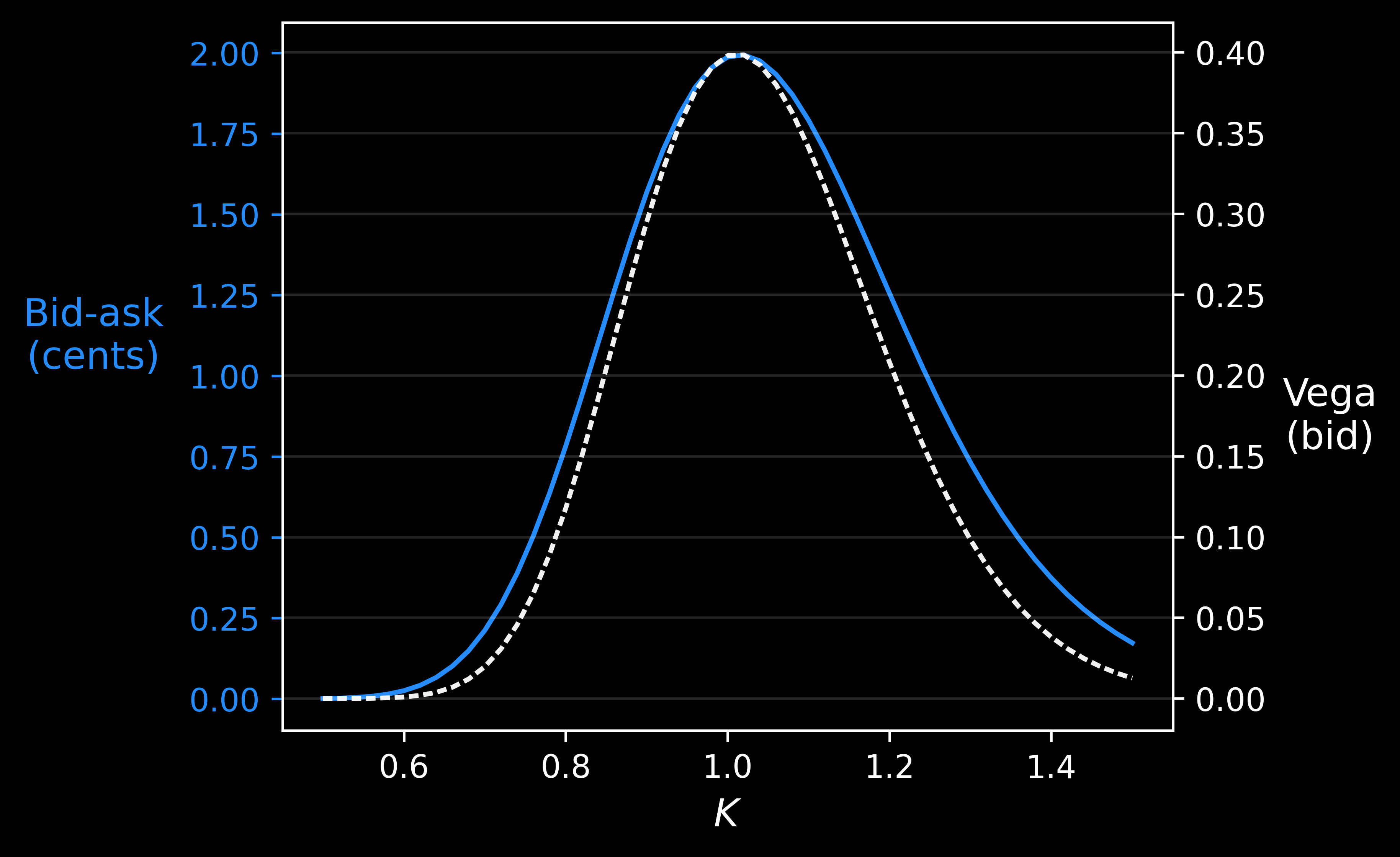}
        \label{fig:BidAskVega}
    \end{subfigure}
    \hfill
    \begin{subfigure}[t]{0.45\linewidth}
        \centering
        \includegraphics[width=2.75in,height  = 2in]{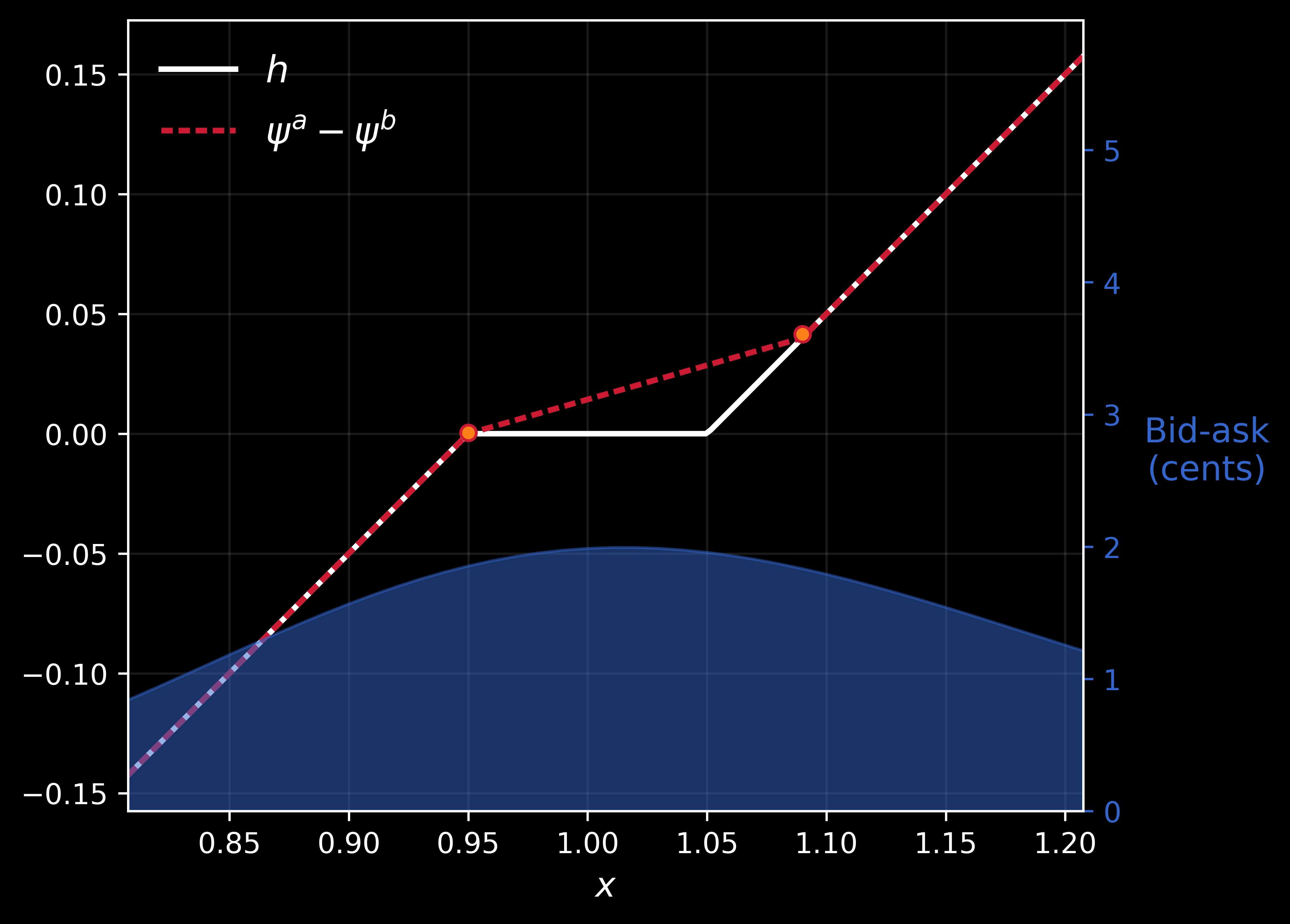}
        \label{fig:RRDominance}
    \end{subfigure}
    \label{fig:BS_numerics}
\end{figure}

The bid--ask distances are obtained by maximizing the bid--ask spread among call options; see \cref{rmk:bid-ask-distance}~(a). By a first-order Taylor expansion of the Black--Scholes price $c_{\text{BS}}(K,T,\sigma)$ with respect to $\sigma$, the bid--ask spread of call options is maximized where the Vega is large, namely near the money ($K\approx 1$). This behavior is confirmed in the left panel of \cref{fig:BidAskVega}, showing the spread for the initial Black--Scholes marginals $\mu^{\rm b}, \mu^{\rm a}$.

For the risk-reversal payoff $h(x)=(x-1.05)^+-(0.95-x)^+$, which is a difference of two convex functions, \cref{prop:MidEstimate} implies a linear convergence rate. The right panel of \Cref{fig:RRDominance} displays the payoff together with an optimal dominating profile $\psi^{\rm a,\star}-\psi^{\rm b,\star}$, and \cref{fig:RRBSConv} confirms the linear rate of convergence. On the other hand, for the ATM digital option $h(x)=\mathds{1}_{\{x\ge 1\}}$, \cref{prop:USC_sqrt_rate} predicts a square-root convergence rate, and \cref{fig:DigitalBSConv} illustrates that this rate is sharp.

\begin{figure}[!ht]
    \centering
    \caption{BAMOT convergence for a $95\%-105\%$ risk-reversal strategy.} %
   \centering
    \begin{subfigure}[t]{0.45\linewidth}
        \centering
            \caption{Superhedging price $P(h)$}
        \includegraphics[width=2.5in,height = 2in]{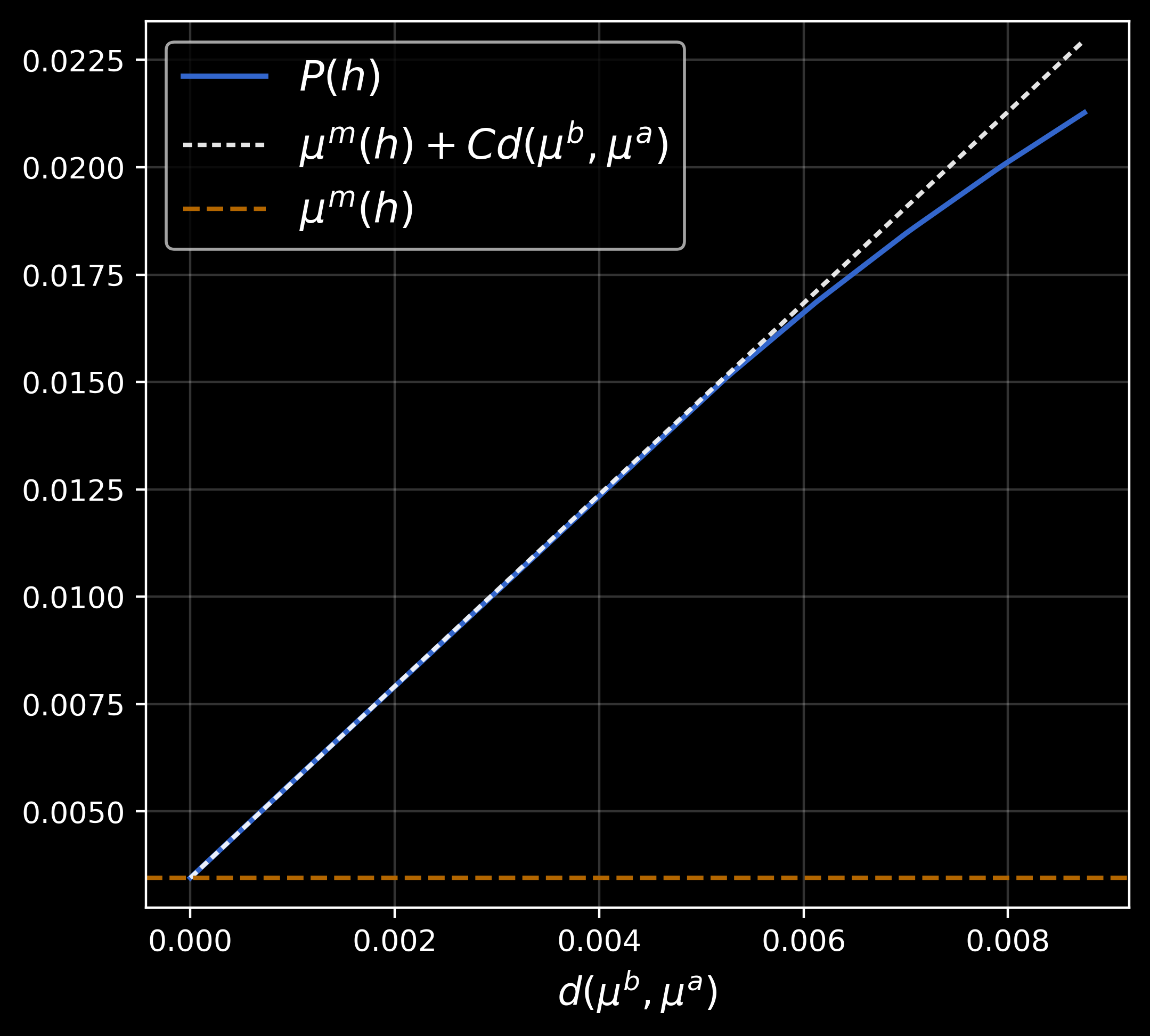}
        \label{fig:RR_conv1}
    \end{subfigure}
    \hfill
        \begin{subfigure}[t]{0.45\linewidth}
        \centering
            \caption{Premium $P(h)- \mu^{\rm m}(h)$}%
        \includegraphics[width=2.5in,height  = 2in]{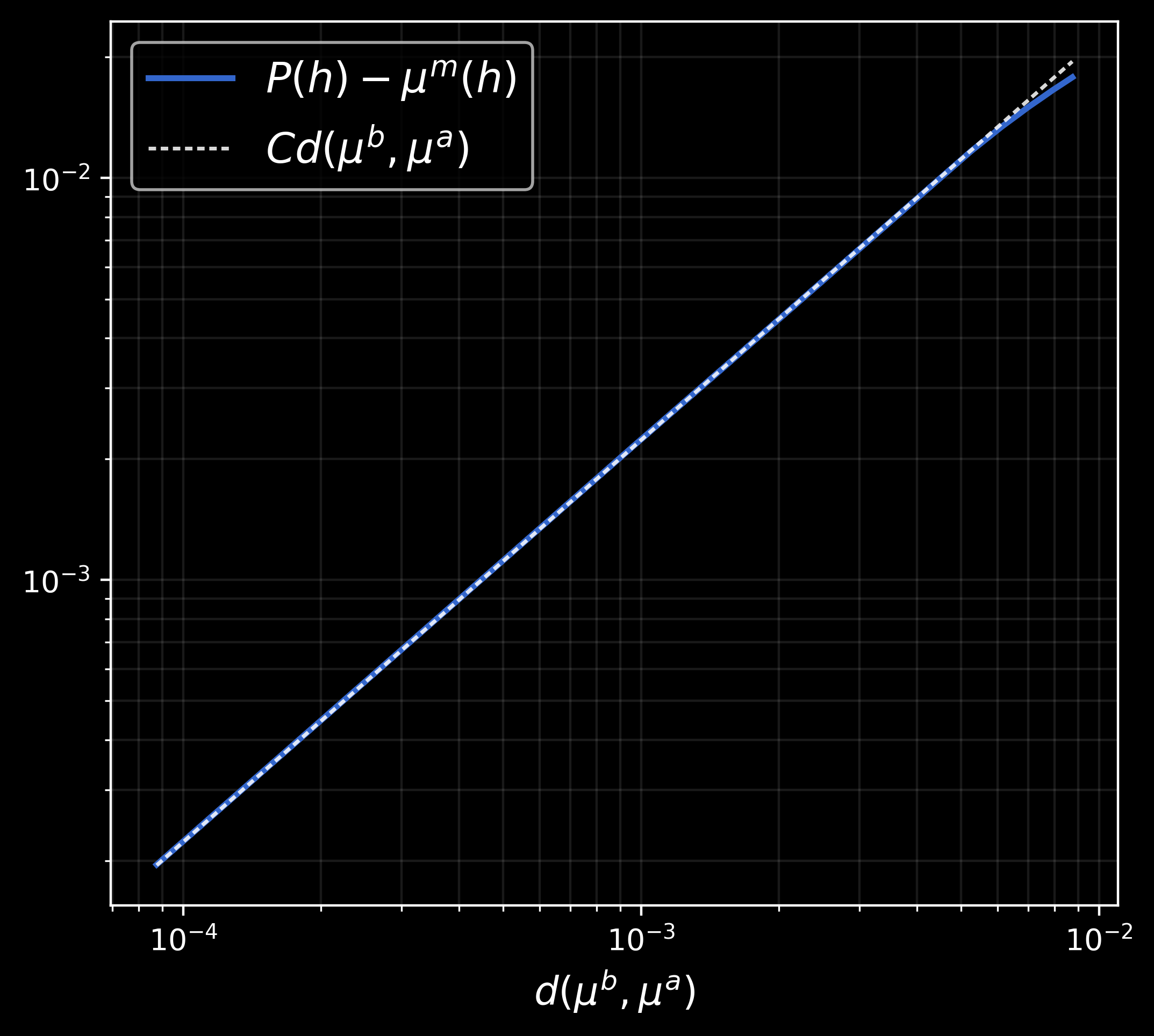}
        \label{fig:RR_conv2}
    \end{subfigure}
    \label{fig:RRBSConv}
\end{figure}

\begin{figure}[!ht]
    \centering
    \caption{BAMOT convergence for an at-the-money digital option.} %
   \centering
    \begin{subfigure}[t]{0.45\linewidth}
        \centering
            \caption{Superhedging price $P(h)$}
        \includegraphics[width=2.5in,height = 2in]{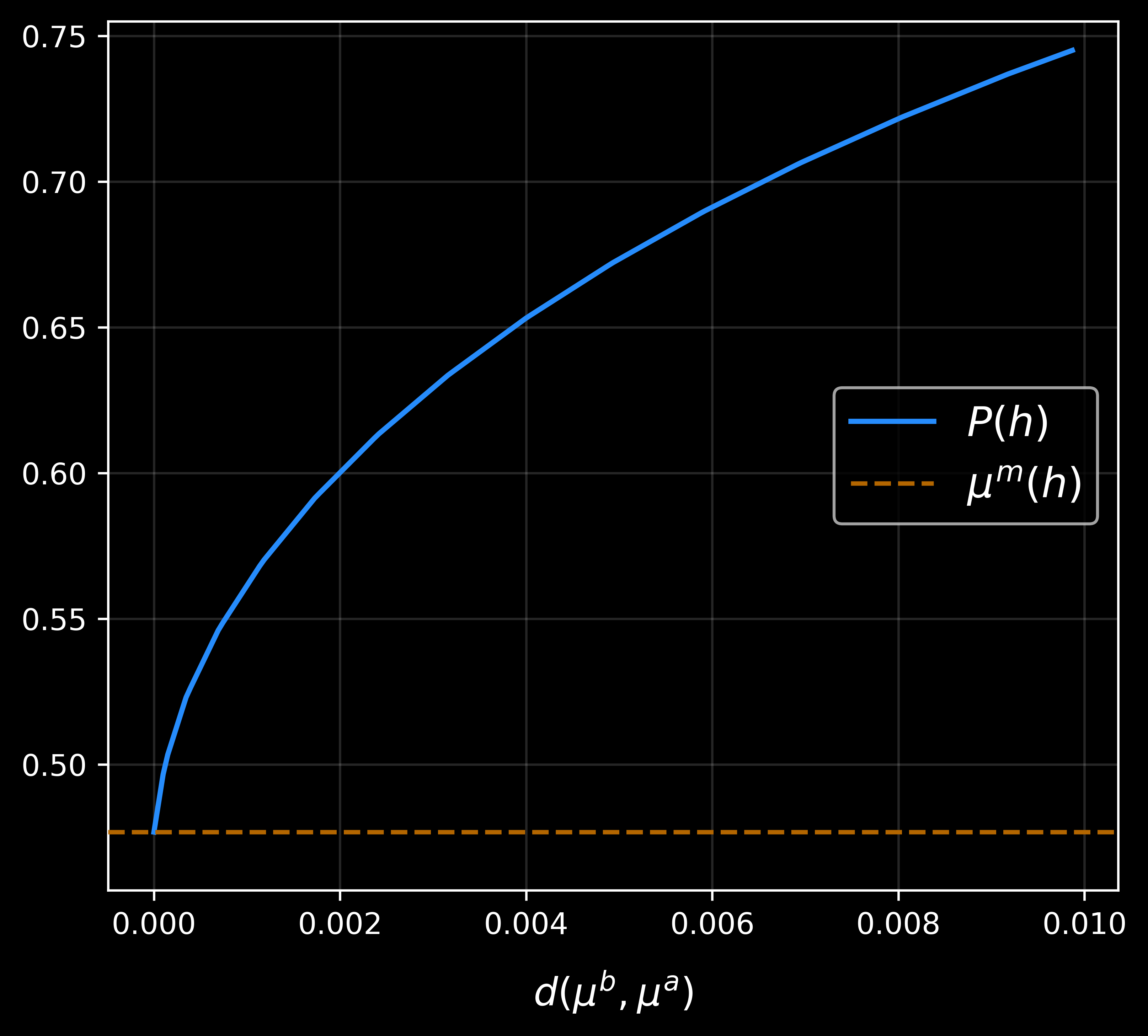}
        \label{fig:digital_conv1}
    \end{subfigure}
    \hfill
        \begin{subfigure}[t]{0.45\linewidth}
        \centering
            \caption{Premium $P(h)- \mu^{\rm m}(h)$}%
        \includegraphics[width=2.5in,height  = 2in]{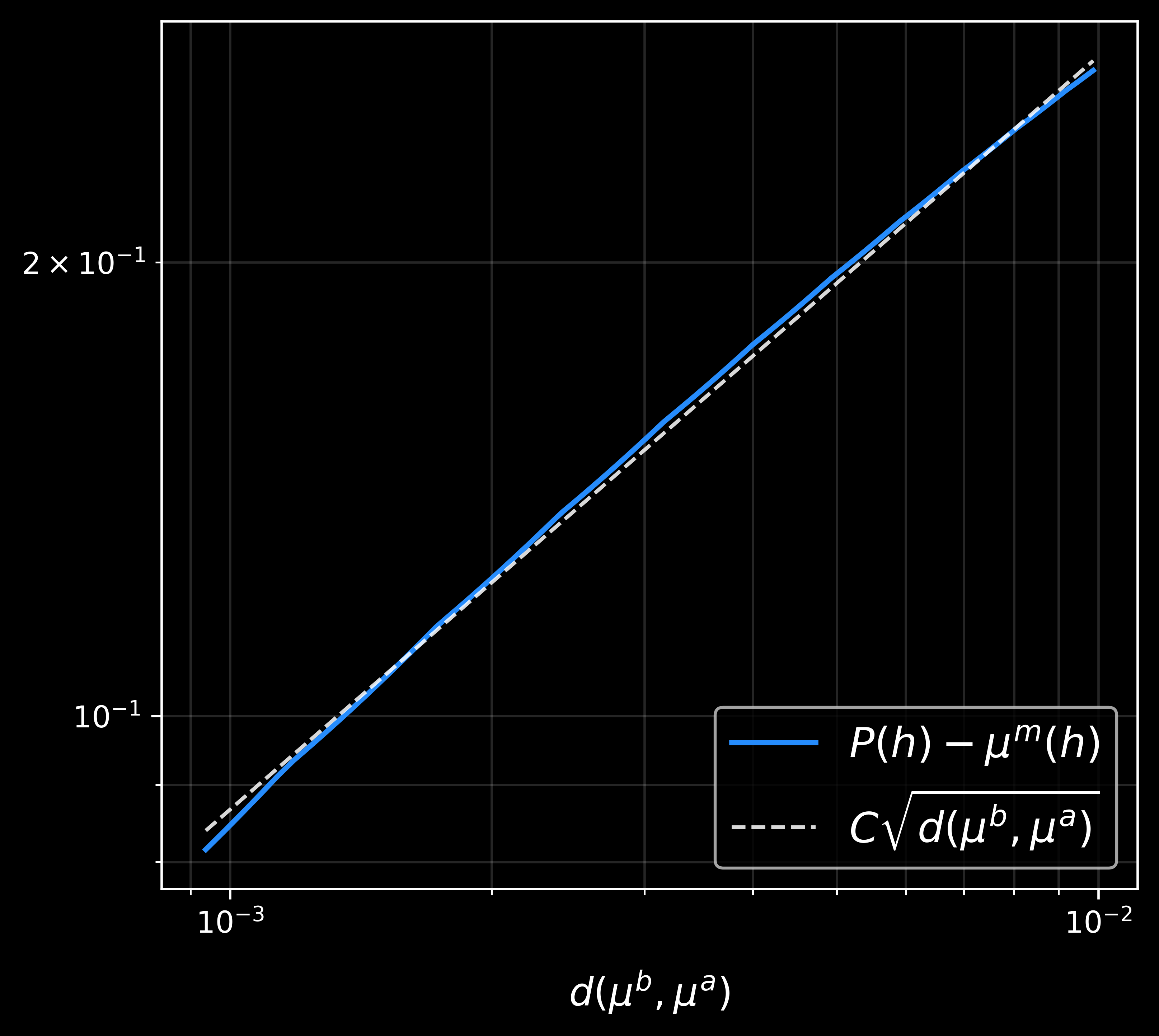}
        \label{fig:digital_conv2}
    \end{subfigure}
    \label{fig:DigitalBSConv}
\end{figure}

\section{Conclusion} 
\label{sec:future-work}

In this work, we propose BAMOT as a practical and mathematically tractable framework for robust pricing that incorporates bid--ask spreads on vanilla options. We establish strong duality for upper semicontinuous payoff functions and convergence to classical MOT as the bid--ask spreads vanish. This opens the door to extending numerous classical MOT results to the bid--ask setting, including strong duality for more general payoffs, dual attainment and regularity of dual optimizers, and the stability under general perturbations of the bid and ask marginals. Another avenue for future work involves interpolating bid and ask volatilities across maturities, consistent with options flow desk practices. This would enable the framework's extension to a continuous-time setting~\cite{DOLINSKYSonerContinuous,HouObloj2018}, aligning with the demands of exotic options trading. 

\appendix

\section{Construction of Bid and Ask Marginals from Enhanced Quotes} \label{sec:existenceBidAsk}

As noted in \cref{sec:bid-ask-marginals}, the existence of bid and ask marginals is implicitly assumed in practice in order to produce arbitrage-free implied volatility skews. A natural question is whether these marginals always exist, which is the focus of this section.  
Specifically, we present a procedure for enhancing market quotes by replacing them with the ``best” available super- and subhedging prices, which enables exact calibration of the ask marginals; see \cref{fig:SMIImprovement}.
On the other hand, our findings suggest that marginal distributions can  only produce a best-fit approximation of bid market quotes.

Consider a strip of  available strikes $K_0 < \ldots < K_M$ for  fixed maturity $T$,  and corresponding   bid and ask quotes 
$$p^{\rm b}_m \le p^{\rm a}_m, \quad c^{\rm b}_m \le c^{\rm a}_m, $$  
of put  and call  options struck at $K_m$, $m=0,\ldots,M$. %
Suppose that $K_0 = 0$, with $p_0^b = p_0^a = 0$ and $c_0^b, c_0^a $ equal to the  $T-$forward price $F$.     If  the bid or ask price is unavailable for some strike, we impute it according to 
$$p^{\rm b}_m = 0,  \quad p^{\rm a}_m = K_m, \quad  c^{\rm b}_m = 0, \quad c^{\rm a}_m = F,$$
using the bounds $ (K_m-X_T)^+ \le K_m$ and $ (X_T-K_m)^+ \le X_T$ for puts and calls, respectively. Assume the market quotes are free of static arbitrage, ensuring that the family  of calibrated models
\begin{equation*}\label{eq:discreteAdmissibleSet}
 \cQ_M^{\rm b,a} = \big\{\Q  \, \mid \,   \text{martingale measure, }   v_m^{\rm b} \le v_m^\Q \le v_m^{\rm a} \ \ \forall \ m \big\}, \quad   \begin{cases}v_m^{\Q} = \E_{\Q}[(K_m - X_T)^+], \quad v_m^{\rm s} = p_m^{\rm s}, \\[0.5em]
 v_m^{\Q} =\E_{\Q}[(X_T-K_m)^+], \quad v_m^{\rm s} = c_m^{\rm s}, 
 \end{cases}
\end{equation*}
for $ s\in \{{\rm b, \rm a}\}$, 
is non-empty. This section explains how $ \cQ_M^{\rm b,a}$ relates to the admissible set $\cQ(\mu^{\rm b},\mu^{\rm a})$   in \eqref{eq:martingale-couplings} with $N=1$. Specifically, we construct bid and ask marginals $\mu^{\rm b},\mu^{\rm a}$ such that the latter precisely recovers the ask quotes following a novel quotes enhancement procedure. A similar procedure is applied  to the bid prices, facilitating the calibration of the bid marginal to market quotes. 
\subsection{Combining Put and Call Options}\label{sec:putcall}
Recall that $c_0^{\rm b} = c_0^{\rm a} = F$, i.e., the forward contract has zero bid--ask spread. Then in view of put-call parity, we can replace the call quotes by 
\begin{align*}
    c^{\rm b}_m \leftarrow  \max\big(c^{\rm b}_m, \ p^{\rm b}_m + F- K_m\big),  \quad 
    c^{\rm a}_m \leftarrow  \min\big(c^{\rm a}_m, \ p^{\rm a}_m + F- K_m\big).
\end{align*}
Then the put quotes no longer carry any additional information 
and can be  discarded.
\subsection{Quotes Enhancement}\label{eq:improve}
Each quote may be further enhanced using other call options as hedging instruments. 
To wit, consider the following super- and subhedging problems, 
\begin{align*}
    \underline{c}^{\rm a}_m &= \inf_{\lambda \in \Lambda_m^{+}} \text{cost}^{+}(\lambda) ,\;  \Lambda_m^{+} = \Big\{(\lambda^{\rm a},\lambda^{\rm b})\in (\R_+^{M+1})^2 \, \mid \, \sum_{\ell=0}^M  (\lambda^{\rm a}_{\ell} - \lambda^{\rm b}_{\ell})(x-K_{\ell})^+ \ge (x-K_m)^+ \; \ \forall \ x\ge 0 \Big\}, \\[1em]
      \overline{c}^b_m &= \sup_{\lambda \in \Lambda_m^{-}} \text{cost}^{-}(\lambda), \;    \Lambda_m^{-} = \Big\{(\lambda^{\rm b},\lambda^{\rm a})\in (\R_+^{M+1})^2 \, \mid \, \sum_{\ell=0}^M  (\lambda^{\rm b}_{\ell} - \lambda^{\rm a}_{\ell})(x-K_{\ell})^+ \le (x-K_{m})^+ \; \ \forall \ x\ge 0 \Big\},
\end{align*}
with the cost functions $\text{cost}^{\pm}(\lambda) = \pm \sum_{\ell=0}^M  (\lambda^{\rm a}_{\ell} c^{\rm a}_{\ell}  - \lambda^{\rm b}_{\ell} c^{\rm b}_{\ell})$. By construction, we have that   $c_m^{\rm b} \le \overline{c}_{m}^{\rm b}$ and $\underline{c}_{m}^{\rm a} \le c_m^{\rm a}$ for all $m$. Hence, the above  quotes enhancement procedure generates tighter bid--ask spreads. 
 \cref{fig:SMIImprovement} illustrates this  for vanilla options on the Swiss Market Index (SMI) expiring on 10-17-2025, as of 10-09-2025. As can be seen, the bid and ask enhanced quotes do not cross and become non-increasing in strike. Moreover, while the bid quotes exhibit mixed convexity---being convex in some regions and concave in others---the ask quotes appear to be convex in the strike. We formalize these empirical observations in the next result. 

\begin{figure}[!t]
\caption{Quotes enhancement of  10-17-2025  call options on the  Swiss Market Index, as of 10-09-2025. }
    \begin{subfigure}[b]{0.48\linewidth}
        \centering
          \caption{Global view}
        \includegraphics[width=3.1in, height = 2.6in]{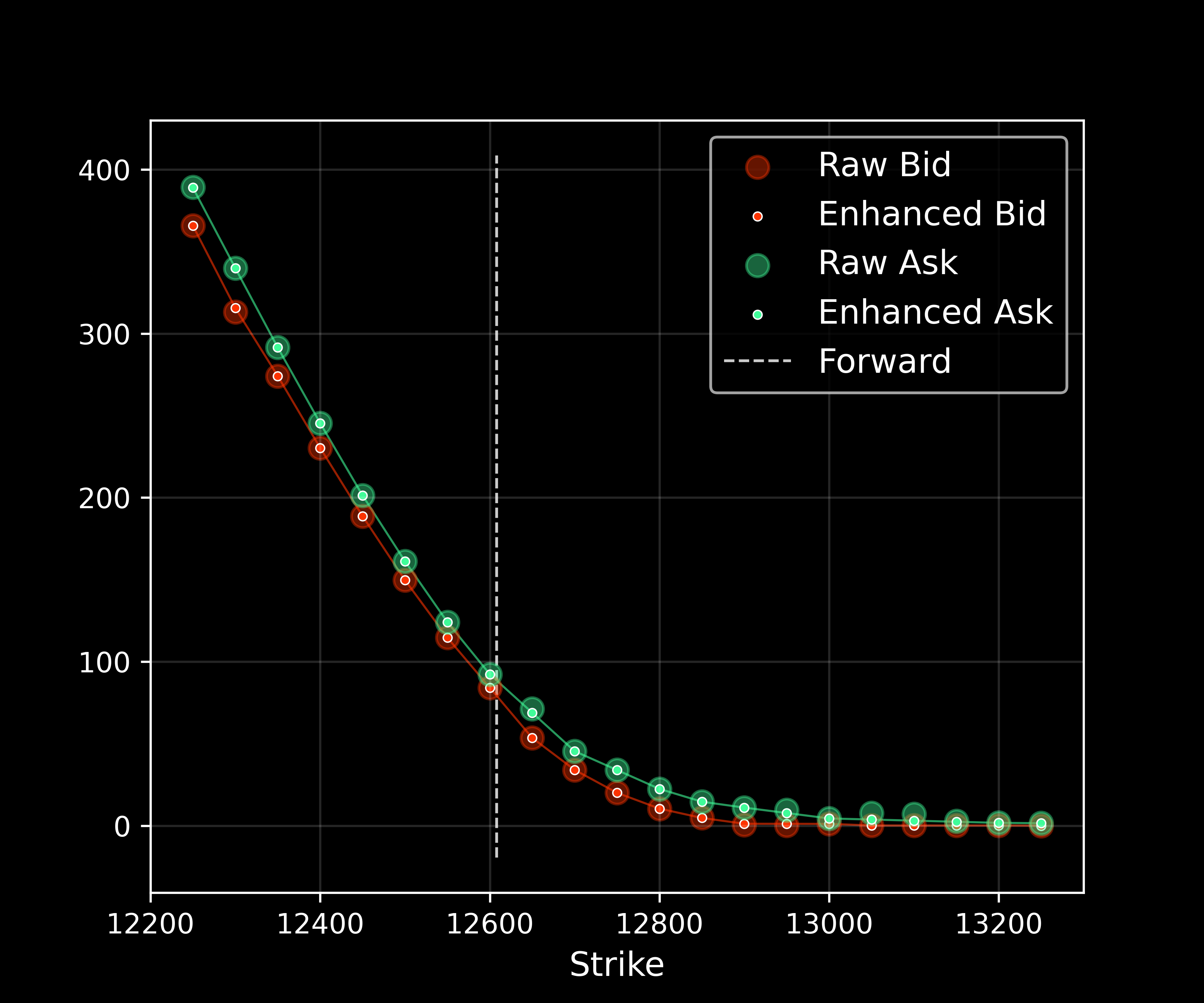}
    \end{subfigure}
    \begin{subfigure}[b]{0.48\linewidth}
        \centering
        \caption{Zoom-in View}
        \includegraphics[width=3.1in, height = 2.6in]{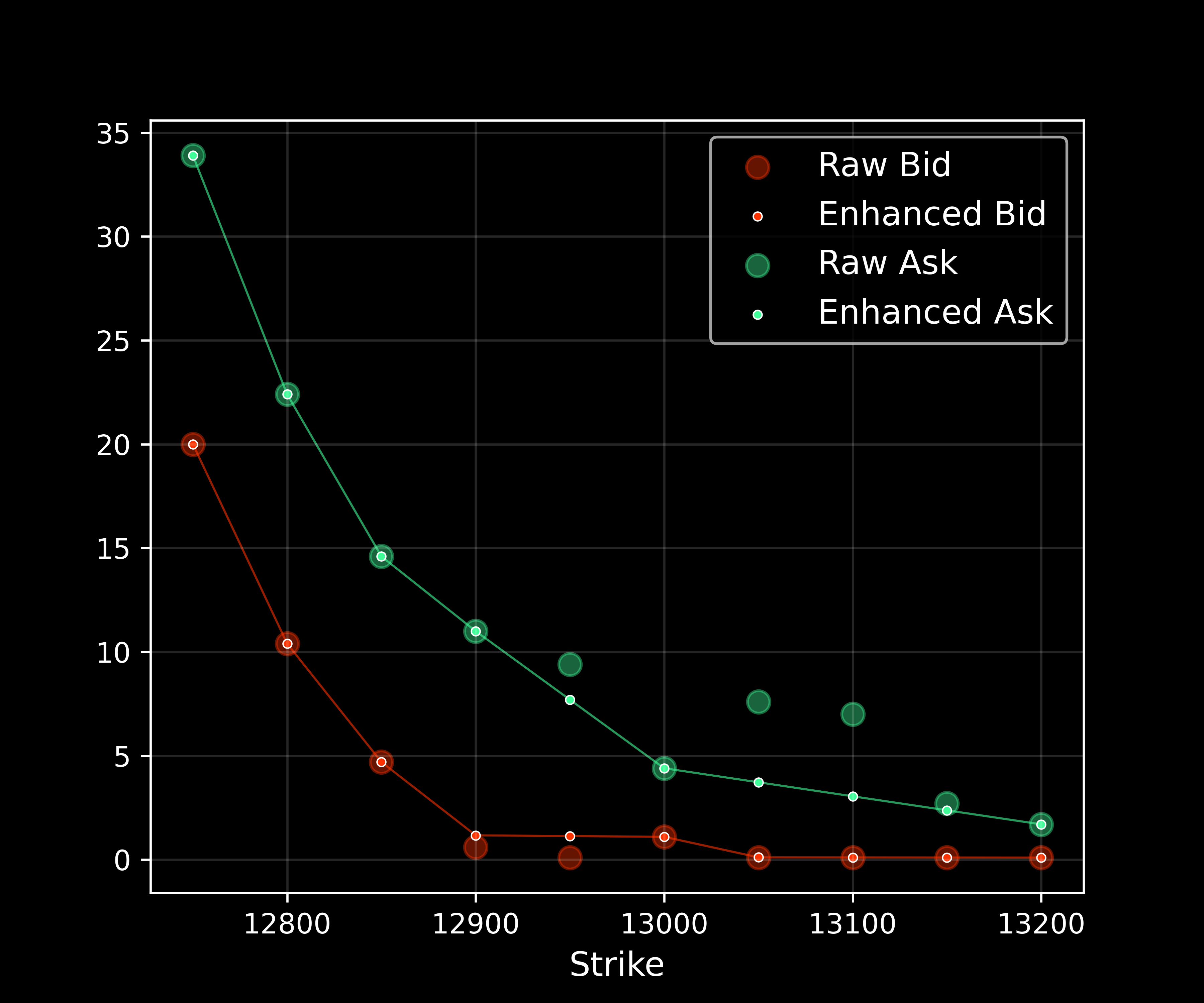}
    \end{subfigure}
    \label{fig:SMIImprovement}
\end{figure}

\begin{proposition} \label{prop:bidAsk}
    The enhanced bid and ask quotes satisfy the following properties: 
    \begin{enumerate}[label = (\roman*)]
    \item (Consistency) $\overline{c}^{\rm b}_m \le \underline{c}^{\rm a}_m$ for all $m$.
        \item (Monotonicity) $F = \underline{c}^{\rm a}_0 \ge \ldots \ge \underline{c}^{\rm a}_M$, and $F = \overline{c}^{\rm b}_0 \ge \ldots \ge \overline{c}^{\rm b}_M$. 
        \item (Convexity, ask quotes) For all $m_{-} < m < m_{+} $,  
        \begin{equation}\label{eq:convexity}
            \underline{c}^{\rm a}_{m} \le \gamma\underline{c}^{\rm a}_{m_{-}} + (1-\gamma)\underline{c}^{\rm a}_{m_{+}}\quad \text{when} \quad K_{m} = \gamma K_{m_{-}} + (1-\gamma)K_{m_{+}}.
        \end{equation}
    \end{enumerate}
\end{proposition}
\begin{proof}   

\textit{(i)} Suppose that $\overline{c}^{\rm b}_m >\underline{c}^{\rm a}_m$ for some $m\le M$. Then there exists $\lambda^{\pm} \in \Lambda_m^{\pm}$ such that $\text{cost}^{-}(\lambda^{-}) > \text{cost}^{+}(\lambda^{+})$. Consequently, 
$$ \sum_{\ell=0}^M  (\lambda^{\rm b,-}_{\ell} - \lambda^{\rm a,-}_{\ell})(x-K_{\ell})^+ \le (x-K_m)^+ \le \sum_{\ell=0}^M  (\lambda^{\rm a,+}_{\ell} - \lambda^{\rm b,+}_{\ell})(x-K_{\ell})^+. $$
Then the portfolio associated with $\lambda:= (\lambda^{\rm a,+}+\lambda^{\rm a,-},\lambda^{\rm b,+}+\lambda^{\rm b,-})$ yields a nonnegative payoff for a negative price, contradicting our assumption that the options market is arbitrage-free. 

\textit{(ii)} Fix $\ell < m$. As $(x-K_{\ell})^+ \ge (x-K_{m})^+$,  any superhedging portfolio of $(x-K_{\ell})^+$ also superhedges  $(x-K_{m})^+$. Hence $\Lambda_{\ell}^{+}\subseteq \Lambda_{m}^{+}$ and $\underline{c}^{\rm a}_{\ell} \ge \underline{c}^{\rm a}_{m}$ follows.  Similarly, $\Lambda_{m}^{-}\subseteq \Lambda_{\ell}^{-}$, which yields $\overline{c}^{\rm b}_{\ell} \ge \overline{c}^{\rm b}_{m}$. %

    \textit{(iii)} For any $\vae>0$, let %
    $\lambda_{m_{-}} \in \Lambda_{m_{-}}^+$,  $\lambda_{m_{+}} \in \Lambda_{m_{+}}^+$ 
    be  $\frac{\vae}{2}$-approximate optimizers to $\underline{c}^{\rm a}_{m_{-}}$ and  $\underline{c}^{\rm a}_{m_{+}}$, respectively. Since $K_{m} = \gamma K_{m_{-}} + (1-\gamma)K_{m_{+}}$, we have 
    $\gamma (x-K_{m_{-}})^+ + (1-\gamma) (x-K_{m_{+}})^+ \ge (x-K_{m})^+$ by the convexity of $K\mapsto (x-K)^+$. Thus, $\lambda_{m} := \gamma \lambda_{m_{-}} + (1-\gamma) \lambda_{m_{+}} \in \Lambda_{m}^+$. Therefore,
    \[
    \underline{c}^{\rm a}_{m} \leq \text{cost}^{+}(\lambda_{m}) =  \gamma \text{cost}^{+}(\lambda_{m_{-}}) + (1-\gamma)\text{cost}^{+}(\lambda_{m_{+}}) \leq \gamma\underline{c}^{\rm a}_{m_{-}} + (1-\gamma)\underline{c}^{\rm a}_{m_{+}} + \vae.
    \]
As $\vae>0$ is  arbitrary, the result follows. %

\end{proof}

From \cref{sec:putcall} and the consistency property in \cref{prop:bidAsk}, the set of calibrated models can be equivalently described as 
$$\cQ_M^{\rm b,a} = \big\{\Q  \, \mid \,   \text{martingale measure, }   \overline{c}_m^{\rm b} \le c_m^\Q \le \underline{c}_m^{\rm a} \ \ \forall \ m \big\}.$$ 

The quote enhancement mechanism described above ensures the existence of an ask marginal that precisely recovers the filtered ask quotes. A construction is provided in the next section for  completeness, keeping in mind that the parametric approach  in \cref{sec:bid-ask-marginals} yields more realistic implied volatility skews and is thus preferred in practice. Moreover, the following approach shows the existence of a generalized bid marginal, potentially involving negative mass, see \cref{rem:bid}.   %

A  quote-based formulation can be considered by directly optimizing over $\cQ_M^{\rm b,a}$. However, working with bid--ask marginals offers several advantages. First, it yields a clean characterization of the no-arbitrage condition, as formulated in Assumption~2.1. Compared with the quote-level conditions obtained by testing collections of strategies under bid--ask spreads, first studied in \cite{Cousot2007} and summarized in \cite[Table~1]{CohenReisingerWang2020}, Assumption~2.1 provides a concise characterization at the level of fitted marginals. Moreover, the consistency result in Section~4.1 and the bid--ask distance are intrinsically marginal notions: they rely on bid--ask marginals defined over all strikes and have no direct counterpart in a raw finite-quote formulation. In particular, mid-quotes need not be generated by any probability measure; in that case, the analogue of the limiting MOT problem in Proposition~4.1 is not well-posed because its admissible set is empty. Furthermore, the marginal formulation yields primal attainment through the compactness of $\cQ(\mu^{\rm b}, \mu^{\rm a})$ and permits an explicit characterization of the primal optimizers, whereas $\cQ_M^{\rm b,a}$ need not even be weakly compact. As shown in Propositions~5.1 and~5.3 in the single-maturity setting, these optimizers can be expressed directly in terms of the bid and ask marginals. Obtaining a comparable characterization under finite-quote constraints would appear considerably more difficult, if possible at all.

\subsection{Existence of an  Ask Marginal} \label{sec:existenceAsk}

Let $\underline{M} = \sup\{1\leq m \le  M \, \mid \, \underline{c}_{m}^{\rm a}  < \underline{c}_{m-1}^{\rm a}  \}$.\footnote{If the set is empty, then $\underline{c}_0^{\rm a} = \ldots = \underline{c}_M^{\rm a} = F$, and also $c_0^{\rm a} = \ldots = c_M^{\rm a} = F$, leading to a degenerate market.} When $\underline{M}<M$, then the enhanced ask quotes $\underline{c}_{m}^{\rm a}  =\underline{c}_{\underline{M}}^{\rm a}$ for all $m>\underline{M}$, meaning that the cheapest superhedging of the call at $K_m$ is given by the call at $K_{\underline{M}}$. We can therefore truncate the quotes to $ \underline{c}^{\rm a}_0 \ge \ldots \ge \underline{c}^{\rm a}_{\underline{M}}$, as rational agents would always prefer $(X_T - K_{\underline{M}})^+$ over $(X_T - K_m)^+$ for $m>\underline{M}$, or be indifferent if the support of the  risk-neutral distribution is contained in $[0,K_{\underline{M}}]$. 

Write $M$ in lieu of $\underline{M}$. By the maximality of $M$, we have $\underline{c}_{M}^{\rm a} < \underline{c}_{M-1}^{\rm a}$,  hence
$\partial^-\underline{c}^{\rm a}(K_M)
:= \frac{\underline{c}_{M}^{\rm a}-\underline{c}_{M-1}^{\rm a}}{K_M-K_{M-1}} < 0$. Let
$K_{M+1}:=K_M-\underline{c}_M^{\rm a}/\partial^-\underline{c}^{\rm a}(K_M)\ge K_M$ and set $\underline{c}_{M+1}^{\rm a}=0$. Define %
\begin{equation*}
\partial^+\underline{c}^{\rm a}(K)
:= \sum_{m=0}^{M} \mathds{1}_{[K_m,K_{m+1})}(K)
\frac{\underline{c}_{m+1}^{\rm a}-\underline{c}_{m}^{\rm a}}{K_{m+1}-K_m}.
\end{equation*}
Then $\partial^+\underline{c}^{\rm a} = 0$ beyond $K_{M+1}$, and  the convexity property \eqref{eq:convexity} implies that
\[
\underline{c}^{\rm a}(K)
:= F + \int_{0}^K \partial^+\underline{c}^{\rm a}(L)\,\dd L,
\qquad \underline{c}^{\rm a}(0)=F,
\qquad \forall\, K\ge 0,
\]
is a piecewise linear convex function that interpolates the enhanced ask quotes
$\{\underline{c}_m^{\rm a}\}_{m\in[M]}$. We then introduce the measure
\begin{equation}\label{eq:askmeasure}
    \mu^{\rm a}((L,K]) := \partial^+\underline{c}^{\rm a}(K) -\partial^+\underline{c}^{\rm a}(L), \;   L\le K,\quad    \mu^{\rm a}(\{0\}) = 1 + \partial^+\underline{c}^{\rm a}(0) \  \le  \ 1,
\end{equation}
which is nonnegative as $\partial^+\underline{c}^{\rm a}$ is nondecreasing in strike. We finally show that $\mu^{\rm a}$ is a probability measure and matches all enhanced ask quotes, thus corresponding to a feasible ask measure.

\begin{proposition}

The construction $\mu^{\rm a}$ in \eqref{eq:askmeasure} is a probability measure supported on $K_0<\ldots  < K_{M+1}$, such that      $\E_{\mu^{\rm a}}[(X-K_m)^+] = \underline{c}^{\rm a}_{m}$ for all $m$, with barycenter equal to  the forward $F = \underline{c}^{\rm a}_{0}$.  
\end{proposition}

\begin{proof}

Clearly,  $\mu^{\rm a}$ is  supported on $\{K_m\}_{m=0}^{M+1}$. Moreover,  its  total mass is given by 
$$\mu^{\rm a}(\R_+)= \mu^{\rm a}(\{0\}) + \mu^{\rm a}((0,K_{M+1}]) = 1 + \partial^+\underline{c}^{\rm a}(0) + \partial^+\underline{c}^{\rm a}(K_{M+1}) - \partial^+\underline{c}^{\rm a}(0) = 1.$$
Since  $\partial^+\underline{c}^{\rm a}$ is nondecreasing and $\partial^+\underline{c}^{\rm a}(0) \ge -1$ (noting that $\underline{c}^{\rm a}_1 -\underline{c}^{\rm a}_0 \ge (F-K_1)^+ -F \ge -K_1$ and $K_0 =0$),  $\mu^{\rm a}$ is nonnegative. Hence $\mu^{\rm a}$ is a probability measure. 
Finally,  by construction and integration by parts, 
\begin{align*}
    \E_{\mu^{\rm a}}[(X-K_m)^+] &= \int_{K_m}^{K_{M+1}} (x-K_m) \mu^{\rm a}(\rd x)  =  -\int_{K_m}^{K_{M+1}}\partial^+\underline{c}^{\rm a}(x)dx = \underline{c}^{\rm a}(K_m)=  \underline{c}^{\rm a}_{m}.
\end{align*}
    
\end{proof}

\begin{remark}\label{rem:bid}
A parallel construction applies to the enhanced bid quotes, leading to an exact fit. However, because the curve formed by  the enhanced bid quotes is not necessarily convex, the resulting object is generally a \emph{signed} measure. The latter nonetheless yields nonnegative prices for all call options and by extension, for all nonnegative convex payoffs. This points toward a generalization of convex order to signed measures, as explained in Remark~\ref{rem:signed-bid-measures}. %

\end{remark}

\subsection{Quotes Enhancement across Maturities}\label{app:quoteAcrossMaturities}

This section shows how the quotes enhancement algorithm in Section~\ref{eq:improve} can be extended across maturities. To wit, let $T_1<\ldots < T_N$ and consider  quotes $c^{\rm b}_{m,n}, c^{\rm a}_{m,n}$ for the $(K_m,T_n)$ call, assuming put options have already been integrated, see Section~\ref{sec:putcall}.  We then set
\begin{align}
 \underline{c}^{\rm a}_{m,n} &= \inf \Big\{ \text{cost}^{+}(\lambda) := \sum_{\ell = 0}^M \sum_{i=1}^N  (\lambda^{\rm a}_{\ell,i} c^{\rm a}_{\ell,i}  - \lambda^{\rm b}_{\ell,i} c^{\rm b}_{\ell,i}), \; \lambda \in \Lambda_{m,n}^{+} \Big\}, \label{eq:enhancedAskAcrossMat}\\[0.5em]
  \Lambda_{m,n}^{+} &= \Big\{(\lambda^{\rm a},\lambda^{\rm b})\in (\R_+^{(M+1) \times N})^2 \, \mid \, (\psi^{\lambda^{\rm a}},\psi^{\lambda^{\rm b}}) \in \Psi(h_{m,n}) \Big\}, \nonumber \\[0.5em] \; 
  h_{m,n}(x_1,\ldots, x_N) &= (x_n - K_m)^+, \nonumber
\end{align}
where %
$\psi_{i}^{\lambda}(x) = \sum_{\ell =0}^M \lambda_{\ell,i} (x- K_{\ell})^+$, $i = 1,\ldots, N$, and the admissible set $\Psi(h_{m,n})$ given in \eqref{eq:admissibleDual}. The enhanced bid quotes  $ \overline{c}^{\rm b}_{m,n}$ can be defined similarly. While expanding the set of hedging instruments across maturities naturally 
leads to  tighter bid--ask spreads compared to those in Section~\ref{eq:improve},  
it crucially ensures that the term structures $n \mapsto \underline{c}^{\rm a}_{m,n}$ are nondecreasing.  %

\begin{proposition}
    We have %
    $\underline{c}^{\rm a}_{m,1}\le \ldots \le \underline{c}^{\rm a}_{m,N}$ for all $m = 0,\ldots, M$. 
\end{proposition}
\begin{proof} Fix $n<N$ and $m\in \{0,\ldots,M\}$. We argue that $\Lambda_{m,n+1}^{+} \subseteq \Lambda_{m,n}^{+}$, leading to $\underline{c}^{\rm a}_{m,n} \le \underline{c}^{\rm a}_{m,n+1}$ in view of \eqref{eq:enhancedAskAcrossMat}.
Let $\lambda \in \Lambda_{m,n+1}^{+}$, that is, there exists a Delta hedging strategy  $\Delta$ such that 
    \begin{equation} \label{eq:superhedgeCalendar}
\sum_{i=1}^N (\psi_{i}^{\lambda^{\rm a}} - \psi_{i}^{\lambda^{\rm b}})(x_i) + \sum_{i=1}^{N-1} \Delta_i(x_1,\ldots,x_i)(x_{i+1} - x_i) \ge (x_{n+1}- K_m)^+\quad  \forall x_1,\ldots, x_N.
\end{equation}
   Next, we invoke a critical inequality typically used to rule out  calendar arbitrage: 
    \begin{equation} \label{eq:calendar_ineq}
(x_{n+1} - K_m)^+ -\mathds{1}_{\{x_n > K_m\}}(x_{n+1} - x_n) \ge    (x_n - K_m)^+, \quad \forall \ x_n, x_{n+1}. 
\end{equation}
Display \eqref{eq:calendar_ineq} %
follows by noting that $\mathds{1}_{\{x_n>K_m\}}$ is a subgradient of  $x\mapsto (x-K_m)^+$ at $x = x_n$. 
We then combine %
\eqref{eq:superhedgeCalendar} and   \eqref{eq:calendar_ineq} to arrive at 
$$\sum_{i=1}^N (\psi_{i}^{\lambda^{\rm a}} - \psi_{i}^{\lambda^{\rm b}})(x_i) + \sum_{i=1}^{N-1} \tilde{\Delta}_i(x_1,\ldots,x_i)(x_{i+1} - x_i) \ge (x_{n}- K_m)^+, $$
with $\tilde{\Delta}_n(x_1,\ldots,x_n) = \Delta_n(x_1,\ldots,x_n) - \mathds{1}_{\{x_n > K_m\}}$,  while $\tilde{\Delta}_i = \Delta_i$ if $i\ne n$. Hence $\lambda \in \Lambda_{m,n}^{+}$. 
\end{proof}

One can then adapt the construction of the ask measure in Section~\ref{sec:existenceAsk} to the multiple-maturity setting by linearly interpolating the enhanced quotes. To extrapolate beyond the highest available strike $K_M$ without inducing crossings between maturities, we use the flattest left slope  at $K_M$ among all maturities, namely, 
$$\max_{1 \le n \le N} \partial^-\underline{c}^{\rm a}(K_M,T_n), \quad \partial^-\underline{c}^{\rm a}(K_M,T_n) = \frac{\underline{c}_{M,n}^{\rm a}-\underline{c}_{M-1,n}^{\rm a}}{K_M-K_{M-1}}. $$ This yields ask marginals $\mu_1^{\rm a}, \ldots, \mu_N^{\rm a}$ that are nondecreasing in the convex order. In other words, they form a peacock, as announced in Remark~\ref{rem:peacockEnvelopes}. The same result holds for the bid quotes, except that the corresponding peacock may live in the space of signed measures.

\subsection{Comparison with other Approaches}
In this section, we review related work on constructing arbitrage-free option-price surfaces and marginal distributions from finite market quotes. These calibration approaches are complementary to the mixture-of-log-normals (MLN) calibration introduced in Section~2.1. The MLN approach produces separate bid and ask call-price curves over all strikes and can accommodate cases in which only one side of a quote is available. This parametric calibration also has limitations. In particular, the calibrated marginals may depend on the chosen family and fitting criterion, and cross-maturity convex-order consistency is not automatically guaranteed unless it is imposed during or after calibration. Other calibration methods, including nonparametric and LP-based procedures mentioned below, may be used instead, provided that they produce admissible bid and ask marginals. We remark that, conditional on the resulting bid and ask marginals, BAMOT remains model-free with respect to the joint dynamics and provides a continuum formulation, strong duality, and convergence guarantees. 

Although the dual value should be interpreted as the superhedging cost in the calibrated market, rather than directly in the market of raw quotes, this formulation is consistent with market practice. In particular, Bloomberg provides bid--ask implied volatility skews that can be used to construct fitted marginal distributions. Pricing and hedging can then be conducted relative to these calibrated distributions rather than directly against every raw market quote. Thus, while consistency at the quote level remains important, the fitted-marginal formulation is relevant to the pricing and hedging problem considered in the present paper.

\paragraph{Consistency of an option price surface with a martingale model.}
Carr and Madan~\cite{CarrMadan2005} give no-arbitrage conditions for vertical, butterfly, and calendar spreads across finitely many maturities and a (common) countable sequence of strikes $K_m$ going to infinity. These conditions yield marginals increasing in convex order and hence a Markov martingale matching the quoted prices. On the other hand, Davis and Hobson~\cite{DavisHobson2007} characterize consistency for arbitrary finite sets of strikes and maturities under deterministic interest rates and dividends. More closely
related is the work by Cousot~\cite{Cousot2007}. The author treats bid--ask quotes at maturity-dependent finite strike sets: conditions on vertical, calendar vertical, and calendar butterfly spreads ensure an arbitrage-free model whose option prices lie within the quoted intervals. However, unlike our approach, this selects a single price from each bid--ask interval rather than constructing separate bid and ask marginals.

\paragraph{Arbitrage-free calibration and repair.}
Given raw market quotes that may violate no-arbitrage conditions, Cohen et al.~\cite{CohenReisingerWang2020} use an LP to minimally adjust reference prices, such as mid-prices, subject to a collection of no-arbitrage constraints. In their formulation, the bid--ask bounds enter the objective as soft constraints. Consequently, the repaired prices may fall outside the quoted spreads when no arbitrage-free price surface lies entirely within them.

A more direct approach is to jointly optimize over discrete option prices and a risk-neutral measure through a finite-dimensional LP \cite{BayraktarHanNorgilas2026}. When feasible, the selected prices lie within their bid--ask intervals, while the martingale constraints enforce convex-order consistency across maturities. As formulated, however, the method requires both a bid and an ask quote for each option and yields a single marginal at each maturity, rather than separate bid and ask marginals. At a single maturity, model-free bounds and superhedging duality can also be derived from finitely many derivative-price constraints, including bid--ask spreads, through linear semi-infinite optimization \cite{NeufeldAntonisXiang2023}.

Unlike the preceding two approaches, which operate directly on finite market quotes, Alfonsi et al.~\cite{Alfonsi2019sampling,Alfonsi2020sampling} start from discrete approximations of probability measures and modify them through finite-dimensional optimization so that the resulting measures are in convex order. Their procedure is therefore applied after marginal estimation and may alter both the measures and their supports; it neither extracts marginals from finite quotes nor preserves exact calibration to those quotes.

\section{Technical Results and Proofs}\label{app:proof}

The first two lemmas are standard results related to the convex order. Given $\mu\in\cP_1$, we introduce the sublevel set
\begin{equation*}%
\underline{\cQ}(\mu)
:= \left\{ \mu' \in \cP_1 \,\middle|\, \mu' \cleq \mu \right\}.
\end{equation*}
Observe that $\underline{\cQ}(\mu) \subseteq \cP_1(A)$ whenever $A := \supp(\mu)$ is an interval in $\R_+$.

\begin{lemma}\label{lemma:continuity}
Let $f\in \cC_L(\R_+)$ and $\underline{\cQ}(\mu^{\rm a}) \ni \mu_n \overset{w}{\to} \mu$. Then $\mu \in  \underline{\cQ}(\mu^{\rm a}) $ and $\mu_n(f) \to \mu(f)$.
\end{lemma}
\begin{proof}
By the weakly lower semi-continuity of $\cP(\R_+)\ni \nu \mapsto \int x \dd \nu$ \citep[see, e.g.,][Proposition~7.1]{santambrogio2015optimal}, we conclude that $\mu \in \cP_1$. For any $R>0$, let $\chi_R$ be a continuous cutoff function with $\chi_R =1$ on $[0,R]$, $0\leq \chi_R\leq 1$ on $(R,R+1)$, and $0$ outside $[0,R+1]$. Then it holds that
\begin{align*}
    \left|\int f\dd \mu_n - \int f\dd \mu\right| \leq  \left|\int f\chi_R \dd (\mu_n- \mu)\right| + \left|\int_{\R_+} f(1-\chi_R)\dd \mu_n\right| + \left|\int_{\R_+} f(1-\chi_R)\dd \mu\right|.
\end{align*}
The first term converges to 0 by the definition of weak convergence. We now control the second term. Since $|f(1-\chi_R)(x)| \leq |f(x)|\mathds{1}_{\{x>R\}}\leq \lVert f \rVert_{L}(1+x) \mathds{1}_{\{x>R\}}$ and  $x\mathds{1}_{\{x>R\}} \leq 2(x - R/2)^+$ for all $x\in \R_+$, %
\begin{align*}
    \left|\int_{\R_+} f(1-\chi_R)\dd \mu_n\right| &\leq  \lVert f \rVert_{L} \int_{\R_+\setminus [0,R]}\left(1+x\right)\mu_n(\rd x)\\
    & \leq  \lVert f \rVert_{L} \,\mu_n\left(\R_+\setminus[0,R]\right) +  \lVert f \rVert_{L}\int 2(x-R/2)^+ \mu_n(\rd x)\\
    & \leq \frac{\lVert f \rVert_{L}}{R}\int x\mu^{\rm a}(\rd x) + 2\lVert f \rVert_{L}\int (x-R/2)^+ \mu^{\rm a}(\rd x),
\end{align*}
where the last line follows from Markov's inequality and the convexity of $x \mapsto (x - R/2)^+$, $x\mapsto x$. 
 Since $(x - R/2)^+ \to 0$ pointwise as  $R\to \infty$, $(x - R/2)^+ \leq x$ and $\mu^{\rm a} \in \cP_1$, we conclude by the dominated convergence theorem that
 \[
\limsup_{R\to \infty} \sup_{n}  \left|\int_{\R_+} f(1-\chi_R)\dd \mu_n\right|=0.
 \]
The third term can be controlled using similar arguments since $\mu \in \cP_1$. 
Finally, applying the convergence established above to $f(x)=(x-K)^+$ for $K\geq 0$, and using \eqref{eq:callCVX}, we conclude that $\mu\in\underline{\cQ}(\mu^{\rm a})$.
\end{proof}

\begin{lemma}\label{lemma:convex-order-topology}
    Let $\mu^{\rm a}, \mu^{\rm b} \in \cP_1$ satisfy $\mu^{\rm b} \cleq \mu^{\rm a}$. Then the sets  $\underline{\cQ}( \mu^{\rm a})$ and $\cQ(\mu^{\rm b}, \mu^{\rm a})$ are convex and  compact with respect to $\cW_1$. 
\end{lemma}

\begin{proof}
Let us first analyze the set $\underline{\cQ}(\mu^{\rm a})$. Its convexity is immediate, while its weak closedness follows from \cref{lemma:continuity}.
Next, we show that any sequence $(\mu_n) \subset  \underline{\cQ}(\mu^{\rm a})$ contains 
a subsequence 
that converges in $\cW_1$ to some $\mu \in \underline{\cQ}(\mu^{\rm a})$. First, take $\psi(x) = |x|$, and  observe that $\sup_{\mu \in \underline{\cQ}(\mu^{\rm a})} \mu(\psi) = \mu^{\rm a}(\psi)$. 
Hence, applying Markov's inequality implies the tightness of $(\mu_n)$. By Prokhorov's theorem (see, e.g., \cite[Theorem~8.6.2]{Bogachev2007}) and by the weak closedness of $\underline{\cQ}(\mu^{\rm a})$, there exists a subsequence $(\mu_n)$, up to re-indexing, such that $\mu_n \overset{w}{\to} \mu \in \underline{\cQ}(\mu^{\rm a})$. Moreover,  the mapping $\underline{\cQ}(\mu^{\rm a}) \ni \mu  \to \mu(\psi) $ is weakly continuous by \cref{lemma:continuity}, implying that $\mu_n \to \mu$ with respect to  $\cW_1$.

Moving on to  the general case,  note that  $\cQ(\mu^{\rm b},\mu^{\rm a}) \subseteq \underline{\cQ}(\mu^{\rm a})$ for all $\mu^{\rm b} \cleq \mu^{\rm a}$. %
It is thus sufficient to show  the closedness of   $\cQ(\mu^{\rm b},\mu^{\rm a})$ under $\cW_1$-topology, which follows from the Kantorovich--Rubinstein theorem~\citep[Theorem~1.14]{villani2009optimal} and \eqref{eq:callCVX}.%
\end{proof}

\begin{proof}[Proof of \cref{prop:bidAskDistance}]\label{proof:bidAskProp}
        \textit{(i)} If $\mu \cleq \nu$, then clearly $\vec{d}(\mu,\nu) \le 0$. Moreover, choosing $\psi(x) = x$ yields
    $$0 \ge \vec{d}(\mu,\nu) \ge \E_{\mu}[X] - \E_{\nu}[X] = 0,$$ 
    recalling that measures in convex order have identical barycenters. %
    Conversely, suppose that $\vec{d}(\mu,\nu) = 0$. Then 
    $\E_\mu[X] = \E_\nu[X]$ and $\E_\mu[(K-X)^+] \leq \E_\nu[(K-X)^+]$  $\forall \ K\ge 0$, so we  conclude from \eqref{eq:callCVX} that $\mu \cleq \nu$. 

\textit{(ii)} The triangle inequality is clear. Also, $\vec{d}$ is nonnegative since $\psi\equiv 0$ is convex and $1$-Lipschitz.  Finally, suppose that $\vec{d}(\mu,\nu) = \vec{d}(\nu,\mu) = 0$. Then $\mu \cleq \nu$ and $\nu \cleq \mu$, which implies that $\mu = \nu$.

\textit{(iii)} 
Let $\vec{d}_C(\mu,\nu) := \sup \left\{\E_{\mu}[(X-K)^+]-\E_{\nu}[(X-K)^+] \, \mid \, K\ge 0 \right\}$. Let $\psi$ convex with $\textnormal{Lip}(\psi) \le 1$. From Carr--Madan's formula, $\psi$ can be decomposed as  
\begin{equation}\label{eq:carr-madan}
    \psi(x) = \psi(0) + \psi'(0) x + \int_{\R_+}(x-K)^+ \lambda^{\psi}(\rd K), 
\end{equation}
where $\psi'$ is the right derivative of $\psi$ and $\lambda^{\psi}$  the nonnegative Radon measure on $\R_+$ induced by $\psi'$. Note that since $\textnormal{Lip}(\psi) \leq 1$ and $\psi'$ is nondecreasing, we have $\lambda^\psi(\R_+) \leq 2$. In addition, $(\mu-\nu)(1) = (\mu-\nu)(x) =0$ as $\mu,\nu\in\cP_1$ with identical barycenters. Therefore, by~\eqref{eq:carr-madan}, we have
\begin{equation}\label{eq:carr-madan-2}
    \begin{aligned}
       (\mu-\nu)(\psi) & = \int \left(\int (x-K)^+ \lambda^\psi(\rd K)\right)(\mu -\nu)(\rd x).
    \end{aligned}
\end{equation}
Applying Fubini's theorem to~\eqref{eq:carr-madan-2} yields 
\begin{equation*}
\begin{aligned}
        (\mu-\nu)(\psi) &= \int_{\R_+} \big(\E_{\mu}[(X-K)^+] - \E_{\nu}[(X-K)^+] \big) \lambda^{\psi}(\rd K) \\
        & \leq \sup_{K\geq 0} \left(\E_{\mu}[(X-K)^+] - \E_{\nu}[(X-K)^+] \right) \lambda^{\psi}(\R_+) \\
        & \leq 2\vec{d}_C(\mu,\nu).
\end{aligned}
\end{equation*}
Next, we show that $2\vec{d}_C(\mu,\nu) \le \vec{d}(\mu,\nu)$. Fix $K\geq 0$, take $\psi(x) = |x-K|$. Then, using $|x-K| = 2(x-K)^+ + (K-x)$, we derive that
\begin{align*}
     \vec{d}(\mu,\nu) \geq (\mu-\nu)(\psi) = 2\left(\E_{\mu}[(X-K)^+] - \E_{\nu}[(X-K)^+] \right).
\end{align*}
Taking supremum over $K\geq 0$ on both sides, we conclude that $ \vec{d}(\mu,\nu) \geq 2\vec{d}_C(\mu,\nu)$, as desired.

\textit{(iv)} Recall  that $\cW_1(\mu,\nu) = \sup \left\{(\mu-\nu)(\varphi) \, \mid \, \lVert \varphi \rVert_{\text{Lip}} \le 1 \right\}$ by Kantorovich--Rubinstein theorem. Then  $\vec{d}(\mu,\nu) \le \cW_1(\mu,\nu) $ for all $\mu,\nu \in \cP_1$, hence $d(\mu,\nu) \le \cW_1(\mu,\nu)$. 
Let us move on to the second statement. For ease of presentation, we construct  sequences of discrete probability measures on $\R$ satisfying the desired property. Given their compact support, these measures can be shifted to $\mathbb{R}_+$, leading to  valid sequences in $\mathcal{P}_1$. Let $\xi = \frac{1}{2}(\delta_{-1} + \delta_1)$ be the Rademacher distribution. For each $n \ge 1$, introduce 
$$\mu_n = c_n \sum_{|m|\le n} \delta_{2m}, \; c_n = \frac{1}{2n+1}, \quad \nu_n = \mu_n * \xi := c_n \sum_{|m|\le n} \frac{\delta_{2m-1} + \delta_{2m+1}}{2}.$$
Observe that  $\nu_n$ is supported on the odd integers $-2n-1,-2n+1,\ldots,2n+1$, with probability $c_n$ except at the endpoints where the probability is halved. 
We first verify that $\cW_1(\mu_n,\nu_n) = 1$ irrespective of $n$. Indeed, consider the  butterfly spreads portfolio
$$\varphi(x) = \sum_{|m|\le n} B(x,2m), \quad B(x,K) = (x-K+1)^+ - 2(x-K)^+ + (x-K-1)^+.$$
Then $\varphi(2m) = 1$ for all $|m|\le n$, while $\varphi$ vanishes on the support of $\nu_n$. It is easy to check that $\varphi$ is $1$-Lipschitz, then we conclude from Kantorovich--Rubinstein theorem that 
$$\cW_1(\mu_n,\nu_n) \ge (\mu_n-\nu_n)(\varphi) = \mu_n(\varphi) = \frac{1}{2n+1}\sum_{|m|\le n} \varphi(2m) = 1. $$
Next, consider 
 the  coupling $\Q_n = \text{Law}(X_0,X_1)$, where $X_0\sim \mu_n$ and $X_1 = X_0 + Z$ with  $Z\sim \xi$  independent of $X_0$. Then $X_1\sim \nu_n$, and $\Q_n$ is a coupling between $\mu_n$ and $\nu_n$. Hence, 
 $$\cW_1(\mu_n,\nu_n) \le \E_{\Q_n}[|X_1 - X_0|] =\E_{\Q_n}[|Z|] = 1,$$
which shows that $\cW_1(\mu_n,\nu_n) = 1$. 
 
Moving to the bid--ask distance, observe that $\mu_n \cleq \nu_n$ as $\xi$ is centered. This can also be seen from the above martingale coupling and Strassen's theorem. Hence $\vec{d}(\mu_n,\nu_n) = 0$ by the second property. Flipping the roles of $\mu_n,\nu_n$, observe that $\vec{d}(\nu_n,\mu_n) = c_n = \frac{1}{2n+1}$. Indeed, %
for any $1$-Lipschitz convex function $\psi$,  we have  after rearranging that 
\begin{align*}
   \frac{1}{c_n} (\nu_n-\mu_n)(\psi) &= \sum_{m = -n +1}^{n} \psi(2m-1) -  \frac{1}{2}[\psi(2m-2)+\psi(2m)] \\[1em]
   &+ \frac{1}{2}[\psi(-2n-1) - \psi(-2n) + \psi(2n+1) - \psi(2n)].
\end{align*}
The first term is nonpositive using the (midpoint) convexity of $\psi$, while the second is bounded by 
$$ \frac{1}{2}\big[|\psi(-2n-1) - \psi(-2n)| + |\psi(2n+1) - \psi(2n)|\big]\le 1, $$
since $\psi$ is $1$-Lipschitz. As $\psi$ is arbitrary, we conclude that $\vec{d}(\nu_n,\mu_n) \le c_n $, where the upper bound is achieved, e.g., by  the strangle  $\psi(x) = (-2n-x)^+ + (x-2n)^+$; see \cref{fig:counterexample2}. As $\lim_{n\to \infty}c_n = 0$, the claim follows.

\textit{(v)} Recall from assertion (iv) that $\cW_1$ dominates the bid--ask distance. Hence, if  $(\mu_n) \subset \cP_1$ is a sequence such that $\lim_{n\to \infty}\cW_1(\mu_n,\mu) = 0$ for some limit $\mu\in \cP_1$, then $\lim_{n\to \infty}d(\mu_n,\mu) = 0$ as well. Conversely, suppose that $\lim_{n\to \infty}d(\mu_n,\mu) = 0$. First,  $\lim_{n\to \infty}  \E^{\mu_n}[|X|] = \E^{\mu}[|X|]$ using  the  test function $\psi(x) = |x|$. In particular, the sequence $(\mu_n)$ is tight. Moreover,  %
testing against call payoffs yield the convergence of the call functions $C_n(K) = \E^{\mu_n}[(X-K)^+]$ to $C(K) = \E^{\mu}[(X-K)^+]$ uniformly in $K$. Consequently, their second distributional derivatives $C_n''(K)$ converge in the vague topology to $C''(K)$. Together with the tightness of $(\mu_n)$, we gather that  $\mu_n\to \mu$ weakly.  The $\cW_1$ convergence then follows from  \cite[Theorem 6.9]{villani2009optimal}.  

\end{proof}

\begin{figure}
    \centering
    \caption{Discrete measures $\mu_n,\nu_n$, $n=1$,  and dual optimizers of $\cW_1(\mu_n,\nu_n)$ and $\vec{d}(\nu_n,\mu_n)$ ($\varphi$ and $\psi$, respectively) in the proof of \cref{prop:bidAskDistance}~(iv).}
    \includegraphics[width= 4in, height = 2in]{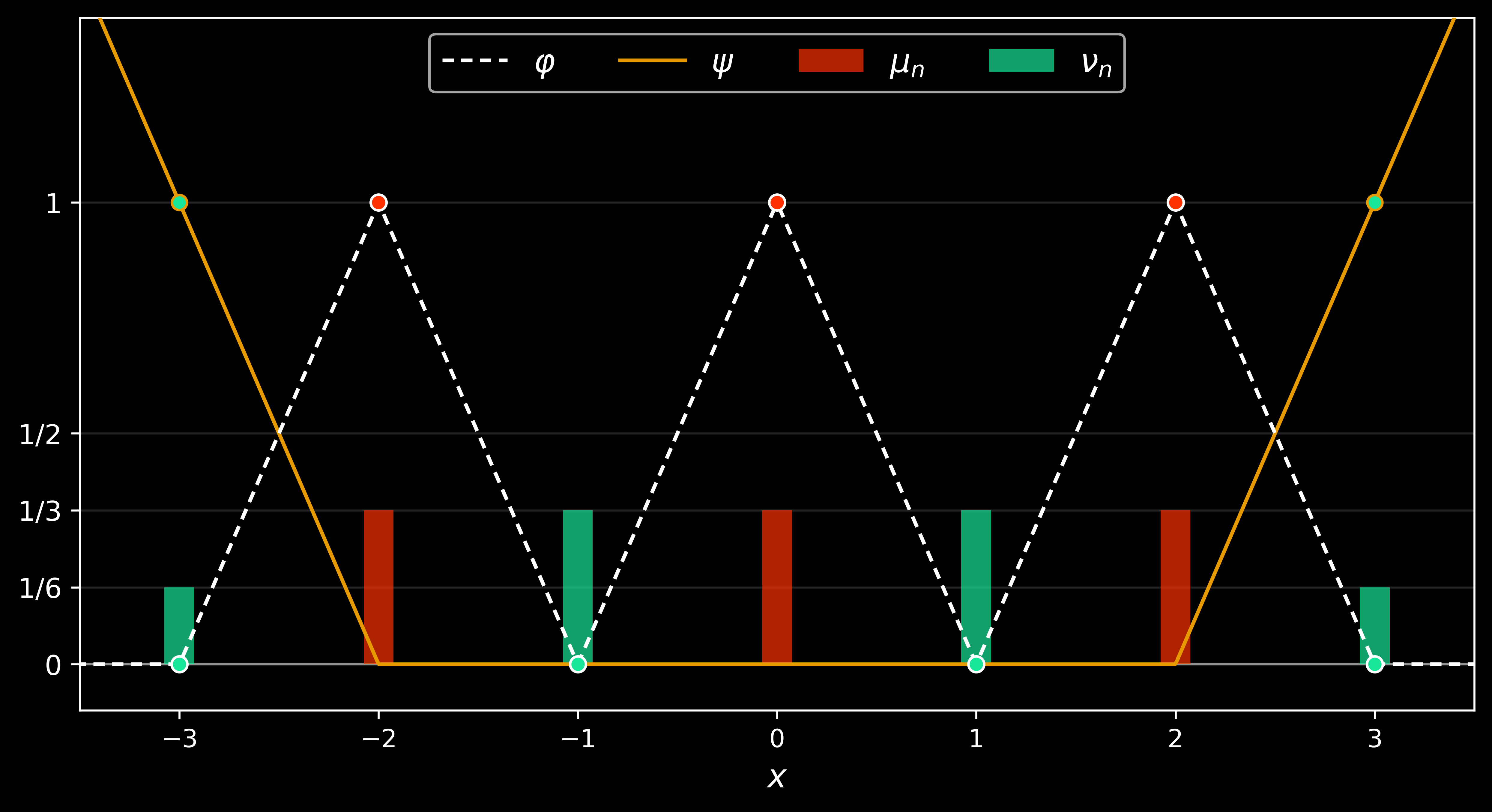}
    
    \label{fig:counterexample2}
\end{figure}

\begin{lemma}\label{lem:CVXEstimate} 
Let $\psi \in \textnormal{CVX}(\R) \cap \cC^{1}(\R)$ with $\psi'$ $M$-Lipschitz for some $M>0$. Then for any $\varphi \in \textnormal{CVX}(\R)$ satisfying
$$\sup_{x\in \R} |(\psi-\varphi)(x)| \le \varepsilon,$$
we have $$\sup_{x\in \R} \sup_{\gamma\in \partial \varphi(x)}|\gamma - \psi'(x)| \le 2\sqrt{M\varepsilon},$$
where $\partial \varphi(x)$ denotes the subdifferential of $\varphi$ at $x$.
\end{lemma}
\begin{proof}
Fix $x$,  let  $\gamma \in \partial \varphi(x)$. 
As $\psi'$ is $M$-Lipschitz, %
\cite[Theorem~9.22]{rockafellar1998variational}  implies that 
\begin{equation}\label{eq:quadEstimate}
     \psi'(x)\delta - \frac{M}{2}\delta^2 \le \psi(x + \delta) - \psi(x) \le  \psi'(x)\delta + \frac{M}{2}\delta^2,\quad \delta>0.
\end{equation}
On the other hand, from $\sup_{x\in \R} |(\psi-\varphi)(x)| \le \varepsilon$ and the convexity of $\varphi$, we derive that 
\[
\psi(x + \delta) + \varepsilon \ge \varphi(x + \delta)  \ge \varphi(x) + \gamma \delta \ge \psi(x) - \varepsilon + \gamma \delta.
\]
Subtracting $\psi(x)$ and applying the upper bound in \eqref{eq:quadEstimate},  we arrive at
\[
\psi'(x)\delta + \frac{M}{2}\delta^2 + \varepsilon \ge \gamma \delta - \varepsilon.
\]
Rearranging the terms gives 
\begin{equation}\label{eq:cvx-upper-bound}
     \gamma - \psi'(x) \le \frac{2\varepsilon}{\delta} + \frac{M}{2}\delta.
\end{equation}
Minimizing the right-hand side of~\eqref{eq:cvx-upper-bound} over $\delta>0$ yields the optimal $\delta^* = 2\sqrt{\varepsilon/M}$, which gives the upper bound 
\[\gamma - \psi'(x) \le 2\sqrt{M\varepsilon}.\] When $\delta<0$, repeating the argument for $\psi(x-\delta) - \psi(x)$ shows the lower bound 
\[
\gamma - \psi'(x) \geq -2\sqrt{M\varepsilon},
\]
as desired.
\end{proof}

\begin{lemma}\label{lem:BV-ibp}
Let $\nu$ be a finite signed measure on $\R_+$ with $\nu(\R_+)=0$ and let $h:\R_+\to\R$ have bounded variation $V(h)<\infty$. Set $h(0^-):=h(0)$ and let
$h^{\rm r}(K):=h(K^+)$ be the right-continuous modification of $h$. Denote by
$Dh^{\rm r}$ the signed measure induced by
$h^{\rm r}$ with $Dh^{\rm r}(\{0\})=0$, and define its nonatomic part by
\[
D^{\rm na}h^{\rm r}
:=
Dh^{\rm r}
-\sum_{K>0}\bigl(h(K^+)-h(K^-)\bigr)\delta_K.
\]
Then, set $G_\nu(K):=\nu([0,K])$ and $ G_\nu(K^-):=\nu([0,K))$, we have
\[
\begin{aligned}
\nu(h)
={}&-\int_{\R_+}G_\nu(K)\,D^{\rm na}h^{\rm r}(\rd K)
-\sum_{K\ge0}G_\nu(K^-)\bigl(h(K)-h(K^-)\bigr)\\
&+\sum_{K\ge0}G_\nu(K)\bigl(h(K)-h(K^+)\bigr).
\end{aligned}
\]
Consequently, it follows that
\[
|\nu(h)|
\le V(h)\,\sup_{K\ge0}|G_\nu(K)|.
\]
\end{lemma}

\begin{proof}
Since $h$ has bounded variation, $h^{\rm r}$ has bounded variation as well and
$\sum_{K>0}|h(K^+)-h(K^-)|\le |Dh^{\rm r}|(\R_+)<\infty$.
Thus, $D^{\rm na}h^{\rm r}$ is well-defined. Applying integration by parts for the Lebesgue--Stieltjes integral gives
\[
\nu(h^{\rm r})
=-\int_{\R_+}G_\nu(K^-)\,Dh^{\rm r}(\rd K).
\]
The boundary terms vanish because $G_\nu(0^-)=0$, $\lim_{K\to\infty}G_\nu(K)=\nu(\R_+)=0$, and $h^{\rm r}$ is bounded.

Noticing that the function $h-h^{\rm r}$ is nonzero over at most countably many points, it holds that
\[
\nu(h-h^{\rm r})
=\sum_{K\ge0}\bigl(h(K)-h(K^+)\bigr)\nu(\{K\}).
\]
Using $\nu(\{K\})=G_\nu(K)-G_\nu(K^-)$ and the decomposition of $Dh^{\rm r}$, we obtain
\begin{equation}\label{eq:ibp-jumps}
\begin{aligned}
\nu(h)
={}&-\int_{\R_+}G_\nu(K^-)\,D^{\rm na}h^{\rm r}(\rd K)\\
&-\sum_{K>0}G_\nu(K^-)\bigl(h(K^+)-h(K^-)\bigr)\\
&+\sum_{K\ge0}\bigl(h(K)-h(K^+)\bigr)
\bigl(G_\nu(K)-G_\nu(K^-)\bigr).
\end{aligned}
\end{equation}
For each $K>0$, combining the jump terms in \eqref{eq:ibp-jumps} gives
$-G_\nu(K^-)(h(K)-h(K^-))+G_\nu(K)(h(K)-h(K^+))$. The same identity holds at $K=0$ since $G_\nu(0^-)=0$ and $h(0^-)=h(0)$. Moreover, $G_\nu(K)\ne G_\nu(K^-)$ only at the atoms of $\nu$, which form an at most countable set. Since $D^{\rm na}h^{\rm r}$ is nonatomic, the integrals of $G_\nu$ and $G_\nu(\cdot^-)$ against $D^{\rm na}h^{\rm r}$ coincide. This proves the first claim.

Finally, noting that $V(h)$ can be expressed as
\[
V(h)
=
|D^{\rm na}h^{\rm r}|(\R_+)
+\sum_{K\ge0}\left(
|h(K)-h(K^-)|+|h(K)-h(K^+)|
\right),
\]
then the claimed estimate now follows from the first claim, the triangle inequality, and the fact that $\sup_{K\ge0}|G_\nu(K^-)| \leq \sup_{K\ge0}|G_\nu(K)|$.
\end{proof}

\section{Discretized Problems} \label{app:LP}

Suppose the spot price of the underlying asset is $x_0$, and options at strikes $\{K_{m}\}_{m\in[M]}$ are available at time $T_i$, $i\in [N]$. By the same argument as in \cite{beiglbock2013model}, we could restrict the dual optimization to 
\[
\Psi^{\rm b,a}_{\cS}(h) =\left\{\psi^{\rm b,a} =(\psi^{\rm a},\psi^{\rm b}) \in (\cS)^{2N}\,\mid\,  \exists\,\, \Delta\,\,{\rm s.t.}\,\, {\rm P\&L}^h_{\psi^{\rm a}-\psi^{\rm b}, \Delta}\geq 0\right\}
\]
where
\[
\cS :=\left\{\psi:\R_+\to \R\,\mid\, \psi(x)=a+ bx + \sum_{m=1}^M c_m(x-K_m)^+,\,a, b \in \R,\, c_m\geq 0\right\}.
\]
Introducing the restricted superhedging problem as
\begin{equation*}\label{eq:dual-finite}
    D_{\rm fin} := \inf_{\psi^{\rm b, a}\in \Psi^{\rm b,a}_{\cS}(h)} \sum_{i=1}^N\mu_i^{\rm a}(\psi_i^{\rm a}) - \mu_i^{\rm b}(\psi_i^{\rm b}).
\end{equation*}
Write $c^{\rm s}(K_{m},T_i)$, $i\in[N]$, $m\in[M]$, for the call option price with strike $K_{m}$ and maturity $T_i$ under the marginal $\mu_i^{\rm s}$, ${\rm s}\in\{{\rm a},{\rm b}\}$. 
  As $\E_{\mu_i^{\rm b}}[X] = \E_{\mu_i^{\rm a}}[X] = x_0$ for all $i\in [N]$, the optimization becomes
\begin{align*}
   &\min a + \sum_{i=1}^N b_ix_0 +  \sum_{i=1}^N \sum_{m=1}^M c_{i,m}^{\rm a} c^{\rm a}(K_{m}, T_i) - \sum_{i=1}^N \sum_{m=1}^M c_{i,m}^{\rm b} c^{\rm b}(K_{m}, T_i)\\
   {\rm s.t.}\quad &
a  + \sum_{i=1}^N b_ix_i +\sum_{i=1}^N \sum_{m=1}^M \left(c_{i,m}^{\rm a} - c_{i,m}^{\rm b} \right)(x_i - K_{m})^+\\
&+ \sum_{i=1}^{N-1} \Delta(x_1,\ldots, x_i)(x_{i+1} -x_i) \geq h(x_1,\ldots, x_N),\quad  (x_1,\ldots, x_N) \in (\R_+)^N,\\
&a, b_i \in \R, c_{i,m}^{\rm a}, c_{i,m}^{\rm b} \geq 0.
\end{align*}
For numerical implementation, the state space is further discretized on a finite grid, yielding a finite-dimensional linear program. For completeness, we present the discretized version of the primal problem~\eqref{eq:primal1M} in the single-maturity case $N=1$, which is used to find a primal optimizer in \cref{subsec:digtial_spx}. To that end, we construct a finite grid $\{y_i\}_{i\in[I]}$ on the support of $\mu^{\rm a}$. Given the precomputed call option prices $\{c^{\rm s}(K_m)\}_{m\in[M],\,{\rm s}\in\{{\rm a},{\rm b}\}}$, the discretized problem reads:
\begin{align*}
    &\max_{p} \sum_{i=1}^I p_i\, h(y_i)\\
    \text{s.t.}\quad 
    & \sum_{i=1}^I p_i = 1, \qquad p_i \geq 0,\quad \sum_{i=1}^I p_i y_i = x_0,\\
    & \sum_{i=1}^I p_i (y_i - K_m)^+ - c^{\rm a}(K_m) \leq 0,\\
    & \sum_{i=1}^I p_i (y_i - K_m)^+ - c^{\rm b}(K_m) \geq 0, \qquad m \in [M].
\end{align*}

\linespread{1}
\bibliography{ref}

\end{document}